\pdfoutput=1
\documentclass[paper=a4]{scrartcl}

\usepackage[utf8]{inputenc}
\usepackage[T1]{fontenc}
\usepackage{algorithm,algorithmic}

\renewcommand{\algorithmiccomment}[1]{\hfill{\itshape #1}}
\usepackage{tikz}
\usetikzlibrary{shapes.arrows, positioning}
\usepackage{amsmath, amsthm, amssymb, mathtools}
\usepackage{lmodern}
\usepackage[strings]{underscore}
\usepackage{xspace}
\usepackage[shortlabels]{enumitem}
\usepackage{rotating}
\usepackage{hyperref}
\usepackage[version=4]{mhchem}
\usepackage{placeins}

\newcommand{\NN}{\mathbb{N}}
\newcommand{\QQ}{\mathbb{Q}}
\newcommand{\RR}{\mathbb{R}}
\newcommand{\ZZ}{\mathbb{Z}}
\newcommand{\mN}{\mathcal{N}}
\newcommand{\mM}{\mathcal{M}}
\newcommand{\mU}{\mathcal{U}}
\newcommand{\mP}{\mathcal{P}}
\newcommand{\ha}{\mathrel{\triangleright}}
\newcommand{\notha}{\mathrel{\not\triangleright}}
\newcommand{\wf}{\widehat f}
\newcommand{\wg}{\widehat g}

\newcommand{\wdelta}{\bar\delta}
\newcommand{\Rho}{\mathrm{P}}

\newcommand{\myd}[1]{\operatorname{d}\!{#1}}
\newcommand{\diff}[2]{\frac{\myd{#1}}{\myd{#2}}}
\newcommand{\pmat}[1]{\begin{pmatrix}#1\end{pmatrix}}
\newcommand{\pmatr}[1]{\begin{pmatrix*}[r]#1\end{pmatrix*}}
\newcommand{\pmatl}[1]{\begin{pmatrix*}[l]#1\end{pmatrix*}}

\newcommand{\spmatl}[1]{\left(\begin{smallmatrix*}[l]#1\end{smallmatrix*}\right)}
\newcommand{\diag}[2]{\left(\begin{smallmatrix}#1 & & \\ & \ddots & \\ & & #2\end{smallmatrix}\right)}

\newcommand\numberthis{\addtocounter{equation}{1}\tag{\theequation}}
\renewcommand{\emptyset}{\varnothing}
\newcommand{\biomd}[1]{BioModel #1}

\DeclareMathOperator{\round}{round} 
 
\DeclareMathOperator{\sgn}{sgn} 

\newtheorem{theorem}{Theorem}
\newtheorem{proposition}[theorem]{Proposition}
\newtheorem{lemma}[theorem]{Lemma}
\newtheorem{corollary}[theorem]{Corollary}

\allowdisplaybreaks[1]

\KOMAoptions{
  abstract=true,
  fontsize=10pt,
  DIV=12,
  headinclude=false,
  footinclude=false,
  captions=tableheading
}

\addtokomafont{author}{\large}

\recalctypearea

\begin{document}

\title{Algorithmic Reduction of Biological Networks With Multiple Time Scales}

\author{
  Niclas Kruff, RWTH Aachen University, Germany\\
  \texttt{niclas.kruff@matha.rwth-aachen.de}
  \and
  Christoph Lüders, University of Bonn, Germany\\
  \texttt{chris@cfos.de}
  \and
  Ovidiu Radulescu,
  University of Montpellier and CNRS UMR5235 LPHI, France\\
  \texttt{ovidiu.radulescu@umontpellier.fr}
  \and
  Thomas Sturm,
  CNRS, Inria, and the University of Lorraine, France\\
  MPI Informatics and Saarland University, Germany\\
  \texttt{thomas.sturm@loria.fr}
  \and
  Sebastian Walcher, RWTH Aachen University, Germany\\
  \texttt{walcher@matha.rwth-aachen.de}
}

\date{March 2021}

\maketitle

\begin{abstract}
  We present a symbolic algorithmic approach that allows to compute invariant
  manifolds and corresponding reduced systems for differential equations
  modeling biological networks which comprise chemical reaction networks for
  cellular biochemistry, and compartmental models for pharmacology, epidemiology
  and ecology. Multiple time scales of a given network are obtained by scaling,
  based on tropical geometry. Our reduction is mathematically justified within a
  singular perturbation setting. The existence of invariant manifolds is subject
  to hyperbolicity conditions, for which we propose an algorithmic test based on
  Hurwitz criteria. We finally obtain a sequence of nested invariant manifolds
  and respective reduced systems on those manifolds. Our theoretical results are
  generally accompanied by rigorous algorithmic descriptions suitable for direct
  implementation based on existing off-the-shelf software systems, specifically
  symbolic computation libraries and Satisfiability Modulo Theories solvers. We
  present computational examples taken from the well-known BioModels database
  using our own prototypical implementations.
\end{abstract}

\section{Introduction}

Biological network models describing elements in interaction are used in many
areas of biology and medicine. Chemical reaction networks are used as models of
cellular biochemistry, including gene regulatory networks, metabolic networks,
and signaling networks. In epidemiology and ecology, compartmental models can be
described as networks of interactions between compartments. Both in chemical
reaction networks and in compartmental models, the probability that two elements
interact is assumed proportional to their abundances. This property, called mass
action law in biochemistry, leads to polynomial differential equations in the
kinetics.

For differential equations that describe the development of such networks over
time, a crucial question is concerned with reduction of dimension. We illustrate
such a reduction and the steps involved for the classic Michaelis--Menten system,
an archetype of enzymatic reactions. This system is described by the chemical
reactions
\[
  \ce{E + S <=>[k_1][k_{-1}] ES ->[k_2] E + P},
\] 
where \ce{E}, \ce{S}, \ce{ES}, and \ce{P} are the enzyme, substrate,
enzyme-substrate complex, and product, respectively. The mechanism has two
conserved quantities $\ce{[E]} + \ce{[ES]} = c_1$ and
$\ce{[S]} + \ce{[ES]} + \ce{[P]} = c_2$, where $c_1$ and $c_2$ are constant
functions representing the respective total concentrations. Let us choose the
concentration units such that $c_2=1$, and furthermore rename
$c_1 = \varepsilon$. In the ordinary differential equations describing the kinetics of
\ce{[E]}, \ce{[S]}, \ce{[ES]}, and \ce{[P]} according to
\cite[Sect.~2.1.2]{Feinberg:19a} we eliminate the variable
$\ce{[E]}= \varepsilon - \ce{[ES]}$. This yields a reduced system made of variables
$y_1 = \ce{[S]}$ and $y_2 = \ce{[ES]}$ as follows:
\begin{align*}
\dot{y_1} &=-\varepsilon k_1y_1+(k_1y_1+k_{-1})y_2\\
\dot{y_2} &=\varepsilon k_1y_1-(k_1y_1+k_{-1}+k_2)y_2.
\end{align*}
Notice that $\dot{\ce{[P]}}$ is already determined by the algebraic conservation
constraint $\ce{[P]} = 1 - y_1 - y_2$. The parameter $\varepsilon$ represents the ratio of
total concentrations $c_1$ to $c_2$. The general idea is that $\varepsilon$ is small.

In a first step toward reduction, a scaling transformation $y_1=x_1$ and
$y_2=\varepsilon x_2$ yields
\begin{align*}
\dot{x_1} &= \varepsilon (-k_1x_1+(k_1x_1+k_{-1})x_2)  \\
\dot{x_2} &= k_1x_1-(k_1x_1+k_{-1}+k_2)x_2.
\end{align*}
In a second step, one uses singular perturbation theory to obtain the famous
Michaelis--Menten equation. It consists of two components: First, we obtain a one
dimensional invariant manifold given approximately by the quasi-steady state
condition $k_1x_1-(k_1x_1+k_{-1}+k_2)x_2=0$. This considers the fast variable
$x_2$ to be at the steady state and lowers dimension from two to one. Second, we
obtain a reduced system for the slow variable:
\[
\dot{x_1} = -\varepsilon\frac{k_1k_2x_1}{k_1x_1+k_{-1}+k_2}.
\]

With our example, we paraphrased the approach in a seminal paper by Heineken et
al.~\cite{hta}, which was the first one to rigorously discuss quasi-steady state
from the perspective of singular perturbation theory. Realistic network models
may have many species and differential equations. Considerable effort has been
put into model order reduction, i.e., finding approximate models with a smaller
number of species and equations, where the reduced model can be more easily
analyzed than the full model \cite{rgzn}.

The scaling of parameters and variables by a small parameter $\varepsilon$ and the study
of the limit $\varepsilon \to 0$ is central in singular perturbation theory. It is rather
obvious that arbitrary scaling transformations are unlikely to provide useful
information about a given system. Successful scalings, in contrast, are
typically related to the existence of nontrivial invariant manifolds.
Applications of scaling rely on the observation that, loosely speaking, any
result that holds asymptotically for $\varepsilon \to 0$ remains valid for sufficiently
small positive $\varepsilon_*$, provided some technical conditions are satisfied. To
determine scalings of polynomial or rational vector fields that model biological
networks, tropical equilibration methods were introduced and developed in a
series of papers by Noel et al.~\cite{noelgvr}, Radulescu et al.~\cite{rvg},
Samal et al.~\cite{sgfwr,sgfr}, and others. These methods open a feasible path
for biological networks of high dimension. For a given system they provide a
list of possible slow-fast systems, which may or may not yield invariant
manifolds and reduced equations. Other methods due to Goeke et al.~\cite{gwz},
and recently extended to multiple time scales by Kruff and Walcher~\cite{krwa},
determine critical parameter values and manifolds for singular perturbation
reductions.

The principal purpose of the present paper is to complement scaling with an
algorithmic test for the existence of invariant manifolds and the computation of
those manifolds along with corresponding reduced systems of differential
equations. In the asymptotic limit, methods from singular perturbation theory,
principally developed by Tikhonov \cite{tikh} and Fenichel \cite{fenichel}, are
available. A recent extension to multiscale systems by Cardin and Teixeira
\cite{cartex} turns out to be a valuable tool for the systematic computation of
reductions with nested invariant manifolds and allows an algorithmic approach.

In the language of dynamical systems, the behavior of systems with multiple time
scales can be described as follows. All variables evolve towards the steady
state of the ordinary differential equation that they follow, but not with the
same speed, in other words, not within the same \emph{time scale}. \emph{Slow
  variables} correspond to long time scales, and \emph{fast variables}
correspond to short time scales. The steady state of a subset of the ordinary
differential equations is called \emph{quasi-steady state}, and the evolution of
a variable or group of variables towards its quasi-steady state is called
\emph{relaxation}. At a given time scale, a group of variables relaxes towards
one of their quasi-steady states. The set of all quasi-steady states of
variables relaxing within the same time scale forms a \emph{critical manifold},
which provides the lowest order approximation to the corresponding
\emph{invariant manifold}. In this article, we do not distinguish between
critical and invariant manifolds, since we do not discuss higher order
approximations. All slower variables can be considered fixed, and all faster
variables have already relaxed and satisfy quasi-steady state conditions. As the
number of relaxed variables increases and thus the set of quasi-steady state
conditions grows, the respective invariant manifolds get nested so that later
manifolds are contained in earlier ones. Local linear approximations of these
manifolds were proposed by Valorani and Paolucci \cite{valorani2009g} using
numerical methods based on the local Jacobian. However, to the best of our
knowledge, constructive approaches providing nonlinear descriptions of these
manifolds and reduced models are still missing.

From a computer science point of view, we propose a novel symbolic
computation-based algorithmic workflow for the reduction process outlined above.
This includes in particular the automatic verification of certain hyperbolicity
conditions required for the validity of the reductions. We restrict ourselves to
the case of polynomial differential equations that covers mass action chemical
reaction networks and compartmental models. We present a series of algorithms
that takes as input a system of polynomial autonomous ordinary differential
equations together with numerical information related to the desired coarse
graining of the scaling. As output one finally obtains a collection of nested
invariant manifolds for the input system, associated with smaller dimensional
systems that govern the dynamics on those manifolds. This output establishes the
reduced systems discussed above.

The computationally hard parts of our methods are reduced to decision problems
in interpreted first-order logic over various theories. It turns out that
quantifier alternation can be entirely avoided, so that the Satisfiability
Modulo Theories (SMT) framework by Nieuwenhuis et al.
\cite{NieuwenhuisOliveras:06b} can be applied. Several corresponding SMT solvers
are freely available and professionally supported \cite{cvc4,
  mathsat,smtrat,z3}. It is remarkable that we arrive with our comprehensive
algorithmic work here at SMT sub-problems for several different logics,
viz.~linear integer arithmetic, linear real arithmetic, and non-linear real
arithmetic. The algorithms presented here are suitable for straightforward
implementation provided that a symbolic computation library, or computer algebra
system, and an SMT solver are available. We created two independent prototypical
realizations in software ourselves, one in Python using freely available
libraries, and one in Maple.

The plan of the paper is as follows: In Sect.~\ref{se:scale} we introduce an
abstract scaling procedure, which assumes, for given $0 < \varepsilon_* < 1$, the
existence of families of exponents $c_{k,J}$ and $d_k$ for scaling polynomial
coefficients and variables, respectively. From the scaled system, higher order
terms are truncated, and the obtained system is partitioned into several time
scales, ordered from fastest to slowest. A corresponding generic algorithm uses
black-box functions $c$ and $d$. In Sect.~\ref{se:tropscale}, we make precise
one possible way to realize $c$ and $d$, based on tropical geometry. So far, our
transformations are mostly of formal nature. On these grounds, we
algorithmically determine in Sect.~\ref{se:spt} invariant manifolds and
corresponding reduced systems, which makes the formal scaling meaningful in a
mathematically precise way. In general, this is possible only for a certain
number $\ell$ of time scales, where $\ell$ is explicitly found and---in contrast to
existing alternative approaches---often larger than $2$. Technically, we apply
recent results by Cardin and Teixeira \cite{cartex} based on Fenichel theory. In
Sect.~\ref{se:simpl}, we employ symbolic computation techniques, specifically
Gröbner basis theory, to equivalently simplify our reduced systems, which are
still scaled in terms of $\varepsilon_*$, $c$, and $d$. In Sect.~\ref{se:tb}, we finally
transform back to the principal scale of the original system while preserving
the obtained multiple time scales and the structure of the corresponding reduced
systems. In particular, the various time scale factors remain explicit. In
Sect.~\ref{se:picture}, we summarize what we have gained from the overall
procedure for our original input. At this point, the mathematical development of
our framework has been accompanied by nine algorithms, and we give a tenth
top-level algorithm, which makes precise how various modules are combined and
interact with one another. In Sect.~\ref{se:compex} we discuss computational
examples with our prototypical software mentioned above. We consider models from
the BioModels database, a repository of mathematical models of biological
processes \cite{le2006biomodels}. The focus is on successful reductions for
biologically interesting examples. This is counterbalanced in Appendix
\ref{app:examples} by further examples to support the understanding of our
algorithms. This also provides some examples where we do not obtain meaningful
reductions. In Sect.~\ref{se:complexity} we highlight some computational steps
in our algorithms from the point of view of asymptotic worst-case complexity. In
Sect.~\ref{se:conclusion}, we wrap up and point at possible future research
directions.

\section{Scaling of Polynomial Vector Fields}

In what follows, we adopt a rather general scaling formalism that has been used
recently in \cite{noel2012tropical,noelgvr,rgzn,rvg,sgfwr} and is recurrent in
the literature on singular perturbations, see for instance
\cite[Sect.~3]{nipp1988algorithmic}. We use the convention that the natural
numbers $\NN$ include $0$.

\subsection{An Abstract Scaling Procedure}\label{se:scale}
Our starting point is a parameter dependent system $S$ of polynomial
differential equations
\begin{equation}
\label{origeq}
\dot{y_k} := \diff{y_k}{t} = \sum_J \gamma_{k,J}y^J,\quad 1\leq k\leq n,
\end{equation}
where the summation ranges over multi-indices $J=(j_1,\ldots,j_n)\in \NN^n$,
$\gamma_{k,J} \in \RR$, and only finitely many $\gamma_{k,J}$ are non-zero. We abbreviate
$y^J=y_1^{j_1}\cdots y_n^{j_n}$, as usual. In terms of network models, $y_k$
represents the concentration of either a chemical species or a type of
individual in a compartment. Note that we use positive integers as indices,
instead of concrete names for species and compartments. The real coefficients
$\gamma_{k,J}$ describe actions of other species or individuals on the species or
individual $k$. If these actions are activations one has $\gamma_{k,J}>0$, whereas
for repressions one has $\gamma_{k,J}<0$. Several species may interact to produce an
action on a given species $k$. This information is contained in the number of
non-zero components of $J$, which is called the order of the action. This
terminology is inspired from chemical reactions, where the order represents
essentially the number of reactant species.

Throughout this paper, we require that positive $y_k$ remain positive as time
progresses. In other words, the positive first orthant
$\mU = (0, \infty)^n \subseteq \RR^n$ is positively invariant for system \eqref{origeq}, which
is the case, e.g., in chemical reaction networks when $\gamma_{k,J}y^J \geq 0$ on all
intersections of hyperplanes $\{\,(y_1, \dots, y_n) \in \RR^n \mid y_k = 0\,\}$ with
$\overline\mU$.

We fix some small $\varepsilon_* \in (0,1)$, and we impose that
\begin{equation}
\label{coeffscale}
\gamma_{k,J}=\varepsilon_*^{c_{k,J}}\bar{\gamma}_{k,J},
\end{equation}
with rational numbers $c_{k,J}$. The tacit understanding is that only nonzero
$ \gamma_{k,J}$ are being considered. The intuitive idea, which will be made more
precise in Sect.~\ref{se:tropscale}, is that the $\bar{\gamma}_{k,J}$ are close to
one. Moreover, we introduce a positive parameter $\varepsilon$ and consider the system
\begin{equation}
\label{paramscaleeq}
\dot{y_k} = \sum_J \varepsilon^{c_{k,J}}\,\bar{\gamma}_{k,J}\,y^J ,\quad 1\leq k\leq n,
\end{equation}
with $\varepsilon$-dependent coefficients. Notice that \eqref{paramscaleeq} matches
\eqref{origeq} at $\varepsilon=\varepsilon_*$. By renormalizing
$y_k=\varepsilon^{d_k}x_k$, $d_k\in\mathbb Q$, one obtains a system in scaled variables
\begin{equation}
\label{fullscaleeq}
\dot{x_k}=\sum_J \varepsilon^{c_{k,J}+\left<D,J\right>-d_k}\,\bar{\gamma}_{k,J}\,x^J,\quad 1\leq k\leq n,
\end{equation}
with $D=(d_1,\ldots,d_n)$ and the dot product in $\RR^n$ denoted by
$\left<\cdot,\cdot\right>$. This transformation preserves the positive invariance of
$\mU$. The scaling comes with the implicit assumption that for $i$,
$j \in \{1, \dots, n\}$, the relative order of $y_i$ with respect to $y_j$ is bounded
by $y_i/y_j = \Theta(\varepsilon^{d_i-d_j})$ for $\varepsilon \to 0$, so that all
$x_k$ get the same order of magnitude. Continuing, we set
$\nu_k = \min \{\, c_{k,J}+\left<D,J\right>-d_k \mid \bar{\gamma}_{k,J}\not=0 \,\}$ to
obtain
\begin{equation}
\label{fullscaleeqnorm}
\dot{x_k}=\varepsilon^{\nu_k}\sum_J \varepsilon^{c_{k,J}+\left<D,J\right>-d_k-\nu_k}\,\bar{\gamma}_{k,J}\,x^J,\quad 1\leq k\leq n,
\end{equation}
where now all exponents of $\varepsilon$ inside the sums are nonnegative. Finally one may
perform a preliminary time scaling $\tau=\varepsilon^\mu t$,
$\mu=\min \, \{ \nu_1, \ldots, \nu_n \}$ to arrive at
\begin{equation}
\label{fullscaleeqnormplus}
x_k' := \diff{x_k}{\tau} = \varepsilon^{\nu_k-\mu}\sum_J \varepsilon^{c_{k,J}+\left<D,J\right>-d_k-\nu_k}\,\bar{\gamma}_{k,J}\,x^J,\quad 1\leq k\leq n,
\end{equation}
with all exponents nonnegative. We are interested in system
\eqref{fullscaleeqnormplus} for variable $\varepsilon>0$, in the asymptotic limit $\varepsilon \to 0$.

We restructure \eqref{fullscaleeqnormplus} by collecting all variables with
equal $\nu_i-\mu$ in vectors $z_1$, \dots,~$z_m$, where $z_k \in \RR^{n_k}$ for
$k \in \{1, \dots, m\}$, in ascending order of exponents and such that
$n_1 + \dots + n_m = n$. We obtain a system of the form
\begin{equation}
\label{fullscalesimplified}
z_k'= \varepsilon^{a_k}\widetilde f_k(z,\varepsilon)
= \varepsilon^{a_k}\left(\widetilde f_k(z,0)+\varepsilon^{a_{k,2}'}p_{k,2}+\cdots+\varepsilon^{a_{k,w_{k}}'}p_{k,w_{k}}\right)
= \varepsilon^{a_k}\left(\widetilde f_k(z,0)+o(1)\right), \quad 1\leq k\leq m,
\end{equation}
where $a_k$, $a_{k,j}'\in \QQ$, $0=a_1<a_2<\ldots < a_m$, $0 < a_{k,j}'$, and
$p_{k,j}$ are multivariate polynomials in $z$ for $1\leq k\leq m$ and
$2 \leq j \leq w_k$. Note that the case $m=1$ is not excluded. By substituting
$\delta := \varepsilon^{1/q}$ with a sufficiently large positive integer $q$, one ensures that
only nonnegative integer powers of $\delta$ appear:
\begin{equation}
\label{fullscalesimplifiedint}
z_k'= \delta^{b_k}\wf_k(z,\delta)
= \delta^{b_k}\left(\wf_k(z,0)+\delta^{b_{k,2}'}p_{k,2}+\cdots+\delta^{b_{k,w_{k}}'}p_{k,w_{k}}\right)
= \delta^{b_k}\left(\wf_k(z,0)+o(1)\right),
\quad 1\leq k\leq m,
\end{equation}
where $b_k$, $b_{k,j}'\in \NN$, $0=b_1<b_2<\ldots < b_m$, $0 < b_{k,j}'$ for
$1\leq k\leq m$ and $2 \leq j \leq w_k$.

Our idea is that the indices $k$ correspond to different time scales
$\delta^{b_k}\tau$. For $m>1$, system \eqref{fullscalesimplifiedint}, as
$\delta\to0$, may be thought of as separating fast variables from increasingly slow
ones. It will turn out in Sect.~\ref{se:spt} that the exact number of time
scales finally obtained by our overall approach can actually be smaller than
$m$.

Given certain conditions, which will be made explicit in Theorem~\ref{th:ct} and
with its application in Sect.~\ref{se:extfen}, we may formally truncate the
right hand sides of \eqref{fullscalesimplifiedint} and keep only terms of lowest
order in $\delta$:
\begin{equation}
\label{eq:scaleandtruncate}
z_k'= \delta^{b_k}\wf_k(z,0), \quad 1\leq k \leq m.
\end{equation}
In the following, we refer to the transformation process from \eqref{origeq} to
\eqref{fullscalesimplifiedint} as \emph{scaling}. Strictly speaking, this
comprises scaling in combination with \emph{partitioning}. We refer to the step
from \eqref{fullscalesimplifiedint} to \eqref{eq:scaleandtruncate} as
\emph{truncating}.

Algorithm~\ref{alg:scale} 
\begin{algorithm}
  \renewcommand{\algorithmiccomment}[1]{%
    \hfill\parbox{\widthof{$\in \overline\QQ[x_1, \dots, x_n][\delta]$}}{\textcolor{gray}{#1}}}
  \caption{$\operatorname{ScaleAndTruncate}$\label{alg:scale}}  
  \begin{algorithmic}[1]
    \REQUIRE
    \begin{enumerate}[noitemsep]
    \item
    A list $S = \big[ \diff{y_1}{t} = f_1, \dots, \diff{y_n}{t} = f_n \big]$ of
    autonomous first-order ordinary differential equations where $f_1$, \dots,
    $f_n \in \QQ[y_1, \dots, y_n]$;
    \item $c: \{1, \dots, n\} \times \{1, \dots, n\}^n \to \QQ$;
    \item $d: () \to \QQ^n \cup \{\bot\}$;
    \item $\varepsilon_* \in (0,1) \cap \QQ$
    \end{enumerate}
    \smallskip
    
    \ENSURE 
    \begin{enumerate}
    \item A list $[T_1, \dots, T_m]$ where, abbreviating $\diff{}{\tau}$ by a prime,
      $T_k = (z_k' = \delta^{b_k} f_k)$ with
      $z_k' \subseteq [x_1', \dots, x_n']$,
      $\bigcup_k z_k' = [x_1', \dots, x_n']$, $z_1'$, \dots,~$z_m'$ pairwise disjoint,
      $b_1 < \dots < b_m \in \NN$, and $f_k \subseteq \QQ[x_1, \dots, x_n]$, or the empty list;
    \item A list $[P_1, \dots, P_m]$ of lists with $P_k \subseteq \QQ[x_1, \dots, x_n][\delta]$ and
      $|P_k| = |T_k|$ for $k \in \{1, \dots, m\}$; 
    \item A substitution $\sigma$ for $x_1$, \dots, $x_n$, $\tau$, $\delta$, and $\varepsilon$
    \end{enumerate}
    \smallskip
    
    The first output $[T_1, \dots, T_m]$ contains differential equations
    $z_k' = \delta^{b_k} \wf_k(z, 0)$ for $k \in \{1, \dots, m\}$ in terms of
    system~\eqref{fullscalesimplifiedint}. The second output $[P_1, \dots, P_m]$
    contains the higher order terms in \eqref{fullscalesimplifiedint} as
    polynomials
    $p_k = \delta^{b_k+b_{k,2}'}p_{k,2}+\cdots+\delta^{b_k+b_{k,w_{k}}'}p_{k,w_{k}}$. The last
    output is a substitution that undoes all substitutions applied for obtaining
    \eqref{fullscalesimplifiedint} from \eqref{origeq}.\smallskip
    
    This gives the following invariant: Denote
    $\widetilde S = \bigl(\bigcup_{k=1}^m T_k \oplus P_k\bigr)\sigma$, where
    $(x' = g) \oplus p$ stands for $x' = g + p$ and is applied elementwise. Then
    $\widetilde S$ is equal to $S$ up to multiplication of the differential
    equation $\dot{y_i} = \sum_J \gamma_{i,J}y^J$ in $S$ with a positive scalar factor
    $1/\varepsilon_*^{\mu+d_i}$.\smallskip

    For $q \in \QQ[x_1, \dots, x_n](\delta)$ we use $\deg_\delta(q)$ for the univariate degree of
    $q$ in $\delta$. Similarly, $\operatorname{tmon}_\delta(q)$ is the trailing monomial
    in $\delta$.
    \smallskip

    \IF {\label{sat:dtest}$d() = \bot$}
    \RETURN {$[\, ]$, $[\, ]$, $[\, ]$}
    \ENDIF
    \STATE{$\mu := \infty$}
    \STATE{$q := 1$}
    \STATE{$(d_1, \dots, d_n) := d()$
      \COMMENT{$\in \QQ^n$}}
    \FOR{$k := 1$ \TO $n$}
    \STATE{$h_k := 0$}
    \FORALL{monomials $\gamma y^J$ \textbf{in} $f_k$}
    \STATE{$\bar\gamma := \gamma / \varepsilon_*^{c(k, J)}$
    \COMMENT{$\in \overline\QQ$}}
    \STATE{$\eta := c(k, J) + \langle (d_1, \dots, d_n), J \rangle - d_k$
      \COMMENT{$\in \QQ$}}
    \STATE{$\mu := \min(\mu, \eta)$
      \COMMENT{$\in \QQ$}}
    \STATE{$q := \operatorname{lcm}(q, \operatorname{denom} \eta)$
      \COMMENT{$\in \NN \setminus \{0\}$}}
    \STATE{$h_k := h_k + \varepsilon^{\eta} \bar\gamma x^J$}
    \ENDFOR
    \ENDFOR
    \FOR{$k:=1$ \TO $n$}
    \STATE{$h_k := h_k / \varepsilon^{\mu}$}
    \STATE{$h_k := h_k[\varepsilon \gets \delta^q]$
      \COMMENT{$\in \overline\QQ[x_1, \dots, x_n][\delta]$}}
    \STATE{$g_k := \operatorname{tmon}_\delta h_k$}
    \STATE{$p_k := h_k - g_k$}
    \ENDFOR
    \STATE{$L := \big[\diff{x_1}{\tau} = g_1, \dots, \diff{x_n}{\tau} = g_n\big]$}
    \STATE{$[b_1, \dots, b_m] := \operatorname{sort}(\deg_\delta g_1, \dots, \deg_\delta g_n)$,
      ascending and removing duplicates}
    \FOR{$k:=1$ \TO $m$}
    \STATE{$T_k := [\,\diff{x}{\tau} = g \in L \mid \deg_\delta g = b_k\,]$}
    \STATE{$P_k := [\,p_j \in \{p_1, \dots, p_n\} \mid \deg_\delta g_j = b_k\,]$}
    \ENDFOR
    \STATE{$\sigma := [x_1 \gets y_1 / \varepsilon^{d_1}, \dots, x_n \gets y_n / \varepsilon^{d_n}] \circ
      [\tau \gets \varepsilon^\mu t] \circ [\delta \gets \varepsilon^{1/q}] \circ [\varepsilon \gets \varepsilon_*]$}
    \RETURN{$[T_1, \dots, T_m]$, $[P_1, \dots, P_m]$, $\sigma$}
  \end{algorithmic}
\end{algorithm}
reflects our discussions so far. It takes as input a list $S$ of differential
equations representing system~\eqref{origeq} and a choice of
$0 < \varepsilon_* < 1$ for \eqref{coeffscale}. For our practical purposes, the polynomial
coefficients in $S$ as well as $\varepsilon_*$ are taken from $\QQ$. Our algorithm is
furthermore parameterized with a function $c$ mapping suitable indices to
rational numbers and a constant function $d$ yielding either a tuple
$D=(d_1, \dots, d_n)$ of rational numbers or $\bot$. The black-box functions $c$ and
$d$ reflect the mathematical assumptions around \eqref{coeffscale} and
\eqref{fullscaleeq} that suitable $c_{k,J}$ and $d_k$ exist, respectively.
Suitable instantiations for the parameters $c$ and $d$ can be realized, e.g.,
using tropical geometry, which will be the topic of the next section. It will
turn out that instantiations of $d$ can fail on the given combination of $S$ and
$\varepsilon_*$, which is signaled by the return value $\bot$ of $d$, and checked right away
in l.\ref{sat:dtest} of Algorithm~\ref{alg:scale}.

\subsection{Scaling via Tropical Geometry}\label{se:tropscale}

So far, the above transformations leading to \eqref{fullscaleeq} are a formal
exercise. No particular strategy was applied for choosing
$\varepsilon_* \in(0,1)$. Early model reduction studies used dimensional analysis to obtain
$\varepsilon_*$ as a power product in model parameters \cite{hta,segel1989quasi}.

Here we discuss a different approach, based on tropical geometry, also called
max-plus or idempotent algebra. This is a relatively recent field of mathematics
that draws its origins from fields as diverse as algebraic geometry,
optimization, and physics \cite{litvinov2009tropical}. In all these fields,
tropical geometry appears as a technique to simplify non-linear objects.
Polynomials are replaced by piecewise-linear functions, and geometrical problems
are transformed into combinatorial problems \cite{mikhalkin2005enumerative}.

Tropical geometry is natural for any computation with orders of magnitude. In
physics, it occurs as Litvinov--Maslov dequantization of real numbers leading to
degeneration of complex algebraic varieties into tropical varieties
\cite{litvinov2007maslov,viro2001dequantization}. The name dequantization
originates from the formal analogy between this limiting process and
Schrödinger's dequantization that turns quantum physics into classical physics
when the Planck constant is considered a small parameter $\varepsilon$ that tends to zero.
Closer, in the physical world, to Litvinov--Maslov dequantization are the
vanishing viscosity phenomena, well known in problems of wave propagation. In
mathematics, tropical varieties and prevarieties establish a modern tool in the
theory of Puiseux series \cite{bogart2007computing}.

In contrast to the dimensional analysis approach mentioned above, the value
$\varepsilon_*$ is now not dictated by physico-chemistry. Instead, it is freely chosen to
provide ``power'' parametric descriptions of all the quantities occurring in the
differential equations (parameters, monomials, time scales), in a similar way to
describing curves by continuously varying real parameters.

Next, we explain how to obtain the orders $c_{k,J}$ and $D$ introduced in the
previous section with \eqref{coeffscale} and \eqref{fullscaleeq}, respectively.
The orders $c_{k,J}$ are computed from $\varepsilon_* \in (0,1)$ and $\gamma_{k,J}$ as
\begin{equation}\label{rationalindex}
c_{k,J} = \frac{\round (p \log_{\varepsilon_*} |\gamma_{k,J}|)}{p}.
\end{equation}
The function $\round: \RR \to \ZZ$ rounds to nearest, ties to even, in the sense of
IEEE~754 \cite{IEEE754}. The positive integer $p$ controls the precision of the
rounding step. Using
$\bar \gamma_{k,J} = \gamma_{k,J} / \varepsilon_*^{c_{k,J}}$ as defined in \eqref{coeffscale}, our
definition satisfies the constraint
$\varepsilon_*^{1/(2p)} \leq | \bar \gamma_{k,J} | \leq \varepsilon_*^{-1/(2p)}$. The orders
$D=(d_1, \dots, d_n)$ satisfy certain constraints as well. Those constraints result
heuristically from the idea of compensation of dominant monomials
\cite{noel2012tropical}. Slow dynamics is possible if for each dominant,
i.e.,~much larger than the other, monomial on the right hand side of
\eqref{fullscaleeqnormplus}, there is at least one other monomial of the same
order but with opposite sign. This condition, named tropical equilibration
condition \cite{noel2012tropical,noelgvr,rgzn,rvg,sgfwr,sgfr}, reads
\begin{equation}
\label{Teq}
\min_{\gamma_{k,J} > 0}   (c_{k,J}+\left<D,J\right>)
= \min_{\gamma_{k,J'} < 0}   (c_{k,J'}+\left<D,J'\right>).
\end{equation}

On these grounds, given system \eqref{origeq}, the choice of $\varepsilon_*$ boils down to
defining orders of magnitude. Model parameters are coarse-grained and
transformed to orders of magnitude in order to apply tropical scaling. The
result depends on which parameters are close and which are very different as
dictated by the coarse-graining procedure, i.e., by the choice of
$\varepsilon_*$. Decreasing $\varepsilon_*$ destroys details, and parameters tend to have the same
order of magnitude. Increasing $\varepsilon_*$ refines details, and parameters range over
several orders of magnitude. For instance, using \eqref{rationalindex} and
$p=1$, parameters $k_1=0.1$ and $k_2=0.01$ have orders $c_1 = 1$ and $c_2 =2$
for $\varepsilon_*=1/10$ but $c_1=c_2=1$ for $\varepsilon_*=1/50$. This is the perspective taken in
\cite{noel2012tropical,noelgvr,sgfwr}.

On the one hand, we have just seen that smaller choices of $\varepsilon_*$ possibly hide
details. On the other hand, in the following section we are going to review
singular perturbation methods, which provide asymptotic results as a small
parameter $\delta$ approaches zero. Following the construction in
Sect.~\ref{se:scale}, small choices of $\varepsilon_*$ lead to small $\delta$, which gives a
heuristic argument for choosing $\varepsilon_*$ rather small. Thus in practice one has to
reconcile two competing requirements, which unfortunately, still requires some
human intuition.

We are now ready to instantiate the black-box functions $c$ and $d$ in our
generic Algorithm~\ref{alg:scale} with tropical versions as given in
Algorithm~\ref{alg:tropicalc} and Algorithm~\ref{alg:tropicald}, respectively.
\begin{algorithm}
  \renewcommand{\algorithmiccomment}[1]{%
    \hfill\parbox{\widthof{$\in \QQ$}}{%
      \textcolor{gray}{#1}}}
  \caption{$\operatorname{TropicalC}$\label{alg:tropicalc}}
  \begin{algorithmic}[1]
    \REQUIRE
    \begin{enumerate}[nolistsep]
    \item $k \in \{ 1, \dots, n \}$;
    \item $J \in \{ 1, \dots, n \}^n$;
    \item A list $S = [\dot{y_1} = f_1, \dots, \dot{y_n} = f_n]$ of
    autonomous first-order ordinary differential equations where $f_1$, \dots,
    $f_n \in \QQ[y_1, \dots, y_n]$;
    \item $\varepsilon_* \in (0,1) \cap \QQ$.
    \item $p \in \NN \setminus \{0\}$
    \end{enumerate}
    \smallskip
    
    \ENSURE
    $c \in \QQ$
    \smallskip
    \STATE{$\gamma := \operatorname{coeff}(f_k, y^J)$
    \COMMENT{$\in \QQ$}}
    \STATE{$c := \round(p \log_{\varepsilon_*} |\gamma|) / p$
    \COMMENT{$\in \QQ$}}
    \RETURN{$c$}
  \end{algorithmic}
\end{algorithm}
\begin{algorithm}
  \renewcommand{\algorithmiccomment}[1]{%
    \hfill\parbox{\widthof{$\in \QQ$}}{%
      \textcolor{gray}{#1}}}
  \caption{$\operatorname{TropicalD}$\label{alg:tropicald}}
  \begin{algorithmic}[1]
    \REQUIRE
    \begin{enumerate}[nolistsep]
    \item A list $S = [\dot{y_1} = f_1, \dots, \dot{y_n} = f_n]$ of
    autonomous first-order ordinary differential equations where $f_1$, \dots,
    $f_n \in \QQ[y_1, \dots, y_n]$;
    \item $\varepsilon_* \in (0,1) \cap \QQ$.
    \item $p \in \NN \setminus \{0\}$
    \end{enumerate}
    \smallskip
    
    \ENSURE
    $(d_1, \dots, d_n) \in \QQ^n \cup \{\bot\}$
    \smallskip

    \STATE{\label{td:te}$\Pi(a_1, \dots, a_n) := \operatorname{TropicalEquilibration}(S, \varepsilon_*, p)$}
    \IF{\label{td:sat}\NOT $\RR \models \exists a_1 \dots \exists a_n \Pi$}
    \RETURN{\label{td:bot}$\bot$}
    \ENDIF
    \STATE{\label{td:choice}$(d_1, \dots, d_n) := \text{one possible choice for $a_1$, \dots, $a_n$}$}
    \RETURN{$(d_1, \dots, d_n)$}
  \end{algorithmic}
\end{algorithm}
\begin{algorithm}[t]
  \renewcommand{\algorithmiccomment}[1]{%
    \hfill\parbox{\widthof{$\langle A_k - A_\ell \mid A_0 \rangle \in \QQ[a_1, \dots, a_n]$}}{%
      \textcolor{gray}{#1}}}
  \caption{\label{alg:te}$\operatorname{TropicalEquilibration}$\label{alg:tropicalize}}
  \begin{algorithmic}[1]
    \REQUIRE
    \begin{enumerate}[nolistsep]
    \item
    A list $S = [\dot{y_1} = f_1, \dots, \dot{y_n} = f_n]$ of
    autonomous first-order ordinary differential equations where $f_1$, \dots,
    $f_n \in \QQ[y_1, \dots, y_n]$;
    \item
    $\varepsilon_* \in (0,1) \cap \QQ$;
    \item
    $p \in \NN \setminus \{0\}$.
    \end{enumerate}
    \smallskip
    
    \ENSURE A formula $\Pi(a_1, \dots, a_n)$ describing a finite union of convex
    polyhedra in $\RR^n$.
    \unskip\smallskip

    \STATE{$A_0 := (1, a_1, \dots, a_n)$}
    \COMMENT{$\in \QQ[a_1, \dots, a_n]^{n+1}$}
    \FOR{$j:=1$ \TO $n$}
    \STATE{$c := 0$}
    \FORALL{monomials $\gamma y_1^{\alpha_1} \cdots y_n^{\alpha_n}$
      \textbf{in} $f_j$}
    \STATE{$\alpha_0 := \round(p \log_{\varepsilon_*} |\gamma|) / p$
      \label{te:scaleks}
      \COMMENT{$\in \QQ$}}
    \STATE{$c := c+1$}
    \STATE{$\Sigma_c := \sgn \gamma$
      \COMMENT{$\in \{-1, 1\}$}}
    \STATE{$A_c := (\alpha_0, \alpha_1, \dots, \alpha_n)$
      \COMMENT{$\in \QQ \times \ZZ^n \subseteq \QQ^{n+1}$}}
    \ENDFOR
    \STATE{$B_j := \emptyset$}
    \FOR{$k := 1$ \TO $c$}
    \FOR{$\ell := k + 1$ \TO $c$}
    \IF{$\Sigma_k\Sigma_\ell < 0$}
    \STATE{$P := \{\langle A_k - A_\ell,  A_0 \rangle = 0\}$
      \COMMENT{$\langle A_k - A_\ell, A_0 \rangle \in \QQ[a_1, \dots, a_n]$}}
    \FOR{$m := 1$ \TO $c$}
    \STATE{$P := P \cup \{\langle A_m - A_k, A_0 \rangle \geq 0\}$
      \COMMENT{$\langle A_m - A_k, A_0 \rangle \in \QQ[a_1, \dots, a_n]$}}
    \ENDFOR
    \STATE{$B_j := B_j \cup \{ P \}$
      \COMMENT{set of sets of constraints}}
    \ENDIF
    \ENDFOR
    \ENDFOR
    \ENDFOR
    \STATE{$\Pi := \operatorname{DisjunctiveNormalForm}(\bigwedge_{j=1}^n \bigvee_{P \in B_j} \bigwedge
      P)$\label{te:dnf}}
    \RETURN{$\Pi$}
  \end{algorithmic}
\end{algorithm}
Algorithm~\ref{alg:tropicalc} explicitly uses, besides the parameters $k$ and
$J$ specified for $c$ in Algorithm~\ref{alg:scale}, also the right hand sides of
the input system \eqref{origeq} and the choice of $\varepsilon_*$. As yet another
parameter it takes the desired precision $p$ for rounding in
\eqref{rationalindex}. Notice that the use of this extra information is
compatible with the abstract scaling procedure in Sect.~\ref{se:scale}. Currying
\cite{CurryFeys:58} allows to use Algorithm~\ref{alg:tropicalc} in place of $c$
in a formally clean manner.

Similarly, Algorithm~\ref{alg:tropicald} takes parameters $\varepsilon_*$ and
$p$, while $d$ is specified in Algorithm~\ref{alg:scale} to have no parameters
at all. In l.\ref{td:te} we use Algorithm~\ref{alg:te} as a subalgorithm for
tropical equilibration. One obtains a disjunctive normal form $\Pi$, which
explicitly describes a set $\mP = \{\, p \in \QQ^n \mid \Pi(p)\,\}$ as a finite union of
convex polyhedra, as known from tropical geometry. Every
$(d_1, \dots, d_n) \in \mP$ satisfies \eqref{Teq}. The satisfiability condition in
l.\ref{td:sat} tests whether $\mP \neq \emptyset$. We employ \emph{Satisfiability Modulo
  Theories (SMT)} solving \cite{NieuwenhuisOliveras:06b} using the logic
\texttt{QF\_LRA} \cite{BarrettFontaine:17a} for quantifier-free linear real
arithmetic. The set $\mP$ can get empty, e.g, when all monomials on the right
hand side of some differential equation have the same sign. Such an exceptional
situation is signaled with a return value $\bot$ in l.\ref{td:bot}. In the regular
case $\mP \neq \emptyset$, the choice $(d_1, \dots, d_n)$ in l.\ref{td:choice} is provided by
the SMT solver. From a practical point of view, the disjunctive normal form
computation in Algorithm~\ref{alg:te} is a possible bottleneck and requires good
heuristic strategies \cite{Lueders:20a}.

With applications in the natural sciences one often wants to make in
l.\ref{td:choice} an adequate choice for $(d_1, \dots, d_n)$ lying in a specific
convex polyhedron $P \subseteq \mP$, which technically corresponds to one conjunction in
$\Pi$. Such choices are subtle and typically require human interaction. For
instance, when the chain of reduced dynamical systems ends with a steady state,
it is interesting to consider the polyhedron $P$ that is closest to that steady
state. Such strategies are not covered by our algorithms presented here.

At this stage we have obtained a scaled system as defined in
Sect.~\ref{se:scale}, including partitioning. The focus of the next section is
to utilize this scaling for analytically substantiated reductions.

\section{Singular Perturbation Methods}\label{se:spt}

The theory of singular perturbations is used to compute and justify
theoretically the limit of system \eqref{fullscalesimplifiedint} when
$\delta \to 0$. There are several types of results in this theory. The results of
Tikhonov, further improved by Hoppensteadt, show the convergence of the solution
of system \eqref{fullscalesimplifiedint} to the solution of a
differential-algebraic system in which the slowest variables $z_m$ follow
differential equations, and the remaining fast variables follow algebraic
equations \cite{tikh,hopp}. The results of Fenichel are known under the name of
{\em geometrical singular perturbations}. He showed that the algebraic equations
in Tikhonov's theory define a slow invariant manifold that is persistent for
$\delta > 0$ \cite{fenichel}. For geometrical singular perturbations,
differentiability in $\delta$ is needed in system \eqref{fullscalesimplifiedint}.

Samal et al.~have noted that Tikhonov's theorem is applicable to tropically
scaled systems \cite{sgfwr}. For instance, with $\delta_1 = \delta^{b_2}$, system
\eqref{fullscalesimplifiedint} may be rewritten as
\begin{equation}\label{eq:samal}
  z_1' = \wg_1(z, \delta_1),\quad z_2' = \delta_1 \wg_2(z, \delta_1),\quad \dots, \quad z_m' = \delta_1 \wg_m(z,
  \delta_1).
\end{equation}
However, this approach comes with certain limitations. To start with, it allows
only two time scales. Furthermore, in case $b_2>1$, there may be
differentiability issues with respect to $\delta_1$, and some care has to be taken
when one tries to apply to \eqref{eq:samal} also Fenichel's results
\cite{fenichel}. In this section, we are going to generalize geometrical
singular perturbations, and compute invariant manifolds and reduced models for
more than two time scales, introducing further $\delta_2$,
\dots,~$\delta_{\ell-1}$. Our generalization is based on a recent paper by Cardin and
Teixeira \cite{cartex}.

Section~\ref{se:extfen} presents relevant results from \cite{cartex} adapted to
our purposes here and applied to our system \eqref{fullscalesimplifiedint}. In
contrast to the original article, which is based on a series of hyperbolicity
conditions, we introduce the notion of hyperbolic attractivity, which is
stronger but still adequate for our purposes. In Sect.~\ref{se:ha} we describe
efficient algorithmic tests for hyperbolic attractivity. Section~\ref{se:diff}
gives sufficient algorithmic criteria addressing the above-mentioned
differentiability issues.

\subsection{Application of a Fenichel Theory for Multiple Time Scales}
\label{se:extfen}
We consider our system \eqref{fullscalesimplifiedint} over the positive first
orthant $\mU = (0, \infty)^n \subseteq \RR^n$. A recent paper by Cardin and Teixeira
\cite{cartex} generalizes Fenichel's theory to provide a solid foundation to
obtain more than one nontrivial invariant manifold. This allows, in particular,
the reduction of multi-time scale systems such as system
\eqref{fullscalesimplifiedint}. Technically, the approach considers a
multi-parameter system using time scale factors $\delta_1$, $\delta_1\delta_2$, \dots~instead of
increasing powers of one single $\delta$.

We let $\ell \in \{2, \dots, m\}$ and define
\begin{equation}\label{eq:defbeta}
  \beta_1 = b_2-b_1 = b_2,\quad \dots,\quad \beta_{\ell-1} = b_{\ell}-b_{\ell-1},
\end{equation}
and furthermore $\delta_1 = \delta^{\beta_1}$, \dots,
$\delta_{\ell-1} = \delta^{\beta_{\ell-1}}$, and $\wdelta = (\delta_1, \dots, \delta_{\ell-1})$.

These definitions allow us to express also all $\delta^{b_{k,j}'}$ occurring in
\eqref{fullscalesimplifiedint} as products of powers of $\delta_1$,
\dots,~$\delta_{\ell-1}$, with nonnegative but possibly non-integer rational exponents, via
expressing each $b_{k,j}'$ as a nonnegative rational linear combination of
$\beta_1$, \ldots,~$\beta_{\ell-1}$. This yields
\begin{equation}\label{eq:defwg}
  \wg_k(z, \delta_1, \dots, \delta_{\ell-1}) = \wf_k(z, \delta),\quad 1 \leq k \leq m.
\end{equation}
Moreover, we express
$\delta^{b_{\ell+1}}=\delta_1\cdots\delta_{\ell-1}\cdot\eta_{\ell+1}$, \dots,~$\delta^{b_m}=\delta_1\cdots\delta_{\ell-1}\cdot\eta_{m}$, via
\begin{equation}\label{eq:defeta}
  \eta_k(\delta_1, \dots, \delta_{\ell-1}) = \delta^{b_k-b_\ell},\quad \ell+1 \leq k \leq m,
\end{equation}
which is obtained by writing each $b_k - b_\ell$ as a nonnegative rational linear
combination of $\beta_1$, \ldots,~$\beta_{\ell-1}$. In these terms our system
\eqref{fullscalesimplifiedint} translates to
\begin{align*}
\label{fullsys}
  z_1'&= \wg_1(z,\wdelta)\\
  z_2'&=\delta_1 \wg_2(z, \wdelta)\\
      & \mathrel{\makebox[\widthof{=}]{\vdots}} \\
  z_\ell'&=\delta_1\cdots \delta_{\ell-1}\wg_\ell(z, \wdelta)\\
  z_{\ell+1}' &= \delta_1\cdots \delta_{\ell-1}\eta_{\ell+1}(\wdelta)\wg_{\ell+1}(z, \wdelta)\\
      & \mathrel{\makebox[\widthof{=}]{\vdots}}\\
  z_m' &= \delta_1\cdots \delta_{\ell-1}\eta_m(\wdelta)\wg_m(z, \wdelta).\numberthis
\end{align*}

In terms of the right hand sides of \eqref{fullsys} the application of relevant
results in \cite{cartex} requires that $\wg_1$, \dots,~$\wg_\ell$ and
$\eta_{\ell+1}\wg_{\ell+1}$, \dots,~$\eta_m\wg_m$ are smooth on an open neighborhood of
$\mU \times [0,\vartheta_1) \times \cdots \times [0,\vartheta_{\ell-1})$ with
$\vartheta_1>0$, \dots, $\vartheta_{\ell-1}>0$. We are going to tacitly assume such
smoothness here and address this issue from an algorithmic point of view in
Sect.~\ref{se:diff}.

We are now ready to transform our system into $\ell$ time scales as follows, where
possibly $\ell > 2$:
\[
  \tau_1 = \tau,\quad \tau_2 = \delta_1 \tau,\quad \dots,\quad \tau_\ell=\delta_1 \cdots \delta_{\ell-1} \tau.
\]
In time scale $\tau_k$, with $1\leq k\leq \ell$, system \eqref{fullsys} then becomes
\begin{align*}\label{dmone}
  \delta_1 \cdots \delta_{k-1} \diff{z_{1}}{\tau_k} &= \wg_1(z,\wdelta)\\
                                  & \mathrel{\makebox[\widthof{=}]{\vdots}} \\
  \delta_{k-1} \diff{z_{k-1}}{\tau_k} &= \wg_{k-1}(z,\wdelta)\\[1ex]
  \diff{z_k}{\tau_k} &= \wg_k(z,\wdelta)\\[1ex]
  \diff{z_{k+1}}{\tau_k} &= \delta_k\wg_{k+1}(z,\wdelta)\\
                                  & \mathrel{\makebox[\widthof{=}]{\vdots}} \\
  \diff{z_{\ell}}{\tau_k} &= \delta_k \cdots \delta_{\ell-1}\,\wg_\ell(z,\wdelta)\\[1ex]
  \diff{z_{\ell+1}}{\tau_k} &= {\delta_k \cdots \delta_{\ell-1}}\eta_{\ell+1}(\wdelta)\wg_{\ell+1}(z, \wdelta)\\
                                  & \mathrel{\makebox[\widthof{=}]{\vdots}} \\
  \diff{z_{m}}{\tau_k} &={\delta_k \cdots \delta_{\ell-1}} \eta_m(\wdelta)\wg_m(z, \wdelta).\numberthis
\end{align*}
For $k=1$ and $k=\ell$ we obtain empty products, which yield the neutral element
$1$, as usual.

Similarly to Sect.~\ref{se:scale}, we are interested in the asymptotic behavior
for $\wdelta \to 0$, which is approximated by the elimination of higher order
terms. We are now going to introduce a construction required for a justification
of this approximation, which also clarifies the greatest possible choice for
$\ell \leq m$ above. Define $F_0 = 0$ and
\begin{displaymath}
  Z_k = \pmat{z_1\\ \vdots\\ z_k},\quad
  F_k(z, \delta)=\begin{pmatrix}\wf_1(z, \delta)\\ \vdots \\ \wf_k(z, \delta)\end{pmatrix},
\quad 1\leq k\leq m.
\end{displaymath}
With system \eqref{fullscalesimplifiedint} in mind, we are going to use
$\wf_k(z, 0)$ in favor of $\wg_k(z, 0, \dots, 0)$. It is easy to see that both are
equal. We define furthermore
\begin{equation}
\label{submanifolds}
M_k = \bigl(F_k(z^*, 0)=0\bigr),\quad \mM_k = \left\{\, z^* \in \mU \mid F_k(z^*, 0)=0
  \,\right\},\quad 0 \leq k \leq m.
\end{equation}
The sets $\mM_k$ are obtained from varieties defined by the systems $M_k$ via
intersection with the first orthant. Furthermore,
\begin{equation}\label{eq:chain}
  \mU = \mM_0 \supseteq \mM_1 \supseteq \dots \supseteq \mM_m
\end{equation}
establishes a chain of nested subvarieties,
again intersected with the first orthant.

We define that $\mM_1$ is \emph{hyperbolically attractive} on $\mM_0$, if
$\mM_1 \neq \emptyset$, and for all $z \in \mM_1$ all eigenvalues of the Jacobian
$D_{z_1}\wf_1(z, 0)$ have negative real parts. Therefore $\mM_1$ is a manifold.
For $k \in \{ 2, \dots, m \}$, $\mM_{k}$ is hyperbolically attractive on
$\mM_{k-1}$, if $\mM_{k} \neq \emptyset$, and the following holds. Recall that using the
defining polynomials $F_{k-1}$ of $\mM_{k-1}$, the implicit function theorem
yields a unique local resolution of $Z_{k-1}$ as functions of $z_k$,
\ldots,~$z_m$, provided that $D_{z_k}F_{k-1}$ has no zero eigenvalues. We thus obtain
\[
  \wf_k(z,0)= \wf_k^*(z_k,\ldots,z_m,0) \quad\text{on}\quad \mM_{k-1}.
\]
For our definition we now require that for all $z \in \mM_{k}$ all eigenvalues of
$D_{z_k}\wf_k^*(z_k, \ldots, z_m, 0)$ have negative real parts. Again, $\mM_k$ is a
manifold. When $\mM_k$ is hyperbolically attractive on $\mM_{k-1}$ we write
$\mM_{k-1} \ha \mM_{k}$, where $Z_k$ will be clear from the context.

If we find for some $\ell \in \{1, \dots, m \}$ that $\mM_{0} \ha \mM_{1}$,
$\mM_{1} \ha \mM_{2}$, \dots, $\mM_{\ell-1} \ha \mM_{\ell}$, then we simply write
$\mM_0 \ha \dots \ha \mM_\ell$, and call this a \emph{hyperbolically attractive
  $\ell$-chain}. Such a chain is called \emph{maximal} if either $\ell = m$ or
$\mM_\ell \notha \mM_{\ell+1}$.

Let $\mM_0 \ha \dots \ha \mM_\ell$ be a hyperbolically attractive $\ell$-chain. Consider
for each $k \in \{1, \dots, \ell\}$ the following differential-algebraic system:
\begin{equation}\label{fullsyslimit}
0 = F_{k-1}(z,0), \quad
\diff{z_k}{\tau_k} = \wf_k(z,0),\quad
\diff{z_{k+1}}{\tau_k} = 0,\quad \dots,\quad \diff{z_m}{\tau_k} = 0.
\end{equation}
In the limiting case $\wdelta=0$, this corresponds to system \eqref{dmone}.
Recall that
\begin{displaymath}
  \tau_k = \delta_1 \cdots \delta_{k-1} \tau = \delta^{b_2-b_1} \cdots \delta^{b_{k}-b_{k-1}} \tau = \delta^{b_k-b_1} \tau
  = \delta^{b_k} \tau,
\end{displaymath}
and equivalently rewrite \eqref{fullsyslimit} as a triplet $(M_{k-1}, T_k, R_k)$
with entries as follows:
\begin{equation}\label{fullsyslimitdelta}
F_{k-1}(z,0) = 0, \quad
\diff{z_k}{\tau} = \delta^{b_k}\wf_k(z,0),\quad
\diff{z_{k+1}}{\tau} = \dots = \diff{z_m}{\tau} = 0.
\end{equation}
For a given index $k$, we call $(M_{k-1}, T_k, R_k)$ a \emph{reduced system} on
$\mM_{k-1}$, where the relevant hyperbolic attractivity relation is
$\mM_{k-1} \ha \mM_{k}$. In order to indicate the relevance of
$\mM_0 \ha \dots \ha \mM_\ell$ we write
$(M_0, T_1, R_1) \ha \dots \ha (M_{\ell-1}, T_\ell, R_\ell)$ also for reduced systems, where
$\mM_\ell$ is not made explicit but relevant for the last triplet. Slightly abusing
language, we speak of a hyperbolically attractive $\ell$-chain of reduced systems,
which is maximal if $\mM_0 \ha \dots \ha \mM_\ell$ is.

The following theorem is a consequence of \cite[Theorem~A and
Corollary~A]{cartex}, specialized to the situation at hand.

\begin{theorem}\label{th:ct}
  Let $\ell \geq 2$. Assume that
  $(M_0, T_1, R_1) \ha \dots \ha (M_{\ell-1}, T_\ell, R_\ell)$ is a hyperbolically attractive
  $\ell$-chain of reduced systems for system \eqref{fullsys}. Let
  $K \subseteq \mU$ be compact. Then for sufficiently small $\wdelta$ and all
  $k \in \{ 1, \dots, \ell \}$, system~\eqref{fullsys} admits invariant manifolds
  $\mN_{k-1}$ that depend on $\wdelta$ and are
  $(\delta_1 + \cdots + \delta_{k-1})$-close to $\mM_{k-1} \cap K$ with respect to the Hausdorff
  distance. Moreover, there exists $T > 0$ such that solutions of
  system~\eqref{fullsys} on $\mN_{k-1}$ in time scale $\tau_{k}$ converge to
  solutions of $(M_{k-1}, T_k, R_k)$, uniformly on any closed subinterval of
  $(0,\,T)$, as $\wdelta \to 0$.\qed
\end{theorem}

For $k \in \{1, \ldots, \ell\}$, the $\mM_{k-1}$ are critical manifolds, which contain
only stationary points. The systems $(T_k, R_k)$ of ordinary differential
equations on $\mM_{k-1}$ approximate invariant manifolds $\mN_{k-1}$ in the
sense of the theorem. They furthermore approximate solutions in time scale
$\tau_k$ of system~\eqref{fullsys}, which is equivalent to our system
\eqref{fullscalesimplifiedint}. In other words,
system~\eqref{fullscalesimplifiedint} admits a succession of invariant
manifolds, on which the behavior in the appropriate time scale is approximated
by the respective reduced equations \eqref{fullsyslimit} and, equivalently,
\eqref{fullsyslimitdelta}. Note that only the $\delta^{b_k}\wf_k(z,0)$ without the
higher order terms enter the reduced systems $(M_{k-1}, T_k, R_k)$.

Algorithm~\ref{alg:cls} 
\begin{algorithm}[t]
  \renewcommand{\algorithmiccomment}[1]{%
    \hfill\parbox{\widthof{$= M_0 \circ  [F = 0] \circ [f=0]$}}{\textcolor{gray}{#1}}}
  \caption{\label{alg:cls}$\operatorname{ComputeReducedSystems}$}
  \begin{algorithmic}[1]
    \REQUIRE Output of Algorithm~\ref{alg:scale}:
    \begin{enumerate}[nolistsep]
    \item $[T_1, \dots, T_m]$, a list of lists $z_k' = \delta^{b_k}f_k$;
    \item $[P_1, \dots, P_m]$, a list of lists of polynomials in $\QQ[x_1, \dots, x_n][\delta]$;
    \end{enumerate}

    We denote $\xi_k := |T_k|$, $\Xi_k := \sum_{i=1}^{k}\xi_i$, and
    $X=\{x_1, \dots, x_n\}$.\smallskip
    
    \ENSURE A list $[(M_0, T_1, R_1), \dots, (M_{\ell-1}, T_\ell, R_\ell)]$ of triplets where
    $\ell \in \{2, \dots, m\}$, or the empty list. For
    $k \in \{1, \dots, \ell\}$, $M_{k-1}$ is a list of real constraints defining
    $\mM_{k-1} \subseteq \RR^n$; $T_k$ is a list of differential equations;
    $R_k$ is a list of trivial differential equations $x'=0$ for all
    differential variables from $T_{k+1}$, \dots,~$T_m$. \smallskip

    The triplets $(M_{k-1}, T_{k}, R_{k})$ represent reduced systems
    according to \eqref{fullsyslimitdelta}.
    \smallskip

    \STATE{\label{cm:mzero}$U := [x_1>0, \dots, x_n>0]$}
    \STATE{$M_0, Z, F := [\,]$}
    \STATE{$A := \pmat{~}$}
    \FOR{\label{cm:loop}$k:=1$ \TO $m$}
    \STATE{$z := [\,x \mid x'=\delta^{b_k}g \in T_k\,]$
      \COMMENT{$\subseteq X$, $|z| = \xi_k$}\label{l:xk}}
    \STATE{$f := [\,g \mid x' = \delta^{b_k}g \in T_k\,]$
      \COMMENT{$= \hat{f}_k(z,0) \in \QQ[X]^{\xi_k}$}\label{l:g}}
    \STATE{$M_{k} := M_{k-1} \circ [f = 0]$
      \COMMENT{$= M_0 \circ [F = 0] \circ [f = 0]$}}
    \STATE{\label{cm:hha}$\varphi, A :=
      \operatorname{IsHyperbolicallyAttractive}(U \circ M_{k}, Z, z, F, f, k, A)$}
    \IF{\NOT $\varphi$}
    \STATE{\textbf{break}\label{cm:break}}
    \ENDIF
    \STATE{$R_k := [\,x' = 0 \mid x' = h \in T_{k+1} \cup \dots \cup T_m\,]$
      \COMMENT{$\Xi_{k-1} + \xi_k + |R_k| = n$}}
    \STATE{$Z := Z \circ z$
      \COMMENT{$\subseteq X$, $|Z| = \Xi_k$}}
    \STATE{$F := F \circ f$
      \COMMENT{$\in \QQ[X]^{\Xi_k}$}}
    \ENDFOR\label{cm:endloop}
    \STATE{\textit{\# We either broke in line \ref{cm:break} preserving $k$, or we have $k =
        m + 1$.}}
    \STATE{$\ell := k - 1$\label{cm:defl}}
    \IF{\label{cm:lstart}$\ell < 2$}
    \RETURN{$()$\label{cm:smalll}}
    \ENDIF\label{cm:lend}
    \IF{\label{cm:sdc}$\operatorname{TestSmoothness}([T_1, \dots, T_m], [P_1, \dots,
      P_m], \ell)$ = \text{failed}}
    \PRINT{``Warning: differentiability requires further verification''}
    \ENDIF
    \RETURN{$[(M_0, T_1, R_1)$, \dots, $(M_{\ell-1}, T_\ell, R_\ell)]$}
  \end{algorithmic}
\end{algorithm}
now starts with the output $[T_1, \dots, T_m]$ of Algorithm~\ref{alg:scale}, which
represents the scaled system \eqref{eq:scaleandtruncate}. Notice that each $T_k$
already meets the specification in \eqref{fullsyslimitdelta}. In
l.\ref{cm:mzero} we define $U$ to contain defining inequalities of the first
orthant $\mU$. Starting with $k=1$, the for-loop in
l.\ref{cm:loop}--\ref{cm:endloop} successively constructs $M_k$ and $R_k$ such
that in combination with $T_k$ from the input, $(M_{k-1}, T_k, R_k)$ forms a
reduced system as in \eqref{fullsyslimitdelta}. The loop stops when either
$k=m+1$ or a test for hyperbolic attractivity in l.\ref{cm:hha} finds that
$\mM_{k-1} \notha \mM_k$. We are going to discuss this test in detail in the
next section. Note that we maintain a matrix $A$ for storing information between
the subsequent calls of our test. In either case we arrive at a maximal
hyperbolically attractive $(k-1)$-chain of reduced systems given as a list
$[(M_{0}, T_1, R_1), \dots, (M_{k-2}, T_{k-1}, R_{k-1})]$. Following the notational
convention used throughout this section we set $\ell$ to $k-1$ in l.\ref{cm:defl}.
The test in l.\ref{cm:lstart}--\ref{cm:lend} reflects the choice of
$\ell \in \{2, \dots, m\}$ at the beginning of this section. Finally, l.\ref{cm:sdc} uses
the second input $[P_1, \dots, P_m]$ of the algorithm to address the smoothness
requirements for system \eqref{fullsys}. We are going to discuss the
corresponding procedure in detail in Sect.~\ref{se:diff}. It will turn out that
this procedure provides only a sufficient test. Therefore we issue in case of
failure only a warning, allowing the user to verify smoothness a posteriori,
using alternative algorithms or human intelligence. One might mention that it is
actually sufficient to consider weaker, finite differentiability conditions
instead of smoothness, which can be seen by inspection of the proofs in
\cite{cartex}.

From an application point of view, attracting invariant manifolds are relevant
in the context of biological networks, and our notion of hyperbolic attractivity
holds for large classes of such networks \cite{Feinberg:19a}. This is our
principal motivation for using hyperbolic attractivity here. From a
computational perspective, hyperbolic attractivity can be tested based on
Hurwitz criteria, as we are going to make explicit in the next section.

The relevant results in \cite{cartex}, in contrast, are based on a series of
hyperbolicity conditions, which are somewhat weaker than hyperbolic
attractivity. Hyperbolicity can be tested algorithmically as well, albeit with
more effort. For approaches based on Routh's work see, e.g., \cite[Chapter V, \S
4]{gantmacher}, which checks the number of purely imaginary eigenvalues of a
real polynomial via the Cauchy index of a related rational function.

\subsection{Verification of Hyperbolic Attractivity}\label{se:ha}
Our definition of hyperbolic attractivity $\mM_{k-1} \ha \mM_k$ refers to the
eigenvalues of the Jacobians of the $\wf_k^*$, which cannot be directly obtained
from the Jacobians of the $\wf_k$ \cite{cartex,cartexcorr}. Generalizing work on
systems with three time scales \cite{krwa}, we take in this section a linear
algebra approach to obtain the relevant eigenvalues without computing the
$\wf_k^*$.

To start with, recall the well-known \emph{Hurwitz criterion}
\cite{Hurwitz:95a}:

\begin{theorem}[Hurwitz, 1895]\label{th:hurwitz}
  Consider $f = a_0x^n+a_1x^{n-1}+\dots+a_n \in \RR[x]$, $a_0 > 0$. For
  $i \in \{1, \dots, n\}$ define
  \begin{displaymath}
    H_i = \begin{pmatrix}
      a_1 & a_3 & a_5 & \dots & a_{2i-1}\\
      a_0 & a_2 & a_4 & \dots & a_{2i-2}\\
      0   & a_1 & a_3 & \dots & a_{2i-3}\\
      \hdotsfor{5}\\
      0   & \hdotsfor{3} & a_i
    \end{pmatrix},\quad
    \Delta_i = \det H_i.
  \end{displaymath}
  Then all complex zeros of $f$ have negative real parts if and only if
  $\Delta_1 > 0$, \dots,~$\Delta_n > 0$. Notice that
  $\Delta_n = a_n\Delta_{n-1}$, and therefore $\Delta_n > 0$ can be equivalently replaced with
  $a_n > 0$.
\end{theorem}
We call $H_n$ the \emph{Hurwitz matrix} and $\Delta_i$ the \emph{$i$-th Hurwitz
  determinant} of $f$. Furthermore, we refer to
$\Gamma= (\Delta_1 > 0 \land \dots \land\Delta_{n-1} > 0 \land a_n > 0)$ as the \emph{Hurwitz conditions} for
$f$.

Our first result generalizes \cite[Proposition~1\,(ii)]{krwa}. The proof is
straightforward by induction using the argument in \cite[Lemma~3]{krwa} and its
proof.
\begin{lemma}\label{le:haone}
  For $k \in \{1, \dots, m\}$ define
  \begin{displaymath}
    J_k = \diag{1}{\varrho_1\cdots\varrho_{k-1}} \cdot D_{Z_k}F_k(z,0)
    = \pmatr{D_{z_1}\wf_1(z,0)&\dots & D_{z_k}\wf_1(z,0)\\
    \varrho_1 D_{z_1}\wf_2(z,0)&\dots & \varrho_1 D_{z_k}\wf_2(z,0)\\
    \hdotsfor{3} \\
    \varrho_1\cdots\varrho_{k-1} D_{z_1}\wf_k(z,0)& \dots & \varrho_1\cdots\varrho_{k-1} D_{z_k}\wf_k(z,0)}.
  \end{displaymath}
  Let $\ell \in \{1, \dots, m\}$. Then $\mM_0 \ha \dots \ha \mM_{\ell}$ if and only if
  $\mM_{\ell} \neq \emptyset$ and for all $k \in \{1, \dots, \ell\}$, all sufficiently small
  $\varrho_1^*>0$, \ldots,~$\varrho_{k-1}^*>0$, and all $z^*\in \mM_{k}$, all eigenvalues of
  $J_k(\varrho_1^*, \dots, \varrho_{k-1}^*, z^*)$ have negative real parts.
  
  In particular, one can choose
  $\varrho_1^* = \dots = \varrho_{k-1}^* = \varrho^*$ with sufficiently small
  $\varrho^*$ and consider
  $J_k' = \operatorname{diag}(1, \dots, \varrho^{k-1}) \cdot D_{Z_k}F_k(z, 0)$.
\end{lemma}
Let $\Gamma_k$ denote the Hurwitz conditions for the characteristic polynomial of
$J_k'$. Then Lemma~\ref{le:haone} allows to state hyperbolic attractivity
$\mM_0 \ha \dots \ha \mM_\ell$ as a first-order formula over the reals as follows:
\begin{equation}\label{eq:foea}
  \textstyle\biggl(\exists(0<z): F_\ell(z, 0) = 0\biggr) \land
  \biggl(\bigwedge\limits_{k=1}^\ell \exists(0 < \sigma) \forall(0 < \varrho < \sigma) 
    \forall(0 < z): F_k(z, 0) = 0 \Rightarrow \Gamma_k(\varrho, z)\biggr).
\end{equation}

On these grounds, any real decision procedure
\cite{Tarski:48a,Collins:75,Weispfenning:97b} provides an effective test for
hyperbolic attractivity. However, our formulation \eqref{eq:foea} uses a
quantifier alternation $\exists\sigma\forall\varrho$ in its second part. We would like to use this in
order to use SMT solving over a quantifier-free logic. Our next result allows a
suitable first-order formulation without quantifier alternation. Its proof
combines \cite[Lemma~3]{krwa} with our Lemma~\ref{le:haone}.

\begin{proposition}[Effective Characterization of Hyperbolically
  Attractive $\ell$-Chains]\label{pr:hyper}\sloppy
  Define $A_1 = D_{z_{1}}\wf_{1}(z,0)$. For $k \in \{2, \dots, m\}$ define
  \begin{displaymath}
    \pmatl{A_{k-1} & B_{k}\\ C_{k} & V_{k}} =
    \pmatl{D_{Z_{k-1}}F_{k-1}(z,0) & D_{z_{k}}F_{k-1}(z,0)\\
      D_{Z_{k-1}}\wf_{k}(z,0) & D_{z_{k}}\wf_{k}(z,0)},
  \end{displaymath}
  and note that $\spmatl{A_{k-1} & B_{k}\\ C_{k} & V_{k}} = A_k$. Let
  $\ell \in \{1, \dots, m\}$. Then $\mM_0 \ha \dots \ha \mM_\ell$ if and only if
  \begin{enumerate}[(i)]
  \item $\mM_\ell \ne \emptyset$,
  \item for all $z^* \in \mM_{1}$ all eigenvalues of $W_1(z^*)$, where
    $W_1 = A_1$, have negative real parts,
  \item for all $k \in \{2, \dots, \ell\}$ and all $z^* \in \mM_{k}$,
    $A_{k-1}(z^*)$ is regular and all eigenvalues of $W_k(z^*)$, where
    $W_k = V_{k}-C_{k}A_{k-1}^{-1}B_{k}$, have negative real parts.
  \end{enumerate}
\end{proposition}

\begin{proof}
  Assume $\mM_0 \ha \dots \ha \mM_\ell$. By Lemma~\ref{le:haone} we have
  $\mM_\ell \neq \emptyset$. For all $z^* \in \mM_1$, all eigenvalues of the Jacobian
  $W_1(z^*)$ have negative real parts by the definition of hyperbolic
  attractivity. Let now $k \in \{2, \dots, \ell\}$, $z^* \in \mM_k$, and define
  $\Rho = \operatorname{diag}(1, \dots, \varrho^{k-2})$.
  Using Lemma~\ref{le:haone} we fix $0 < \tau^* < 1$ such that for all
  $0 < \varrho^* < \tau^*$ all eigenvalues of
  $J'_{k-1}(\varrho^*, z^*) = \Rho(\varrho^*)A_{k-1}(z^*)$ have negative real parts. It
  follows that $\Rho(\varrho^*)A_{k-1}(z^*)$, $\Rho(\varrho^*)$, and
  $A_{k-1}(z^*)$ are all regular. Next, consider
  \begin{displaymath}
    J'_k=\pmat{
      \Rho A_{k-1} & \Rho B_{k}\\
      \varrho^{k-1} C_{k} & \varrho^{k-1} V_{k}}.
  \end{displaymath}
  Using Lemma~\ref{le:haone} once more, we find $0 < \sigma^* < \tau^*$ such that for
  all $0 < \varrho^* < \sigma^*$ also all eigenvalues of $J'_k(\varrho^*, z^*)$ have negative
  real parts. Now $J_k'(\varrho^*, z^*)$ satisfies condition (ii) of
  \cite[Lemma~3]{krwa} with $\delta = \sigma^*$ and
  $\varepsilon = (\varrho^*)^{k-1}$, which allows us to conclude that all eigenvalues of
  $(V_k - C_k(\Rho A_{k-1})^{-1}\Rho B_k)(\varrho^*, z^*) = (V_k -
  C_kA_{k-1}^{-1}\Rho^{-1}\Rho B_k)(\varrho^*, z^*) = W_k (z^*)$ have negative real
  parts as well.

  Assume, vice versa, that (i)--(iii) hold. We use induction on $k$ to show
  $\mM_0 \ha \dots \ha \mM_k$ for $1 \leq k \leq \ell$. For $k=1$ we have
  $\mM_0 \ha \mM_1$ by definition of hyperbolic attractivity. Assume that
  $2 \leq k \leq \ell$ and $\mM_0 \ha \dots \ha \mM_{k-1}$. By Lemma~\ref{le:haone} there
  exists $0 < \tau^*$ such that for all $0 < \sigma^* < \tau^*$ and all
  $z^* \in \mM_{k-1}$ all eigenvalues of $\Rho(\sigma^*) A_{k-1}(z^*)$ have negative
  real parts, where $\Rho = \operatorname{diag}(1, \dots, (\sigma^*)^{k-2})$. We rewrite
  $W_{k} = V_{k}-C_{k}(\Rho A_{k-1})^{-1} \Rho B_{k}$. Then $W_k(z^*)$ satisfies
  condition (i) of \cite[Lemma~3]{krwa} with $A=\Rho(\sigma^*) A_{k-1}(z^*)$,
  $B=\Rho(\sigma^*) B_k(z^*)$, $C=C_k(z^*)$ and $D=V_k(z^*)$. Thus there exists
  $0 < \delta$ such that for all $0 < \varepsilon < \delta$ all eigenvalues of
  \begin{displaymath}
    \pmat{ \Rho(\sigma^*) A_{k-1}(z^*) & \Rho(\sigma^*) B_k(z^*) \\ \varepsilon C_k(z^*) & \varepsilon
      V_k(z^*) }
  \end{displaymath}
  have negative real parts. Choosing $\varrho^* = \min\{\sigma^*, \sqrt[k-1]\varepsilon\}$ in
  Lemma~\ref{le:haone} yields $\mM_0 \ha \dots \ha \mM_k$.
\end{proof}

From now on let $\Gamma_k$ denote the Hurwitz conditions for the characteristic
polynomial of $W_k$, which---in contrast to the ones used in \eqref{eq:foea}---do
not depend on $\varrho$ anymore.

\begin{corollary}[Logic-Based Test for Hyperbolically Attractive
  $\ell$-Chains]\label{co:hyper}
  For $k \in \{1, \dots, m\}$ define
  \begin{align*}
    \varphi_k &= \bigl(\exists(0<z): F_k(z, 0) = 0 )\bigr),\\
    \psi_k &= \bigl(\forall(0 < z): F_k(z, 0) = 0 \Rightarrow \Gamma_k(z)\bigr).
  \end{align*}
  Let $\ell \in \{1, \dots, m\}$. Then $\mM_0 \ha \dots \ha \mM_\ell$ if and only if
  $\RR \models \varphi_\ell \land \bigwedge_{k=1}^\ell \psi_k$.
\end{corollary}

\begin{proof}
  Assume $\mM_0 \ha \dots \ha \mM_\ell$. Then Proposition~\ref{pr:hyper} yields its
  conditions (i)--(iii). Now, $\varphi_\ell$ holds as a formalization of (i). Furthermore,
  $\psi_1$ holds as a formalization of (ii), and the validity of $\psi_2$,
  \dots,~$\psi_\ell$ follows directly from (iii). Hence $\RR \models \varphi_\ell \land \bigwedge_{k=1}^\ell \psi_k$.

  Assume, vice versa, that
  $\RR \models \varphi_\ell \land \bigwedge_{k=1}^\ell \psi_k$. We show
  $\mM_0 \ha \dots \ha \mM_\ell$ by induction on $\ell$. If $\ell = 1$, then
  $\varphi_1$ formalizes (i) and $\psi_1$ formalizes (ii) in Proposition~\ref{pr:hyper},
  and we obtain $\mM_0 \ha \mM_1$. Let now $\ell > 1$. Then $\varphi_\ell$ formalizes
  Proposition~\ref{pr:hyper}\,(i). Our induction hypothesis yields
  $\mM_0 \ha \dots \ha \mM_{\ell-1}$. By Lemma~\ref{le:haone} there exists
  $0 < \tau^*$ such that for all $0 < \sigma^* < \tau^*$ and all
  $z^* \in \mM_{\ell-1} \supseteq \mM_\ell$ all eigenvalues of
  $\Rho(\sigma^*) A_{\ell-1}(z^*)$, where
  $\Rho = \operatorname{diag}(1, \dots, (\sigma^*)^{\ell-2})$, have negative real parts. In
  particular, $\Rho(\sigma^*) A_{\ell-1}(z^*)$ is regular and so is
  $A_{\ell-1}(z^*)$. Furthermore, the Hurwitz conditions in $\psi_\ell$ guarantee for all
  $z^*\in \mM_\ell$ that all eigenvalues of $W_\ell(z^*)$ have negative real parts.
  Taking these observations together, Proposition~\ref{pr:hyper}\,(iii) is
  satisfied, hence $\mM_0 \ha \dots \ha \mM_\ell$.
\end{proof}

In contrast to \eqref{eq:foea}, our first-order characterization
\begin{equation}
  \label{eq:foa}
  \textstyle\biggl(\exists(0<z): F_\ell(z, 0) = 0\biggr) \land
  \biggl(\bigwedge\limits_{k=1}^\ell \forall(0 < z): F_k(z, 0) = 0 \Rightarrow \Gamma_k(z)\biggr)
\end{equation}
in Corollary~\ref{co:hyper} has no quantifier alternation. Note that the two
top-level components of \eqref{eq:foa} establish two independent decision
problems, addressing non-emptiness of the manifold and our requirement on the
eigenvalues, respectively.

It is easy to see that for all $\ell \in \{1, \dots, m\}$ and all
$k \in \{1, \dots, \ell-1\}$, $\varphi_\ell$ entails $\varphi_k$. Thus \eqref{eq:foa} can be
equivalently rewritten as $\textstyle \bigwedge_{k=1}^\ell(\varphi_k \land \psi_k)$, explicitly:
\begin{equation}
  \label{eq:foav}
  \textstyle \bigwedge\limits_{k=1}^\ell\bigl(\exists(0<z): F_k(z, 0) = 0 \land
  \forall(0 < z): F_k(z, 0) = 0 \Rightarrow \Gamma_k(z)\bigr).
\end{equation}
Our approach tests the conjunction in \eqref{eq:foav} using a for-loop over $k$
in Algorithm~\ref{alg:cls}. Technically, this construction ensures with the test
for $\mM_{k-1} \ha \mM_k$ in Algorithm~\ref{alg:ha}
\begin{algorithm}[t]
  \renewcommand{\algorithmiccomment}[1]{%
    \hfill\parbox{\widthof{$\in \QQ[X]^{(\Xi+\xi) \times (\Xi+\xi)}$}}{\textcolor{gray}{#1}}}
  \begin{algorithmic}[1]
    \caption{\label{alg:ha}$\operatorname{IsHyperbolicallyAttractive}$}
    \REQUIRE 1.~$M$,~~2.~$Z$,~~3.~$z$,~~4.~$F$,~~5.~$f$,~~6.~$k$,~~7.~$A$, as in
    the calling Algorithm~\ref{alg:cls}

    Knowing that $\mM_{0} \ha \dots \ha \mM_{k-1}$, we check here whether also
    $\mM_{k-1} \ha \mM_{k}$. We denote $\xi := |f| = |z|$, $\Xi := |F| = |Z|$, and
    $X=\{x_1, \dots, x_n\}$. In these terms, $A \in \QQ[X]^{\Xi \times \Xi}$. \smallskip

    \ENSURE 1.~Boolean,~~2.~$A' \in \QQ[X]^{(\Xi + \xi) \times (\Xi + \xi)}$
    \smallskip
    
    \IF{\label{ha:neimodels}\NOT $\RR \models \underline{\exists}\bigwedge M$}
    \RETURN{$\operatorname{false}$, $\pmat{~}$}
    \ENDIF \label{ha:neii}
    \STATE{\label{ha:jac}$V := \operatorname{Jacobian}(f, z)$
      \COMMENT{$\in \QQ[X]^{\xi \times \xi}$}}
    \IF{\label{ha:ifl}$k = 1$}
    \STATE{$W := V$}
    \STATE{$A' := V$}
    \ELSE
    \STATE{$B := \operatorname{Jacobian}(F, z)$
      \COMMENT{$\in \QQ[X]^{\Xi \times \xi}$}}
    \STATE{$C := \operatorname{Jacobian}(f, Z)$
      \COMMENT{$\in \QQ[X]^{\xi \times \Xi}$}}
    \STATE{$W := V - C A^{-1} B$
      \COMMENT{$\in \QQ[X]^{\xi \times \xi}$}}
    \STATE{$A' := \begin{pmatrix} A & B\\ C & V \end{pmatrix}$
      \COMMENT{$\in \QQ[X]^{(\Xi+\xi) \times (\Xi+\xi)}$}}
    \ENDIF\label{ha:endl}
    \STATE{\label{ha:char}$\chi := \lambda^{\xi} + \dots + a_{\xi} := \operatorname{CharacteristicPolynomial}(W)$
      \COMMENT{$\in \QQ[X][\lambda]$}\label{l:chi}}
    \STATE{$H := \operatorname{HurwitzMatrix}(\chi)$
      \COMMENT{$\in \QQ[X]^{\xi \times \xi}$}}
    \FOR{$j := 1$ \TO $\xi-1$}
    \STATE{$\Delta_j := \det \pmat{H_{r, s}}_{1 \leq r, s \leq j}$
      \COMMENT{$\in \QQ[X]$}}
    \ENDFOR
    \STATE{\label{ha:gamma}$\Gamma := \{\Delta_1 > 0, \dots, \Delta_{\xi - 1} > 0, a_{\xi} > 0\}$}
    \RETURN \label{ha:finally}$\RR \models \underline{\forall}(\bigwedge M \longrightarrow \bigwedge \Gamma)$, $A'$
  \end{algorithmic}
\end{algorithm}
that $\mM_0 \ha \dots \ha \mM_{k-1}$ already holds, and exploits the fact that
$\psi_k$ and $\varphi_k$ do not refer to smaller indices than $k$.

In l.\ref{ha:neimodels}--\ref{ha:neii} we test the validity of $\varphi_k$. Using from
the input the defining inequalities and equations $M = U \circ M_k$ of $\mM_{k}$
along with $Z=Z_{k-1}$, $z=z_k$, $F=F_{k-1}$, $f=f_k$, and $A=A_{k-1}$, we
construct in l.\ref{ha:jac}--\ref{ha:endl} $A'=A_k$ as noted in
Proposition~\ref{pr:hyper}. In l.\ref{ha:char}--\ref{ha:gamma} we construct the
Hurwitz conditions $\Gamma=\Gamma_k$ according to Theorem~\ref{th:hurwitz}. On the grounds
of the validity of $\varphi_k$ tested in l.\ref{ha:neimodels}, we finally test in
l.\ref{ha:finally} the validity of $\psi_k$ and return a corresponding Boolean
value. We additionally return $A'=A_{k}$ for reuse with the next iteration. The
validity tests for $\varphi_k$ and $\psi_k$ in l.\ref{ha:neimodels} and
l.\ref{ha:finally}, respectively, again amount to SMT solving, this time using
the logic \texttt{QF\_NRA} \cite{BarrettFontaine:17a} for quantifier-free
nonlinear real arithmetic. Recall the positive integer parameter $p$ used for
the precision with both Algorithm~\ref{alg:tropicalc} and
Algorithm~\ref{alg:tropicald}. For $p > 1$ symbolic computation possibly yields
fractional powers of numbers in the defining equations for manifolds as well as
in the vector fields of the differential equations. Such expressions are not
covered by \texttt{QF\_NRA}. When this happens, we catch the corresponding error
from the SMT solver and restart with floats.

\subsection{Sufficient Smoothness Criteria}\label{se:diff}
Let us get back to the requirement in Sect.~\ref{se:extfen} that $\wg_1$,
\dots,~$\wg_\ell$ and $\eta_{\ell+1}\wg_{\ell+1} $, \dots,~$\eta_m\wg_m$ occurring on the right hand
sides of system \eqref{fullsys} are all smooth on an open neighborhood of
$\mU \times [0,\vartheta_1) \times \cdots \times [0,\vartheta_{\ell-1})$ with
$\vartheta_1>0$, \dots, $\vartheta_{\ell-1}>0$. A first sufficient criterion for
smoothness is that all those expressions are polynomials in $z$ and $\wdelta$.

Recall the definitions of $\wg_k$ for $k \in \{1, \dots, m\}$ in \eqref{eq:defwg} and
of $\eta_k$ for $k \in \{\ell+1, \dots, m\}$ in \eqref{eq:defeta}. For
$k \in \{1, \dots, m\}$ and $j \in \{1, \dots, w_k\}$ one finds nonnegative
$r_1$, \dots,~$r_{\ell-1} \in \QQ$ such that
\begin{displaymath}
  \langle (\beta_1, \dots, \beta_{\ell-1}), (r_1, \dots, r_{\ell-1}) \rangle = b_{k,j}',
\end{displaymath}
and for $k \in \{\ell + 1, \dots, m\}$ one finds nonnegative $r_1$,
\dots,~$r_{\ell-1} \in \QQ$ such that
\begin{displaymath}
  \langle (\beta_1, \dots, \beta_{\ell-1}), (r_1, \dots, r_{\ell-1}) \rangle = b_k - b_\ell.
\end{displaymath}
Such representations always exist but are not unique in general. If one even
finds suitable nonnegative integers $r_1$, \dots,~$r_{\ell-1} \in \NN$, which do not always
exist, then one obtains $\wg_1$, \dots,~$\wg_m$ as polynomials in $z$ and
$\wdelta$, and $\eta_{\ell+1}$, \dots,~$\eta_{m}$ as polynomials in $\wdelta$, which is
sufficient for our first criterion above.

An improved but still only sufficient criterion uses similar constructions to
directly verify the existence of polynomial representations of the products
$\eta_{\ell+1}\wg_{\ell+1} $, \dots,~$\eta_m\wg_m$, in contrast to considering the factors
independently. From an algorithmic point of view, we furthermore have to take
into account that $P_1$, \dots,~$P_m$ obtained in Algorithm~\ref{alg:scale} do not
contain $b_{k,j}'$ but $b_k + b_{k,j}'$. For $k \in \{1, \dots, \ell\}$ and
$j \in \{1, \dots, w_k\}$ we try to find $r_1$, \dots,~$r_{\ell-1} \in \NN$ such that
\begin{equation}\label{eq:sdcone}
  \langle (\beta_1, \dots, \beta_{\ell-1}), (r_1, \dots, r_{\ell-1}) \rangle = b_{k,j}' = (b_k + b_{k,j}') - b_k
  > 0,
\end{equation}
and for $k \in \{\ell + 1, \dots, m\}$ we try to find $r_1$, \dots,~$r_{\ell-1} \in \NN$ such that
\begin{equation}\label{eq:sdctwo}
  \langle (\beta_1, \dots, \beta_{\ell-1}), (r_1, \dots, r_{\ell-1}) \rangle = (b_k - b_\ell) + b_{k,j}' = (b_k +
  b_{k,j}') - b_\ell > 1.
\end{equation}
Notably, such representations exist whenever $1 \in \{\beta_1, \dots, \beta_{\ell-1}\}$.

On these grounds, we introduce Algorithm~\ref{alg:sdc}, 
\begin{algorithm}[t]
  \renewcommand{\algorithmiccomment}[1]{%
    \hfill\parbox{\widthof{$\subseteq \NN \setminus \{0\}$}}{\textcolor{gray}{#1}}}
  \caption{\label{alg:sdc}$\operatorname{TestSmoothness}$}
  \begin{algorithmic}[1]
    \REQUIRE $[T_1, \dots, T_m]$, $[P_1, \dots, P_m]$, $\ell$ as in the calling
    Algorithm~\ref{alg:cls}:
    \begin{enumerate}[nolistsep]\raggedright
    \item $[T_1, \dots, T_m]$, a list of lists $z_k' = \delta^{b_k}f_k$;
    \item $[P_1, \dots, P_m]$, a list of lists of polynomials in $\QQ[x_1, \dots, x_n][\delta]$;
    \item $\ell \in \NN$, $\ell \geq 2$;
    \end{enumerate}
    \smallskip

    We check here a sufficient criterion for smoothness as required for
    \eqref{fullsys}.\smallskip
    
    \ENSURE
    ``true'' or ``failed'' in terms of a 3-valued logic;
    \smallskip
    
    \STATE{\label{sdc:bone}$b_1 := 0$}
    \FOR{$k := 2$ \TO $\ell$}
    \STATE{$b_{k} := \text{the unique exponent of $\delta$ in $T_k$}$}
    \STATE{$\beta_{k-1} := b_k - b_{k-1}$}
    \IF{$\beta_{k-1} = 1$}
    \RETURN{$\operatorname{true}$}
    \ENDIF
    \ENDFOR\label{sdc:endfor}
    \STATE{\label{sdc:ebegin}$E := \emptyset$}
    \FOR{$k=1$ \TO $m$}
    \FORALL{$p$ \textbf{in} $P_k$}
    \STATE{$E := E \cup \{\,\deg_\delta m - b_{\min(k, \ell)} \mid \text{$m$ monomial of $p$}\,\}$
      \COMMENT{$\subseteq \NN \setminus \{0\}$}}
    \ENDFOR
    \ENDFOR\label{sdc:eend}
    \FORALL{$e \in E$}
    \IF{\label{sdc:test}\NOT $\ZZ \models \exists r_1 \ldots \exists r_{\ell-1} (r_1 \ge 0 \land \dots \land r_{\ell-1} \ge 0 \land \langle (\beta_1, \dots, \beta_{\ell-1}), (r_1, \dots, r_{\ell-1}) \rangle = e)$}
    \RETURN{\label{sdc:false}$\operatorname{failed}$}
    \ENDIF
    \ENDFOR
    \RETURN{\label{sdc:true}$\operatorname{true}$}
  \end{algorithmic}
\end{algorithm}
which specifies the sufficient test applied in l.\ref{cm:sdc} of
Algorithm~\ref{alg:cls}. The first two parameters $[T_1, \dots, T_m]$ and
$[P_1, \dots, P_m]$ originate from Algorithm~\ref{alg:scale}, while the last
parameter $\ell$ originates from the calling Algorithm~\ref{alg:cls}.

In l.\ref{sdc:bone}--\ref{sdc:endfor} of Algorithm~\ref{alg:sdc} we compute
$\beta_1$, \dots, $\beta_{\ell-1}$ as defined in \eqref{eq:defbeta} and simultaneously obtain
$b_1$, \dots, $b_\ell$. In l.\ref{sdc:ebegin}--\ref{sdc:eend} we compute the right hand
sides of the conditions in \eqref{eq:sdcone} or \eqref{eq:sdctwo}, depending on
the current index $k$. For checking those conditions in l.\ref{sdc:test} we once
more employ SMT solving, this time using the adequate logic \texttt{QF\_LIA}
\cite{BarrettFontaine:17a} for quantifier-free linear integer arithmetic. Since
we are aiming at nonnegative integer solutions, we introduce explicit
non-negativity conditions $r_1 \ge 0$, \ldots,~$r_{\ell-1} \ge 0$. In case of
unsatisfiability Algorithm~\ref{alg:sdc} returns ``failed'' in
l.\ref{sdc:false}. Recall that is this case the calling Algorithm~\ref{alg:cls}
issues a warning but continues. In case of satisfiability, in contrast,
smoothness is guaranteed, we reach l.\ref{sdc:true}, and return ``true.'' We
remark that the computation time spent on $E$ is negligible compared to the SMT
solving later on. The construction of the entire set $E$ beforehand avoids
duplicate SMT instances.

\section{Algebraic Simplification of Reduced Systems}\label{se:simpl}

In the output $(M_0, T_1, R_1)$, \dots, $(M_{\ell-1}, T_\ell, R_\ell)$ of
Algorithm~\ref{alg:cls}, the $T_k$ are taken literally from the input, and the
$M_{k-1}$ and $R_k$ are obtained via quite straightforward rewriting of the
input. As a matter of fact, the computationally hard part of
Algorithm~\ref{alg:cls} consists in the computation of the upper index $\ell$. We
now want to rewrite the triplets $(M_{k-1}, T_k, R_k)$ once more, aiming at less
straightforward but simpler and, hopefully, more intuitive representations. The
principal idea is to heuristically eliminate on the right hand side of the
differential equations in $T_k$ those variables whose derivatives have already
occurred as left hand sides in one of the $T_1$, \dots,~$T_{k-1}$. Of course, our
simplifications will preserve all relevant properties of $(M_0, T_1, R_1)$, \dots,
$(M_{\ell-1}, T_\ell, R_\ell)$, such as hyperbolic attractivity and sufficient
differentiability. Technically, our next Algorithm~\ref{alg:simpl}
\begin{algorithm}[t]
  \renewcommand{\algorithmiccomment}[1]{%
    \hfill\parbox{\widthof{$= \hat{f}_\ell(z,0) \in \QQ[X]^{\xi_\ell}$}}{\textcolor{gray}{#1}}}
  \caption{$\operatorname{SimplifyReducedSystems}$\label{alg:simpl}}
  \begin{algorithmic}[1]
    \REQUIRE A list
    $[(M_0, T_1, R_1), \dots, (M_{\ell-1}, T_{\ell}, R_{\ell})]$, the output of
    Algorithm~\ref{alg:cls}, with entries corresponding to
    \eqref{fullsyslimitdelta} \smallskip
    
    \ENSURE A list
    $[(M_0', T_1', R_1), \dots, (M_{\ell-1}', T_{\ell}', R_{\ell})]$; $M_{k-1}'$ describes
    the same manifold as $M_{k-1}$ in a canonical form; the system $T_k'$ is
    equivalent to $T_k$ modulo $M_{k-1}'$, its right hand sides are in a
    canonical normal form modulo $M_{k-1}'$, possibly with fewer different
    differential variables than $T_k$\smallskip

    \FOR{\label{simpl:loop}$k:=1$ \TO $\ell$}
    \STATE{$z_k := \{\,x \mid x' = g \in T_k\,\}$}
    \ENDFOR
    \STATE{$y := \{\,x \mid x' = 0 \in R_{\ell}\,\}$}
    \STATE{\label{simpl:omega}$\omega := \text{a block term order with $z_1 \gg \dots \gg
        z_\ell \gg y$}$}
    \FOR{$k:=1$ \TO $\ell$}
    \STATE{$F := [\, f \mid f=0 \in M_{k-1}\,]$}
    \STATE{\label{simpl:gb}$G := \operatorname{GroebnerBasis}(\operatorname{Radical}(F), \omega)$}
    \STATE{\label{simpl:mm}$M_{k-1}' := [\, g=0 \mid g \in G\,]$}
    \STATE{$T_k' := [\,]$}
    \FOR{\label{simpl:loopx}$x'=g$ \textbf{in} $T_k$}
    \STATE{\label{simpl:tk}$T_k' := T_k' \circ [x' = h]$ where $g \longrightarrow_G^* h$ and $h$ is irreducible
      mod $G$}
    \ENDFOR\label{simpl:endfor}
    \ENDFOR
    \RETURN{$[(M_0', T_1', R_1), \dots, (M_{\ell-1}', T_{\ell}', R_{\ell})]$}
  \end{algorithmic}
\end{algorithm}
employs Gröbner basis techniques \cite{Buchberger:65a,BeckerWeispfenning:93a}.

Recall that $z_k$ are the variables occurring on the left hand sides of
differential equations in $T_k$, and $Z_{k-1} = (z_1, \dots, z_{k-1})$. In
l.\ref{simpl:loop}--\ref{simpl:omega} we construct a block term order $\omega$ on all
variables $\{x_1, \dots, x_n\}$ so that variables from $Z_{k-1}$ are larger than
variables from $z_k$. This ensures that all multivariate polynomial reductions
with respect to $\omega$ throughout our algorithm will eliminate variables from
$Z_{k-1}$ in favor of variables from $z_k$ rather than vice versa. Prominent
examples for such block orders are pure lexicographical orders, but ordering by
total degree inside the $z_1$, \dots,~$z_\ell$, $y$ will heuristically give more
efficient computations.

Recall that the radical ideal $\sqrt{\langle F \rangle}$ is the infinite set of \emph{all}
polynomials with the same common complex roots as $F$. In l.\ref{simpl:gb}, we
compute a finite reduced Gröbner basis $G$ with respect to $\omega$ of that radical.
If radical computation is not available on the software side, then the algorithm
remains correct with a Gröbner basis of the ideal $\langle F \rangle$ instead of the radical
ideal, but might miss some simplifications.

In l.\ref{simpl:mm}, the polynomials in $G$ equivalently replace the left hand
side polynomials of the equations in $M_{k-1}$. In l.\ref{simpl:tk}, reduction
with respect to $\omega$, which comes with heuristic elimination of variables,
applies once more to the reduction results $h$ obtained from right hand sides
$g$ of differential equations in $T_k$. Since $G$ is a Gröbner basis, the
reduction in l.\ref{simpl:loopx}--\ref{simpl:endfor} furthermore produces unique
normal forms with the following property: if two polynomials $g_1$, $g_2$
coincide on the manifold $\mM_{k-1}$ defined by $M_{k-1}$, then they reduce to
the same normal form $h$. In particular, if $g_1$ vanishes on $\mM_{k-1}$, then
it reduces to $0$. We call the output of Algorithm~\ref{alg:simpl}
\emph{simplified reduced systems}.

\section{Back-Transformation of Reduced Systems}\label{se:tb}
Let $\ell \in \{2, \dots, m\}$ and $k \in \{1, \dots, \ell\}$. Recall that a triplet
$(M_{k-1}, T_k, R_k)$ obtained from Algorithm~\ref{alg:cls} describes a reduced
system according to \eqref{fullsyslimitdelta}:
\begin{equation*}
F_{k-1}(z,0) = 0, \quad
\diff{z_k}{\tau} = \delta^{b_k}\wf_k(z,0),\quad
\diff{z_{k+1}}{\tau} = \dots = \diff{z_m}{\tau} = 0.
\end{equation*}
A corresponding simplified system $(M_{k-1}', T_k', R_k)$ can be obtained from
Algorithm~\ref{alg:simpl} via an equivalence transformation on the set of
equations $M_{k-1}$ and further equivalence transformations modulo $M_{k-1}$ on
the right hand sides of the differential equations in $T_k$, while the left hand
sides of those differential equations remain untouched. It is not hard to see
that for both these outputs scaling can be reversed using the substitution
\begin{displaymath}
\sigma = [x_1 \gets y_1 / \varepsilon^{d_1}, \dots, x_n \gets y_n / \varepsilon^{d_n}] \circ [\tau \gets \varepsilon^\mu t] \circ [\delta \gets \varepsilon^{1/q}] \circ [\varepsilon \gets \varepsilon_*]
\end{displaymath}
obtained with Algorithm~\ref{alg:scale}.

Using names $M_{k-1}$, $T_k$, $R_k$ as in the unsimplified system, this yields a
raw back-transformation $(M_{k-1}\sigma, T_k\sigma, R_k\sigma)$. We define
$M_{k-1}^* = M_{k-1}\sigma$. The system $T_k\sigma$ can be written as
\begin{equation*}
  \diff{y_j}{\varepsilon_*^{\mu+d_j}t} = \varepsilon_*^{b_k/q} (\wf_k(z,0)_j\sigma),\quad
  x_j \in z_k.
\end{equation*}
We multiply by $\varepsilon_*^{\mu+d_j}$ in order to arrive at differential equations in
$\diff{y_j}{t}$. Furthermore, recall that the explicit factor $\delta^{b_k}$ in the
original $T_k$ corresponds to a time scale $\delta^{b_k}\tau$. The corresponding time
scale in $t$ is given by
$(\delta^{b_k} \tau)\sigma = \varepsilon_*^{b_k/q+\mu} t$, which we make explicit by equivalently
rewriting $T_k\sigma$ as $T_k^*$ as follows:
\begin{equation*}
  \diff{y_j}{t} = \varepsilon_*^{b_k/q+\mu} (\varepsilon_*^{d_j} \wf_k(z,0)_j\sigma),\quad
  x_j \in z_k.
\end{equation*}
Similarly, $R_k\sigma$ can be rewritten as $R_k^*$ as follows:
\begin{equation*}
  \diff{y_j}{t} = 0,\quad
  x_j \in z_{k+1} \cup \dots \cup z_m.
\end{equation*}

Recall that $M_{k-1}$ describes a manifolds $\mM_{k-1}$ over the positive first
orthant $\mU$, which is defined by inequalities $U = \{ x_1>0, \dots, x_n>0 \}$ not
explicit in $M_{k-1}$. Following our construction above, this translates into
$U\sigma = \{ y_1/\varepsilon_*^{d_1} > 0, \dots, y_n/\varepsilon_*^{d_n} > 0 \}$, which describes again the
positive first orthant. Furthermore, the nestedness \eqref{eq:chain} is
preserved:
\begin{displaymath}
  \mU = \mM_0^* \supseteq \mM_1^* \supseteq \dots \supseteq \mM_{\ell-1}^*.
\end{displaymath}
Finally, the system $(T_k^*,R_k^*)$ defines differential equations on
$\mM_{k-1}^*$.

We call $(M_0^*, T_1^*, R_1^*)$, \dots, $(M_{\ell-1}^*, T_\ell^*, R_\ell^*)$
\emph{back-transformed reduced systems}. In terms of the definitions after
\eqref{eq:scaleandtruncate} in Sect.~\ref{se:scale} we have reverted the scaling
but not the partitioning and not the truncating. Furthermore, we have preserved
all information obtained with the computation of the reduced systems in
Sect.~\ref{se:spt}, where we keep the time scale factors explicit, and with
their algebraic simplification in Sect.~\ref{se:simpl}.

\begin{algorithm}
  \renewcommand{\algorithmiccomment}[1]{%
    \hfill\parbox{\widthof{$=\varepsilon_*^{b_k/q + \mu + d_j} (f_j\sigma)$}}{\textcolor{gray}{#1}}}
  \caption{$\operatorname{TransformBack}$\label{alg:tback}}
  \begin{algorithmic}[1]
    \REQUIRE
    \begin{enumerate}[nolistsep]
    \item $[(M_0, T_1, R_1), \dots, (M_{\ell-1}, T_{\ell}, R_{\ell})]$, the output of either
      Algorithm~\ref{alg:cls} or Algorithm~\ref{alg:simpl};
    \item $\sigma$, the output of Algorithm~\ref{alg:scale}
    \end{enumerate}
    \smallskip
    
    \ENSURE A list
    $[(M_0^*, T_1^*, R_1^*), \dots, (M_{\ell-1}^*, T_{\ell}^*, R_{\ell}^*)]$. \smallskip
    \smallskip
    
    \FOR{$k := 1$ \TO $\ell$}
    \STATE{$M_{k-1}^* := M_{k-1}\sigma$}
    \STATE{\label{alg:tback:v}$v := \bigl((\delta^{b_k} \tau)\sigma\bigr)/t$, extracting
      $\delta^{b_k}$ from $T_k$
      \COMMENT{$=\varepsilon_*^{(b_k/q)+\mu}$}}
    \STATE{$T_k^* := [\,]$}
    \FORALL{$x_j' = \delta^{b_k}f_j \in T_k$}
    \STATE{\label{alg:tback:h}$h := (y_jf_j/x_j)\sigma$
    \COMMENT{$= \varepsilon_*^{d_j} (f_j\sigma)$}}
    \STATE{$T_k^* := T_k^* \circ [\dot{y_j} = v h]$
    \COMMENT{$=\varepsilon_*^{b_k/q + \mu + d_j} (f_j\sigma)$}}
    \ENDFOR
    \STATE{$R_k^* := [\,\dot{y_j} = 0 \mid x_j' = 0 \in R_k\,]$}
    \ENDFOR
    \RETURN{$[(M_0^*, T_1^*, R_1^*), \dots, (M_{\ell-1}^*, T_{\ell}^*, R_{\ell}^*)]$}
  \end{algorithmic}
\end{algorithm}
\FloatBarrier

Our back-transformation is realized in Algorithm~\ref{alg:tback}. In
l.\ref{alg:tback:v} we compute the time scale factor
$\varepsilon_*^{(b_k/q) + \mu}$ for $T_k^*$ as described above, and in l.\ref{alg:tback:h}
we compute its co-factor $\varepsilon_*^{d_j} f\sigma$ as $(y_jf/x_j)\sigma$.

\section{The Big Picture}\label{se:picture}
Let us discuss what has been gained in $(M_0^*, T_1^*, R_1^*)$, \dots,
$(M_{\ell-1}^*, T_{\ell}^*, R_{\ell}^*)$ for our original system $S$ in \eqref{origeq}.
We are faced with a discrepancy. On the one hand, we fix
$\varepsilon = \varepsilon_*$. On the other hand, the requirement that $\wdelta$ be sufficiently
small in Theorem~\ref{th:ct} entails that $\varepsilon$ be sufficiently small. It is of
crucial importance whether invariant manifolds of \eqref{paramscaleeq}, which do
exist for sufficiently small $\varepsilon$, persist at
$\varepsilon=\varepsilon_*$. We are not aware of any algorithmic results addressing this question.
In particular, singular perturbation theory is typically concerned with
asymptotic results, which are not helpful here.

In case of persistence, there exist nested invariant manifolds $\mN_{k-1}^*$
which are Hausdorff-close to $\mM_{k-1}^*$ for system~\eqref{origeq}. Moreover,
the differential equations $T_k^*$ associated with $\mM_{k-1}^*$ correspond to
the $k$th level in a hierarchy of time scales and approximate the flow on
$\mN_{k-1}^*$. We have achieved a decomposition of \eqref{origeq} into $\ell$
systems of smaller dimension. At the very least, one obtains a well-educated
guess about possible candidates for invariant manifolds and reductions. For the
investigation of those candidates one may check the $\mN_k^*$ for approximate
invariance using, e.g., numerical methods, or by applying criteria proposed in
\cite{NoethenWalcher:09}.

Algorithm~\ref{alg:tmr} 
\begin{algorithm}[t]
  \renewcommand{\algorithmiccomment}[1]{%
    \hfill\parbox{\widthof{$= [(M_0^*, T_1^*, R_1^*), \dots, (sM_{\ell-1}^*, T_\ell^*,
        R_\ell^*)]$}}{\textcolor{gray}{#1}}}
  \caption{$\operatorname{TropicalMultiReduce}$\label{alg:tmr}}
  \begin{algorithmic}[1]
    \REQUIRE
    \begin{enumerate}[nolistsep]
    \item
    A list $S = [\dot{y_1} = f_1, \dots, \dot{y_n} = f_n]$ of
    autonomous first-order ordinary differential equations where $f_1$, \dots,
    $f_n \in \QQ[y_1, \dots, y_n]$;
    \item $\varepsilon_* \in (0,1) \cap \QQ$;
    \item $p \in \NN \setminus \{0\}$
    \end{enumerate}
    \smallskip

    \ENSURE
    A list $[(M_0^*, T_1^*, R_1^*), \dots, (M_{\ell-1}^*, T_\ell^*, R_\ell^*)]$ of triplets
    where $\ell \in \{2, \dots, m\}$, or the empty list. For
    $k \in \{1, \dots, \ell\}$, $M_{k-1}^*$ is a list of real constraints defining
    $\mM_{k-1}^* \subseteq \RR^n$; $T_k^*$ is a list of differential equations;
    $R_k^*$ is a list of trivial differential equations $\dot{y} = 0$ for all
    differential variables from $T_{k+1}^*$, \dots,~$T_m^*$.\smallskip

    The relevance of the output in terms of the input is discussed in
    Sect.~\ref{se:tb}.\smallskip

    \STATE{$\operatorname{TropicalC}_{S, \varepsilon_*, p} :=
      \operatorname{curry}(\operatorname{TropicalC}, S, \varepsilon_*, p)$
      \COMMENT{$\operatorname{TropicalC}_{S, \varepsilon_*, p}$ is a binary function}}
    \STATE{$\operatorname{TropicalD}_{S, \varepsilon_*, p} :=
      \operatorname{curry}(\operatorname{TropicalD}, S, \varepsilon_*, p)$
      \COMMENT{$\operatorname{TropicalD}_{S, \varepsilon_*, p}$ is a constant function}}
    \STATE{$T, P, \sigma := \operatorname{ScaleAndTruncate}(S,
      \operatorname{TropicalC}_{S, \varepsilon_*, p}, \operatorname{TropicalD}_{S, \varepsilon_*, p}, \varepsilon_*)$}
    \STATE{$\Sigma := \operatorname{ComputeReducedSystems}(T, P)$
      \COMMENT{$= [(M_0, T_1, R_1), \dots, (M_{\ell-1}, T_\ell, R_\ell)]$}}
    \STATE{$\Sigma' := \operatorname{SimplifyReducedSystems}(\Sigma)$
      \COMMENT{$= [(M_0', T_1', R_1), \dots, (M_{\ell-1}', T_\ell', R_\ell)]$}}
    \STATE{$\Sigma^* := \operatorname{TransformBack}(\Sigma', \sigma)$
      \COMMENT{$= [(M_0^*, T_1^*, R_1^*), \dots, (M_{\ell-1}^*, T_\ell^*, R_\ell^*)]$}}
    \RETURN{$\Sigma^*$}
  \end{algorithmic}
\end{algorithm}
provides a wrapper combining all our algorithms to decompose input systems like
\eqref{origeq} into several time scales. The underlying tropicalization is not
made explicit, and the result is presented on the original scale.
Figure~\ref{fig:graph}
\begin{figure}
  \small
  \centering
  
  \tikzset{
    arrowfill/.style={
      top color=orange!30,
      bottom color=orange!30,
    },
    arrowstyle/.style={
      arrowfill,
      single arrow,
      minimum height=#1,
      single arrow head extend=.3cm
    }
  }
  \begin{tikzpicture}[node distance=3.5cm]
    \tikzstyle{algo} = [rectangle, rounded corners, align=center,
    minimum width=4cm, minimum height=1cm, fill=orange!30]

    \tikzstyle{control} = [->]

    \node (TMR) [algo, minimum width = 8cm] {
      \textbf{Algorithm~\ref{alg:tmr}}\\ TropicalMultiReduce};

    \node (TMR0) [minimum width = 7.5cm, minimum height=1cm] {};

    \node (SAT) [algo, xshift=-2cm, yshift=-1.5cm, below of=TMR] {
      \textbf{Algorithm~\ref{alg:scale}}\\\ ScaleAndTruncate};
    \node (CRS) [algo, below of=SAT, xshift=1.25cm, yshift=.5cm] {
      \textbf{Algorithm~\ref{alg:cls}}\\ ComputeReducedSystems};
    \node (SRS) [algo, below of=CRS, yshift=.5cm] {
      \textbf{Algorithm~\ref{alg:simpl}}\\ SimplifyReducedSystems};
    \node (TRB) [algo, below of=SRS, xshift=-1.25cm, yshift=.5cm] {
      \textbf{Algorithm~\ref{alg:tback}}\\ TransformBack};
    \node (TE) [algo, left of=SAT, xshift=-3.5cm, yshift=.25cm] {
      \textbf{Algorithm~\ref{alg:tropicalize}}\\ TropicalEquilibration};
    \node (TD) [algo, above of=TE, yshift=-2cm, xshift=1cm] {
      \textbf{Algorithm~\ref{alg:tropicald}}\\ $\operatorname{TropicalD}$};
    \node (TD0) [below=-1cm of TD, minimum height=1cm, minimum width=3.5cm] {};
    \node (TC) [algo, above of=TD, yshift=-2cm, xshift=1cm] {
      \textbf{Algorithm~\ref{alg:tropicalc}}\\ $\operatorname{TropicalC}$};
    \node (TC0) [below=-1cm of TC, minimum height=1cm, minimum width=3.5cm] {};
    \node (SDC)  [algo, left of=CRS, xshift=-3.75cm, yshift=1.5cm] {
      \textbf{Algorithm~\ref{alg:sdc}}\\ $\operatorname{TestSmoothness}$};
    \node (SDC0) [below=-1cm of SDC, minimum height=1cm, minimum width=3.5cm] {};
    \node (HHA) [algo, left of=CRS, xshift=-4.75cm] {
      \textbf{Algorithm~\ref{alg:ha}}\\ IsHyperbolicallyAttractive};

    \draw[control] (TMR0.south east) |- (SAT);
    \draw[control] (TMR0.south east) |- (CRS);
    \draw[control] (TMR0.south east) |- (SRS);
    \draw[control] (TMR0.south east) |- (TRB);
    \draw[control] (SAT.west) -| (TC0.south east);
    \draw[control] (SAT.west) -| (TD0.south east);
    \draw[control] (TD) to (TE);
    \draw[control] (CRS.west) to (HHA.east);
    \draw[control] (CRS.west) -| (SDC0.south east);

    \node (A0) [arrowstyle=1.5cm, left=0.25cm of TMR]{\phantom{M}};

    \node (A1) [arrowstyle=3.5cm, rotate=-90, yshift=0.25cm, below=-3.05cm of SAT]{\phantom{M}};
    \node [arrowstyle=2cm, rotate=180, right=0.25cm of TC.east, anchor=east] {
      \phantom{$S$, $\varepsilon_0$, $p$}
    };
    \node [right=0.25cm of TC, minimum width=2cm, anchor=west] {
      {$S$, $\varepsilon_*$, $p$}
    };
    \node [arrowstyle=2cm, rotate=180, right=1.25cm of TD.east, anchor=east] {
      \phantom{$S$, $\varepsilon_0$, $p$}
    };
    \node [right=1.25cm of TD, minimum width=2cm, anchor=west] {
      {$S$, $\varepsilon_*$, $p$}
    };

    \node (A2) [arrowstyle=1.5cm, rotate=-90, yshift=0.25cm, below=0.9cm of SAT]{\phantom{M}};
    \node (A21) [arrowstyle=2.25cm, rotate=180, right=.25cm of SDC.east,
    anchor=east]{\begin{turn}{180}$P_1, \dots, P_m$\end{turn}
       };
    \node (A22) [arrowstyle=7.5cm, rotate=-90, yshift=-1.2cm, below=3.8775cm of
    SAT]{\rlap{\phantom{M}}\begin{rotate}{90}\hskip1pt$\sigma$\end{rotate}};
    \node (A3) [arrowstyle=1.5cm, rotate=-90, yshift=-1cm, below=0.9cm of CRS]{\phantom{M}};
    \node (A4) [arrowstyle=1.5cm, rotate=-90, yshift=-1cm, below=0.9cm of SRS]{\phantom{M}};
    \node (A5) [arrowstyle=1.5cm, rotate=180, left=0.25cm of TRB, anchor=west]{\phantom{M}};

    \tikzstyle{data} = [align=left, text width=4.85cm];

    \node [left=0pt of A0, anchor=east, text width=5cm] {
      \begin{enumerate}
      \item $S = [\dot{y_1} = g_1$, \dots,~$\dot{y_n} = g_n]$\\
      \item $\varepsilon_*$
      \item $p$
      \end{enumerate}
    };

    \node [right=0.3 cm of A1.center, data] {
      \begin{enumerate}[1.]
      \item $S$
      \item $\operatorname{TropicalC}_{S, \varepsilon_*, p}$\\
      \item $\operatorname{TropicalD}_{S, \varepsilon_*, p}$\\
      \item $\varepsilon_*$
      \end{enumerate}
    };

    \node [right=0.65 cm of A2.center, data] {
      $T_1, \dots, T_m$,\\
      where $T_k$ is $z_k = \delta^{b_k} \wf_k(z, 0)$\\
      and $b_1 < \dots < b_m$
    };
    
    \node [right=0.65 cm of A3.center, data] {
      $(M_0, T_1, R_1) \ha \cdots \ha (M_{\ell-1}, T_\ell, R_\ell)$
    };

    \node [right=0.65 cm of A4.center, data] {
      $(M_0', T_1', R_1) \ha \cdots \ha (M_{\ell-1}', T_\ell', R_\ell)$
    };

    \node [left=0pt of A5.east, anchor=east, text width=5cm] {
      $(M_0^*, T_1^*, R_1^*) \ha \cdots \ha (M_{\ell-1}^*, T_\ell^*, R_\ell^*)$
    };
  \end{tikzpicture}

  \caption{Functional dependencies (thin arrows) and principal data flow
    (thick arrows) between our algorithms\label{fig:graph}}
\end{figure}
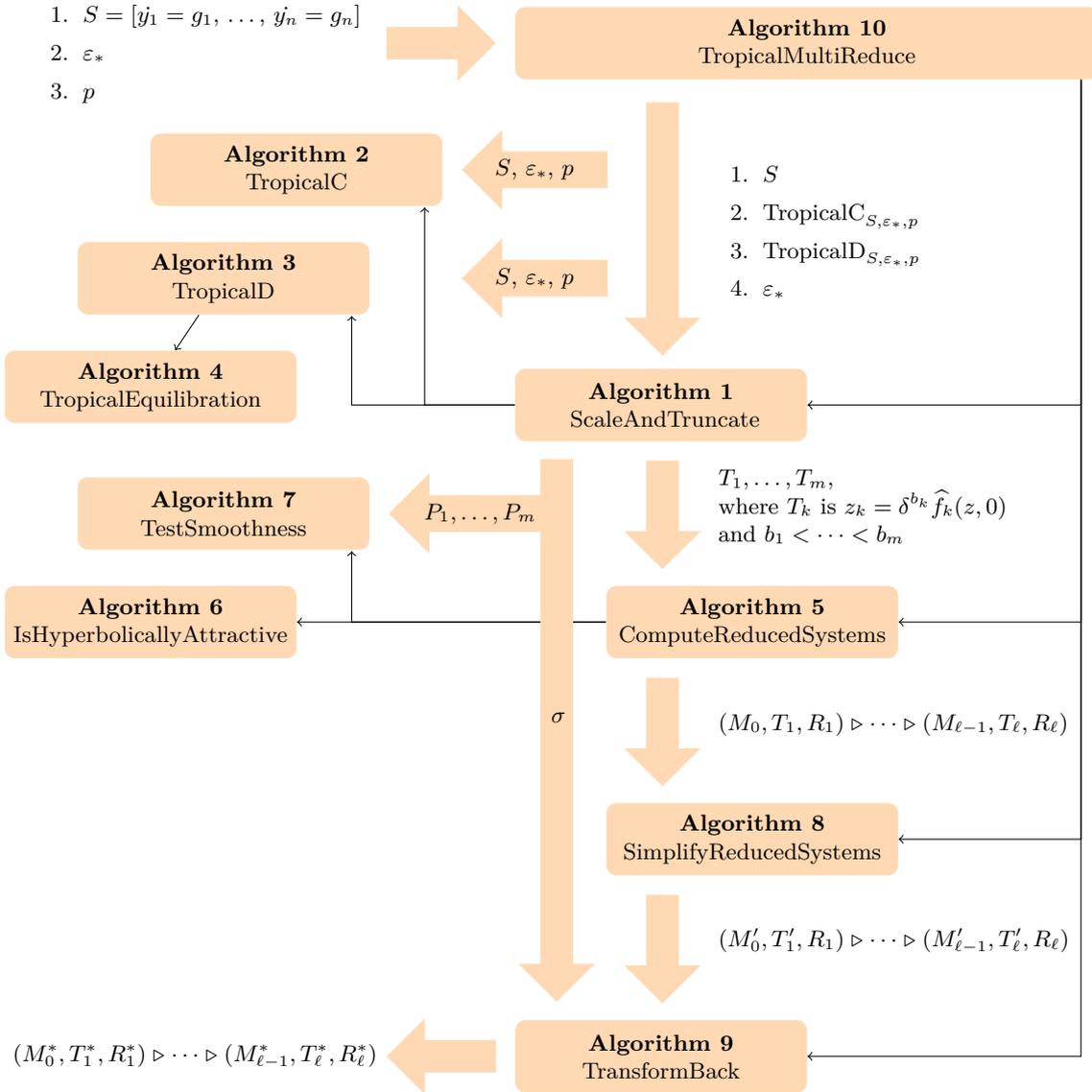
explains the functional dependencies and principal data flow between our
algorithms graphically.

\section{Computational Examples}\label{se:compex}

Based on our explicit algorithms in the present work, we have developed two
independent software prototypes realizing all methods described here. The first
one is in Python using SymPy \cite{sympy:17} for symbolic computation, pySMT
\cite{pysmt:15} as an interface to the SMT solver MathSAT5~\cite{mathsat}, and
SMTcut for the computation of tropical equilibrations~\cite{Lueders:20a}. The
second one is a Maple package, which makes use of Maple's built-in SMTLIB
package \cite{Forrest:17} for using the SMT solver Z3 \cite{z3}. For our
computations here we have used our Python code. Computation results are
identical with both systems, and timings are similar. We have conducted our
computations on a standard desktop computer with a 3.3~GHz 6-core Intel 5820K
CPU and 16~GB of main memory. Computation times listed are CPU times.

In the next section, we discuss in detail the computations for one specific
biological input system from the BioModels database, a repository of
mathematical models of biological processes \cite{le2006biomodels}. The
subsequent sections showcase several further such examples in a more concise
style. The focus here is on biological results. For an illustration of our
algorithms, we discuss in Appendix~\ref{app:examples} examples where reduction
stops at $\ell < m$ for various reasons.

\subsection{An Epidemic Model of Avian Influenza H5N6}
We consider \biomd{716}, which is related to the transmission dynamics of
subtype H5N6 of the influenza A virus in the Philippines in August 2017
\cite{LeeLao:18}. The model specifies four species: \textit{Susceptible birds}
(\texttt{S_b}), \textit{Infected birds} (\texttt{I_b}), \textit{Susceptible
  humans} (\texttt{S_h}), and \textit{Infected humans} (\texttt{I_a}), the
concentrations of which over time we map to differential variables $y_1$,
\dots,~$y_4$, respectively. The input system is given by
\begin{align*}
S = \bigl[ \textstyle{\diff{}{t} y_{1}} &= \textstyle{- \frac{9137}{2635182} y_{1} y_{2} - \frac{1}{730} y_{1} + \frac{412}{73}},  \\
\textstyle{\diff{}{t} y_{2}} &= \textstyle{\frac{9137}{2635182} y_{1} y_{2} - \frac{4652377}{961841430} y_{2}},  \\
\textstyle{\diff{}{t} y_{3}} &= \textstyle{- \frac{1}{6159375000} y_{2} y_{3} - \frac{1}{25258} y_{3} + \frac{40758549}{3650000}},  \\
\textstyle{\diff{}{t} y_{4}} &= \textstyle{\frac{1}{6159375000} y_{2} y_{3} - \frac{112500173}{2841525000000} y_{4}} \bigr].
\end{align*}

We choose $\varepsilon_* = \frac{1}{5}$, $p = 1$, and Algorithm~\ref{alg:tropicald}
non-deterministically selects $D = (-1, -4, -7, -3)$ from the tropical
equilibration. Algorithm~\ref{alg:scale} then yields the following scaled and
truncated system with three time scales:
\begin{align*}
T_{1} = \bigl[ \textstyle{\diff{}{\tau} x_{1}} &= \textstyle{1 \cdot \left(- \frac{5710625}{2635182} x_{1} x_{2} + \frac{412}{365}\right)} \bigr],  \\
T_{2} = \bigl[ \textstyle{\diff{}{\tau} x_{2}} &= \textstyle{\delta^{3} \cdot \left(\frac{5710625}{2635182} x_{1} x_{2} - \frac{116309425}{192368286} x_{2}\right)} \bigr],  \\
T_{3} = \bigl[ \textstyle{\diff{}{\tau} x_{3}} &= \textstyle{\delta^{6} \cdot \left(- \frac{15625}{25258} x_{3} + \frac{40758549}{18250000}\right)}, \\
\textstyle{\diff{}{\tau} x_{4}} &= \textstyle{\delta^{6} \cdot \left(\frac{15625}{15768} x_{2} x_{3} - \frac{112500173}{181857600} x_{4}\right)} \bigr].
\end{align*}
Notice that the lexicographic order of the differential variables is
coincidence. From this input, Algorithm~\ref{alg:cls} produces the following
reduced systems:
\begin{align*}
  M_{0} = \bigl[ &  \  \bigr],  
  & T_{1} = \bigl[ & \textstyle{\diff{}{\tau} x_{1}} = \textstyle{1 \cdot \left(- \frac{5710625}{2635182} x_{1} x_{2} + \frac{412}{365}\right)} \bigr],  \\
  && R_{1} = \bigl[ & \textstyle{\diff{}{\tau} x_{2}} = 0, \\
  &&& \textstyle{\diff{}{\tau} x_{3}} = 0, \\
  &&& \textstyle{\diff{}{\tau} x_{4}} = 0 \bigr],  \\[1ex]
  M_{1} = \bigl[ & \textstyle{2084378125 x_{1} x_{2} - 1085694984 = 0} \bigr], 
  & T_{2} = \bigl[ & \textstyle{\diff{}{\tau} x_{2}} = \textstyle{\delta^{3} \cdot \left(\frac{5710625}{2635182} x_{1} x_{2} - \frac{116309425}{192368286} x_{2}\right)} \bigr],  \\
  && R_{2}  = \bigl[ & \textstyle{\diff{}{\tau} x_{3}} = 0, \\
  &&& \textstyle{\diff{}{\tau} x_{4}} = 0 \bigr],  \\[1ex]
  M_{2} = \bigl[ & \textstyle{2084378125 x_{1} x_{2} - 1085694984 = 0}, 
  & T_{3} = \bigl[ & \textstyle{\diff{}{\tau} x_{3}} = \textstyle{\delta^{6} \cdot \left(- \frac{15625}{25258} x_{3} + \frac{40758549}{18250000}\right)}, \\
  & \textstyle{16675025 x_{1} x_{2} - 4652377 x_{2}} = 0 \bigr],  
  && \textstyle{\diff{}{\tau} x_{4}} = \textstyle{\delta^{6} \cdot \left(\frac{15625}{15768} x_{2} x_{3} - \frac{112500173}{181857600} x_{4}\right)} \bigr], \\
  && R_{3} = \bigl[ & \  \bigr].
\intertext{%
In that course, Algorithm~\ref{alg:ha} confirms hyperbolic attractivity
according to Sect.~\ref{se:ha} for all three scaled systems. 
Furthermore, Algorithm~\ref{alg:sdc}
applies the sufficient smoothness test from Sect.~\ref{se:diff} with
\begin{displaymath}
  \textstyle
  \ell=3,\quad b_1=3,\quad b_2=3,\quad 
  P_1 = 1 \cdot (-\delta^4 \cdot \frac{125}{146}x_1),\quad
  P_2 = \delta^6 \cdot (-\delta^4 \cdot \frac{15625}{15768} x_2 x_3).
\end{displaymath}
This yields $E = \{ 4 \}$, where $4$ cannot be expressed as an integer multiple of
$3$. Thus the test fails, which causes a warning in Algorithm~\ref{alg:cls}.
\endgraf
Algebraic simplification through Algorithm~\ref{alg:simpl} yields
the simplified reduced systems
}
  M_{0}' = \bigl[ &  \  \bigr],
  & T_{1}' = \bigl[ & \textstyle{\diff{}{\tau} x_{1}} = \textstyle{1 \cdot \left(- \frac{5710625}{2635182} x_{1} x_{2} + \frac{412}{365}\right)} \bigr],  \\
  && R_{1}' = \bigl[ & \textstyle{\diff{}{\tau} x_{2}} = 0, \\
  &&& \textstyle{\diff{}{\tau} x_{3}} = 0, \\
  &&& \textstyle{\diff{}{\tau} x_{4}} = 0 \bigr],  \\[1ex]
  M_{1}' = \bigl[ & \textstyle{x_{1}x_{2}} = \textstyle{\frac{1085694984}{2084378125}} \bigr], 
  & T_{2}' = \bigl[ & \textstyle{\diff{}{\tau} x_{2}} = \textstyle{\delta^{3} \cdot \left(- \frac{116309425}{192368286} x_{2} + \frac{412}{365}\right)} \bigr],  \\
  && R_{2}'  = \bigl[ & \textstyle{\diff{}{\tau} x_{3}} = 0, \\
  &&& \textstyle{\diff{}{\tau} x_{4}} = 0 \bigr],  \\[1ex]
  M_{2}' = \bigl[ & \textstyle{x_{1}} = \textstyle{\frac{4652377}{16675025}}, 
  & T_{3}' = \bigl[ & \textstyle{\diff{}{\tau} x_{3}} = \textstyle{\delta^{6} \cdot \left(- \frac{15625}{25258} x_{3} + \frac{40758549}{18250000}\right)}, \\
  & \textstyle{x_{2}} = \textstyle{\frac{1085694984}{581547125}} \bigr],  
  && \textstyle{\diff{}{\tau} x_{4}} = \textstyle{\delta^{6} \cdot \left(\frac{1884887125}{1018870563} x_{3} - \frac{112500173}{181857600} x_{4}\right)} \bigr], \\
  && R_{3}' = \bigl[ & \  \bigr].
\intertext{%
  Notice that our implementations conveniently rewrite equational constraints as
  monomial equations with numerical right hand sides when possible. This supports
  readability but is not essential for the simplifications applied here, which
  are based on Gröbner basis theory. Comparing $T_2'$ with $T_2$, we see that
  the equation for $x_1x_2$ in $M_1'$ is plugged in. Similarly, $M_2$ is
  simplified to $M_2'$, which is in turn used to reduce $T_3$ to $T_3'$.
  \endgraf
The back-transformed reduced systems as computed by Algorithm~\ref{alg:tback} read as follows:
}
  M_{0}^* = \bigl[ &  \  \bigr],  
  & T_{1}^* = \bigl[ & \textstyle{\diff{}{t} y_{1}} = \textstyle{1 \cdot \left(- \frac{9137}{2635182} y_{1} y_{2} + \frac{412}{73}\right)} \bigr],  \\
  && R_{1}^* = \bigl[ & \textstyle{\diff{}{t} y_{2}} = 0, \\
  &&& \textstyle{\diff{}{t} y_{3}} = 0, \\
  &&& \textstyle{\diff{}{t} y_{4}} = 0 \bigr],  \\[1ex]
  M_{1}^* = \bigl[ & \textstyle{y_{1}y_{2}} = \textstyle{\frac{1085694984}{667001}} \bigr], 
  & T_{2}^* = \bigl[ & \textstyle{\diff{}{t} y_{2}} = \textstyle{\frac{1}{125} \cdot \left(- \frac{116309425}{192368286} y_{2} + \frac{51500}{73}\right)} \bigr],  \\
  && R_{2}^*  = \bigl[ & \textstyle{\diff{}{t} y_{3}} = 0, \\
  &&& \textstyle{\diff{}{t} y_{4}} = 0 \bigr],  \\[1ex]
  M_{2}^* = \bigl[ & \textstyle{y_{1}} = \textstyle{\frac{4652377}{3335005}}, 
  & T_{3}^* = \bigl[ & \textstyle{\diff{}{t} y_{3}} = \textstyle{\frac{1}{15625} \cdot \left(- \frac{15625}{25258} y_{3} + \frac{203792745}{1168}\right)}, \\
  & \textstyle{y_{2}} = \textstyle{\frac{5428474920}{4652377}} \bigr],  
  && \textstyle{\diff{}{t} y_{4}} = \textstyle{\frac{1}{15625} \cdot \left(\frac{15079097}{5094352815} y_{3} - \frac{112500173}{181857600} y_{4}\right)} \bigr], \\
  && R_{3}^* = \bigl[ & \  \bigr].
\end{align*}
We compare $T_1^*$, \dots,~$T_3^*$ to the input system $S$: In the equation for
$\diff{}{t}y_1$, the monomial in $y_1$ is identified as a higher order term with
respect to $\delta$ and discarded by Algorithm~\ref{alg:scale}. In the equation for
$\diff{}{t}y_2$, the monomial in $y_1 y_2$ has been Gröbner-reduced to a
constant modulo the defining equation in $M_1'$. Similarly, the equation for
$\diff{}{t}y_3$ loses its monomial in $y_2 y_3$ by truncation of higher order
terms, and in the equation for $\diff{}{t}y_4$, the monomial in $y_2 y_3$ is
Gröbner-reduced to a monomial in $y_3$.

Notice the explicit constant factors on the right hand sides of the differential
equations in $T_1^*$, \dots,~$T_3^*$. They originate from factors $\delta^{b_k}$ in the
respective scaled systems $T_1$, \dots,~$T_3$, corresponding to
\eqref{fullscalesimplifiedint}. They are left explicit to make the time scale of
the differential equations apparent. We see that the system $T_2^* \circ R_2^*$ is
125 times slower than $T_1^* \circ R_1^*$, and $T_3^* \circ R_3^*$ is another 125 times
slower.

Figure~\ref{fig:dirfields} 
\begin{figure}
  \includegraphics[width=0.49\textwidth]{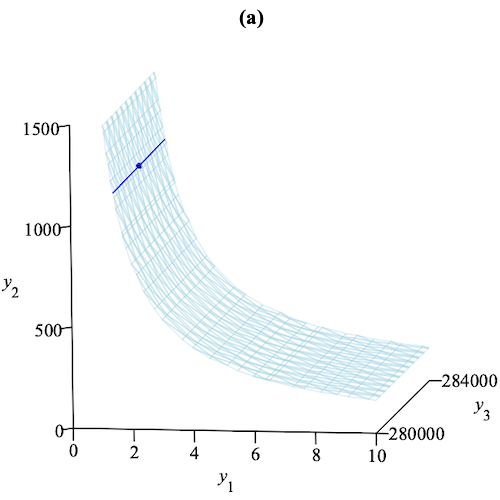}\hfill
  \includegraphics[width=0.49\textwidth]{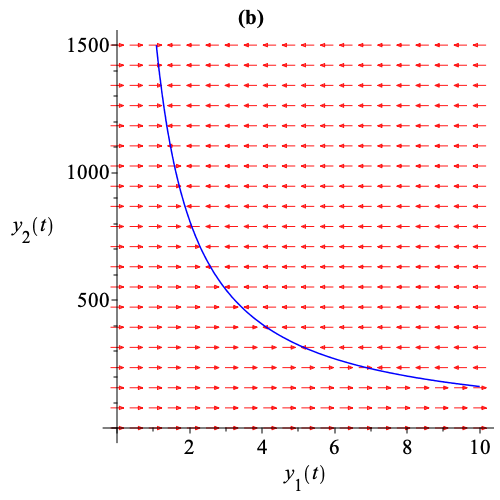}\\[1ex]
  \includegraphics[width=0.49\textwidth]{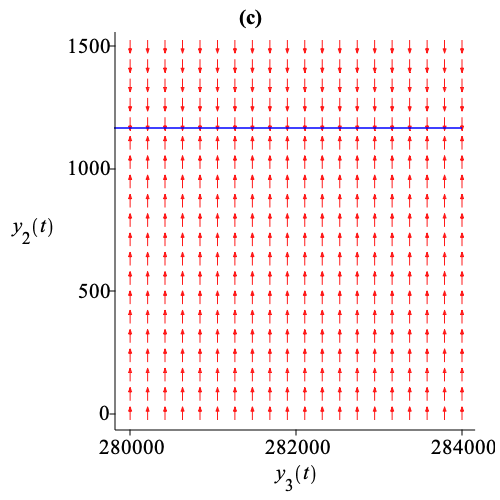}\hfill
  \includegraphics[width=0.49\textwidth]{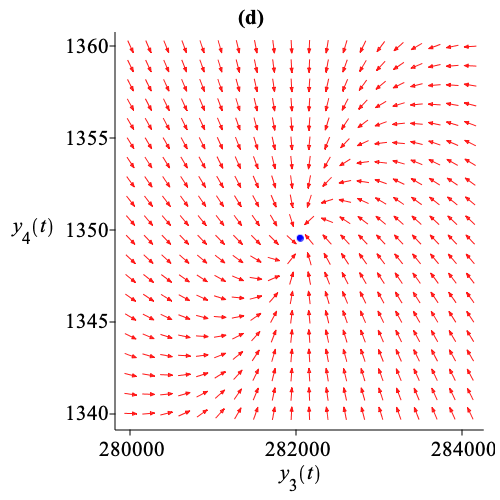}\\
  \caption{Critical manifolds and direction fields of our reductions of \biomd{716}.
    \textbf{(a)}~The surface is the critical manifold
    $\mM_1^* \subseteq \mM_0^* = \mU$ projected from $\RR^4$ into real
    $(y_1, y_2, y_3)$-space. The line located at
    $(y_1, y_2) \approx (1.4, 1166.8)$ is the critical submanifold
    $\mM_2^* \subseteq \mM_1^*$. The dot located at
    $(y_1, y_2, y_3) \approx (1.4, 1166.8, 282049.2)$ is the critical submanifold
    $\mM_3^* \subseteq \mM_2^*$. Both $\mM_1^*$ and $\mM_2^*$ extend to
    $\pm \infty$ in both $y_3$ and $y_4$ direction, and $\mM_3^*$ is located near
    $(1.4, 1166.8, 282049.2, 1349.6)$.\ \textbf{(b)}~The direction field of
    $T_1^*\circ R_1^*$ on $\mM_0^*=\mU$ projected from $\RR^4$ into real
    $(y_1, y_2)$-space. The curve is the critical submanifold
    $\mM_1^* \subseteq \mM_0^*$.\ \textbf{(c)}~The direction field of
    $T_2^*\circ R_2^*$ on $\mM_1^*$ projected from $\RR^4$ into real
    $(y_3, y_2)$-space. The line is the critical submanifold
    $\mM_2^* \subseteq \mM_1^*$. The system here is slower than the one in (b) by a
    factor of $125$.\ \textbf{(d)}~The direction field of $T_3^*\circ R_3^*$ on
    $\mM_2^*$ projected from $\RR^4$ into real $(y_3, y_4)$-space. The dot is the
    critical submanifold $\mM_3^* \subseteq \mM_2^*$. The system here is slower than the
    one in (c) by another factor of $125$.\label{fig:dirfields}}
\end{figure}
visualizes the direction fields of $T_1^* \circ R_1^*$,
\dots,~$T_3^* \circ R_3^*$ on their respective manifolds $\mM_0^*$,
\dots,~$\mM_2^*$ along with their respective critical manifolds $\mM_1^*$, \dots,
$\mM_3^*$, where $\mM_3^*$ can be derived from $\mM_2^*$ by additionally
equating the vector field of $T_3^* \circ R_3^*$ to zero:
\begin{displaymath}
  \textstyle M_3^* = \bigl[
    y_1 = \frac{4652377}{3335005}, \quad
    y_2 = \frac{5428474920}{4652377}, \quad
    y_3 = \frac{7051228977}{25000}, \quad
    y_4 = \frac{441466240042010928888}{327120760850763125}
    \bigr].
\end{displaymath}
This list $M_3^*$ does not explicitly occur in the output. However, its preimage
$M_3$ is constructed in Algorithm~\ref{alg:cls} and justifies the presence of
$(M_2, T_3, R_3)$ in the output there. The total computation time was 0.906~s.

This multiple-time scale reduction of the bird flu model emphasizes a cascade of
successive relaxations of model variables. First, the population of susceptible
birds relaxes. This relaxation is illustrated in Fig.~\ref{fig:dirfields}(b).
Then, the population of infected birds relaxes as shown in
Fig.~\ref{fig:dirfields}(c). Finally, the populations of susceptible and
infected humans relax to a stable steady state as shown in
Fig.~\ref{fig:dirfields}(d), following a reduced dynamics described by $T_3^*$.

\subsection{TGF-$\beta$ Pathway}
\biomd{101} is a simple representation of the TGF-$\beta$ signaling pathway that
plays a central role in tissue homeostasis and morphogenesis, as well as in
numerous diseases such as fibrosis and cancer \cite{VilarJansen:06}.
Concentrations over time of species \textit{Receptor~1} (\texttt{RI}),
\textit{Receptor~2} (\texttt{RII}), \textit{Ligand receptor complex-plasma
  membrane} (\texttt{lRIRII}), \textit{Ligand receptor complex-endosome}
(\texttt{lRIRII_endo}), \textit{Receptor~1 endosome} (\texttt{RI_endo}), and
\textit{Receptor~2 endosome} (\texttt{RII_endo}), are mapped to differential
variables $y_1$, \dots,~$y_6$, respectively. The original \biomd{101} has a change
of \texttt{ligand} concentration at time $t=2500$. For our computation here, we
ignore this discrete event. The input system is given by
\begin{align*}
S = \bigl[ \textstyle{\diff{}{t} y_{1}} &= \textstyle{- \frac{1}{100} y_{1} y_{2} - \frac{90277}{250000} y_{1} + \frac{33333}{1000000} y_{4} + \frac{33333}{1000000} y_{5} + 8},  \\
\textstyle{\diff{}{t} y_{2}} &= \textstyle{- \frac{1}{100} y_{1} y_{2} - \frac{90277}{250000} y_{2} + \frac{33333}{1000000} y_{4} + \frac{33333}{1000000} y_{6} + 4},  \\
\textstyle{\diff{}{t} y_{3}} &= \textstyle{\frac{1}{100} y_{1} y_{2} - \frac{152777}{250000} y_{3}},  \\
\textstyle{\diff{}{t} y_{4}} &= \textstyle{\frac{33333}{100000} y_{3} - \frac{33333}{1000000} y_{4}},  \\
\textstyle{\diff{}{t} y_{5}} &= \textstyle{\frac{33333}{100000} y_{1} - \frac{33333}{1000000} y_{5}},  \\
\textstyle{\diff{}{t} y_{6}} &= \textstyle{\frac{33333}{100000} y_{2} - \frac{33333}{1000000} y_{6}} \bigr].
\end{align*}

We choose $\varepsilon_* = \frac{1}{5}$, $p = 1$, and select
$D = (0, -4, -1, -2, -1, -5)$ from the tropical equilibrium. Our
back-transformed reduced systems read as follows:
\begin{align*}
  M_{0}^* = \bigl[ &  \  \bigr],  
  & T_{1}^* = \bigl[ & \textstyle{\diff{}{t} y_{1}} = \textstyle{5 \cdot \left(- \frac{1}{500} y_{1} y_{2} + \frac{8}{5}\right)} \bigr],  \\
  && R_{1}^* = \bigl[ & \textstyle{\diff{}{t} y_{2}} = 0, 
  \quad \textstyle{\diff{}{t} y_{3}} = 0, 
  \quad \textstyle{\diff{}{t} y_{4}} = 0, 
  \quad \textstyle{\diff{}{t} y_{5}} = 0, \\
  &&& \textstyle{\diff{}{t} y_{6}} = 0 \bigr],  \\[1ex]
  M_{1}^* = \bigl[ & \textstyle{y_{1}y_{2}} = \textstyle{800} \bigr], 
  & T_{2}^* = \bigl[ & \textstyle{\diff{}{t} y_{3}} = \textstyle{1 \cdot \left(- \frac{152777}{250000} y_{3} + 8\right)} \bigr],  \\
  && R_{2}^* = \bigl[ & \textstyle{\diff{}{t} y_{2}} = 0, 
  \quad \textstyle{\diff{}{t} y_{4}} = 0, 
  \quad \textstyle{\diff{}{t} y_{5}} = 0, 
  \quad \textstyle{\diff{}{t} y_{6}} = 0 \bigr],  \\[1ex]
  M_{2}^* = \bigl[ & \textstyle{y_{1}y_{2}} = \textstyle{800}, 
  & T_{3}^* = \bigl[ & \textstyle{\diff{}{t} y_{2}} = \textstyle{\frac{1}{5} \cdot \left(- \frac{90277}{50000} y_{2} + \frac{33333}{200000} y_{6}\right)} \bigr], \\
  & \textstyle{y_{3}} = \textstyle{\frac{2000000}{152777}} \bigr],  
  & R_{3}^* = \bigl[ & \textstyle{\diff{}{t} y_{4}} = 0, 
  \quad \textstyle{\diff{}{t} y_{5}} = 0, 
  \quad \textstyle{\diff{}{t} y_{6}} = 0 \bigr].
\end{align*}
The total computation time was 0.906~s.

The multiple-time scale reduction of the TGF-$\beta$ model emphasizes a cascade of
successive relaxations of concentrations of different species. First, the
concentration of receptor~1 relaxes rapidly. Then follows the membrane complex,
and, even slower, the relaxation of receptor~2.

\subsection{Caspase Activation Pathway}
\biomd{102} is a quantitative kinetic model that examines the intrinsic pathway
of caspase activation that is essential for apoptosis induction by various
stimuli including cytotoxic stress \cite{Legewie:06}. Species concentrations
over time are mapped to differential variables $y_1$, \dots,~$y_{13}$ as described
in Table \ref{tab:bm102}.
\begin{table}
  \centering
  \caption{Mapping of species concentrations to differential variables for
    \biomd{102}\label{tab:bm102}}

  \begin{tabular}{lll}
    \hline
    Species & Species variable & Differential variable \\
    \hline
      \textit{APAF-1}                                & \texttt{A}         & $y_{1}$  \\
      \textit{Caspase 9}                             & \texttt{C9}        & $y_{2}$  \\
      \textit{Caspase 9-XIAP complex}                & \texttt{C9X}       & $y_{3}$  \\
      \textit{XIAP}                                  & \texttt{X}         & $y_{4}$  \\
      \textit{APAF-1-Caspase 9-XIAP complex}         & \texttt{AC9X}      & $y_{5}$  \\
      \textit{APAF-1-Caspase 9 complex}              & \texttt{AC9}       & $y_{6}$  \\
      \textit{Caspase 3}                             & \texttt{C3}        & $y_{7}$  \\
      \textit{Caspase 3 cleaved}                     & \texttt{C3_star}   & $y_{8}$  \\
      \textit{Caspase 3 cleaved-XIAP complex}        & \texttt{C3_starX}  & $y_{9}$  \\
      \textit{Caspase 9 cleaved-XIAP complex}        & \texttt{C9_starX}  & $y_{10}$ \\
      \textit{Caspase 9 cleaved}                     & \texttt{C9_star}   & $y_{11}$ \\
      \textit{APAF-1-Caspase 9 cleaved complex}      & \texttt{AC9_star}  & $y_{12}$ \\
      \textit{APAF-1-Caspase 9 cleaved-XIAP complex} & \texttt{AC9_starX} & $y_{13}$ \\
    \hline
  \end{tabular}
\end{table}
The input system is given by
\begin{align*}
S = \bigl[ \phantom{{}_3} \textstyle{\diff{}{t} y_{1}} &= \textstyle{- \frac{1}{500} y_{1} y_{2} - \frac{1}{500} y_{1} y_{3} - \frac{1}{500} y_{1} y_{10} - \frac{1}{500} y_{1} y_{11} - \frac{1}{1000} y_{1} + \frac{1}{10} y_{5} + \frac{1}{10} y_{6} + \frac{1}{10} y_{12}}  \\
 & \quad\quad \textstyle{{} + \frac{1}{10} y_{13} + \frac{1}{50}},  \\
\textstyle{\diff{}{t} y_{2}} &= \textstyle{- \frac{1}{500} y_{1} y_{2} - \frac{1}{1000} y_{2} y_{4} - \frac{1}{5000} y_{2} y_{8} - \frac{1}{1000} y_{2} + \frac{1}{1000} y_{3} + \frac{1}{10} y_{6} + \frac{1}{50}},  \\
\textstyle{\diff{}{t} y_{3}} &= \textstyle{- \frac{1}{500} y_{1} y_{3} + \frac{1}{1000} y_{2} y_{4} - \frac{1}{500} y_{3} + \frac{1}{10} y_{5}},  \\
\textstyle{\diff{}{t} y_{4}} &= \textstyle{- \frac{1}{1000} y_{2} y_{4} + \frac{1}{1000} y_{3} - \frac{1}{1000} y_{4} y_{6} - \frac{3}{1000} y_{4} y_{8} - \frac{1}{1000} y_{4} y_{11} - \frac{1}{1000} y_{4} y_{12} - \frac{1}{1000} y_{4}}  \\
 & \quad\quad \textstyle{{} + \frac{1}{1000} y_{5} + \frac{1}{1000} y_{9} + \frac{1}{1000} y_{10} + \frac{1}{1000} y_{13} + \frac{1}{25}},  \\
\textstyle{\diff{}{t} y_{5}} &= \textstyle{\frac{1}{500} y_{1} y_{3} + \frac{1}{1000} y_{4} y_{6} - \frac{51}{500} y_{5}},  \\
\textstyle{\diff{}{t} y_{6}} &= \textstyle{\frac{1}{500} y_{1} y_{2} - \frac{1}{1000} y_{4} y_{6} + \frac{1}{1000} y_{5} - \frac{1}{5000} y_{6} y_{8} - \frac{101}{1000} y_{6}},  \\
\textstyle{\diff{}{t} y_{7}} &= \textstyle{- \frac{1}{200000} y_{2} y_{7} - \frac{7}{20000} y_{6} y_{7} - \frac{1}{20000} y_{7} y_{11} - \frac{7}{2000} y_{7} y_{12} - \frac{1}{1000} y_{7} + \frac{1}{5}},  \\
\textstyle{\diff{}{t} y_{8}} &= \textstyle{\frac{1}{200000} y_{2} y_{7} - \frac{3}{1000} y_{4} y_{8} + \frac{7}{20000} y_{6} y_{7} + \frac{1}{20000} y_{7} y_{11} + \frac{7}{2000} y_{7} y_{12} - \frac{1}{1000} y_{8} + \frac{1}{1000} y_{9}},  \\
\textstyle{\diff{}{t} y_{9}} &= \textstyle{\frac{3}{1000} y_{4} y_{8} - \frac{1}{500} y_{9}},  \\
\textstyle{\diff{}{t} y_{10}} &= \textstyle{- \frac{1}{500} y_{1} y_{10} + \frac{1}{1000} y_{4} y_{11} - \frac{1}{500} y_{10} + \frac{1}{10} y_{13}},  \\
\textstyle{\diff{}{t} y_{11}} &= \textstyle{- \frac{1}{500} y_{1} y_{11} + \frac{1}{5000} y_{2} y_{8} - \frac{1}{1000} y_{4} y_{11} + \frac{1}{1000} y_{10} - \frac{1}{1000} y_{11} + \frac{1}{10} y_{12}},  \\
\textstyle{\diff{}{t} y_{12}} &= \textstyle{\frac{1}{500} y_{1} y_{11} - \frac{1}{1000} y_{4} y_{12} + \frac{1}{5000} y_{6} y_{8} - \frac{101}{1000} y_{12} + \frac{1}{1000} y_{13}},  \\
\textstyle{\diff{}{t} y_{13}} &= \textstyle{\frac{1}{500} y_{1} y_{10} + \frac{1}{1000} y_{4} y_{12} - \frac{51}{500} y_{13}} \bigr].
\end{align*}

We choose $\varepsilon_* = \frac{1}{2}$, $p = 1$, and select
$D = (-4, 2, 3, 5, 5, 4, -6, -8, -4, -2, -2, 0, 0)$ from the tropical
equilibration. Our back-transformed reduced systems read as follows:
\begin{align*}
M_{0}^* =
\bigl[ & \ \bigr],
& T_{1}^* = \bigl[ & \textstyle{\diff{}{t} y_{4}} = \textstyle{1 \cdot \left(- \frac{3}{1000} y_{4} y_{8} + \frac{1}{25}\right)} \bigr],  \\
&& R_{1}^* = \bigl[ & \textstyle{\diff{}{t} y_{5}} = 0, \quad \textstyle{\diff{}{t}
  y_{6}} = 0, \quad \textstyle{\diff{}{t} y_{12}} = 0,
\quad \textstyle{\diff{}{t} y_{13}} = 0, \\
&&& \textstyle{\diff{}{t} y_{2}} = 0, \quad \textstyle{\diff{}{t} y_{3}} = 0, \quad
\textstyle{\diff{}{t} y_{10}} = 0,
\quad \textstyle{\diff{}{t} y_{11}} = 0, \\
&&& \textstyle{\diff{}{t} y_{1}} = 0, \quad \textstyle{\diff{}{t} y_{7}} = 0, \quad
\textstyle{\diff{}{t} y_{9}} = 0,
\quad \textstyle{\diff{}{t} y_{8}} = 0 \bigr],  \\[1ex]
  M_{1}^* = \bigl[ & \textstyle{y_{4}y_{8}} = \textstyle{\frac{40}{3}} \bigr], 
  & T_{2}^* = \bigl[ & \textstyle{\diff{}{t} y_{5}} = \textstyle{\frac{1}{8} \cdot \left(\frac{2}{125} y_{1} y_{3} - \frac{102}{125} y_{5}\right)},  \\
  &&& \textstyle{\diff{}{t} y_{6}} = \textstyle{\frac{1}{8} \cdot \left(\frac{2}{125} y_{1} y_{2} - \frac{101}{125} y_{6}\right)},  \\
  &&& \textstyle{\diff{}{t} y_{12}} = \textstyle{\frac{1}{8} \cdot \left(\frac{2}{125} y_{1} y_{11} - \frac{101}{125} y_{12}\right)},  \\
  &&& \textstyle{\diff{}{t} y_{13}} = \textstyle{\frac{1}{8} \cdot \left(\frac{2}{125} y_{1} y_{10} - \frac{102}{125} y_{13}\right)} \bigr],  \\
  && R_{2}^* = \bigl[ & \textstyle{\diff{}{t} y_{2}} = 0, 
  \quad \textstyle{\diff{}{t} y_{3}} = 0, 
  \quad \textstyle{\diff{}{t} y_{10}} = 0, 
  \quad \textstyle{\diff{}{t} y_{11}} = 0, \\
  &&& \textstyle{\diff{}{t} y_{1}} = 0, 
  \quad \textstyle{\diff{}{t} y_{7}} = 0, 
  \quad \textstyle{\diff{}{t} y_{9}} = 0, 
  \quad \textstyle{\diff{}{t} y_{8}} = 0 \bigr],  \\[1ex]
  M_{2}^* = \bigl[ & \textstyle{y_{4}y_{8}} = \textstyle{\frac{40}{3}}, 
  & T_{3}^* = \bigl[ & \textstyle{\diff{}{t} y_{2}} = \textstyle{\frac{1}{16} \cdot \left(- \frac{2}{625} y_{2} y_{8} + \frac{8}{25}\right)} \bigr], \\
  & \textstyle{y_{1} y_{3} - 51 y_{5}} = 0,  
  & R_{3}^* = \bigl[ & \textstyle{\diff{}{t} y_{3}} = 0, 
  \quad \textstyle{\diff{}{t} y_{10}} = 0, 
  \quad \textstyle{\diff{}{t} y_{11}} = 0, 
  \quad \textstyle{\diff{}{t} y_{1}} = 0, \\
  & \textstyle{2 y_{1} y_{2} - 101 y_{6}} = 0,  
  && \textstyle{\diff{}{t} y_{7}} = 0, 
  \quad \textstyle{\diff{}{t} y_{9}} = 0,
  \quad \textstyle{\diff{}{t} y_{8}} = 0 \bigr]. \\
  & \textstyle{2 y_{1} y_{11} - 101 y_{12}} = 0, \\
  & \textstyle{y_{1} y_{10} - 51 y_{13}} = 0 \bigr],
\end{align*}
The total computation time was 8.547~s, of which Algorithm~\ref{alg:te} took 6.188~s.

The multiple-time scale reduction of the caspase activation model emphasizes a
cascade of successive relaxations. First, the inhibitor of apoptosis XIAP binds
rapidly to the cleaved caspase. Then, the four APAF complexes are formed.
Finally, the Caspase~9 is recruited to the apoptosome.

\subsection{Avian Influenza Bird-to-Human Transmission}
\biomd{709} describes bird-to-human transmission of different strains of avian
influenza A viruses, such as H5N1 and H7N9 \cite{Liu:17}. Species concentrations
over time of \textit{Susceptible avians} (\texttt{S_a}), \textit{Infected
  avians} (\texttt{I_a}), \textit{Susceptible humans} (\texttt{S_h}),
\textit{Infected humans} (\texttt{I_h}), and \textit{Recovered humans}
(\texttt{R_h}) are mapped to differential variables $y_1$, \dots,~$y_5$,
respectively. The input system is given by
\begin{align*}
S = \bigl[ & \textstyle{\diff{}{t} y_{1}} = \textstyle{- \frac{1}{8000000000} y_{1}^{3} + \frac{127}{20000000} y_{1}^{2} - \frac{9}{500000000} y_{1} y_{2} - \frac{1}{200} y_{1}},  \\
&\textstyle{\diff{}{t} y_{2}} = \textstyle{\frac{9}{500000000} y_{1} y_{2} - \frac{37123}{50000000} y_{2}},  \\
&\textstyle{\diff{}{t} y_{3}} = \textstyle{- \frac{3}{500000000} y_{2} y_{3} - \frac{391}{10000000} y_{3} + 30},  \\
&\textstyle{\diff{}{t} y_{4}} = \textstyle{\frac{3}{500000000} y_{2} y_{3} - \frac{4445391}{10000000} y_{4}},  \\
&\textstyle{\diff{}{t} y_{5}} = \textstyle{\frac{1}{10} y_{4} - \frac{391}{10000000} y_{5}} \bigr].
\end{align*}

We choose $\varepsilon_* = \frac{1}{5}$, $p = 1$, and select
$D = (-7, 0, -8, 3, -2)$ from the tropical equilibration. Our back-transformed
reduced systems read as follows:
\begin{align*}
  M^*_{0} &= \bigl[   \  \bigr],
            &T^*_{1} = \bigl[&  \textstyle{\diff{}{t} y_{1}} = 1 \cdot \left(\textstyle{- \frac{1}{8000000000} y_{1}^{3} + \frac{127}{20000000} y_{1}^{2}}\right) \bigr],  \\
        && R_{1}^* = \bigl[&  \textstyle{\diff{}{t} y_{4}} = 0,
        \quad \textstyle{\diff{}{t} y_{2}} = 0,
        \quad \textstyle{\diff{}{t} y_{3}} = 0,
        \quad \textstyle{\diff{}{t} y_{5}} = 0 \bigr],  \\[1ex]
 M^*_{1} &=  \bigl[   \textstyle{y_{1}^{3} - 50800 y_{1}^{2}} = 0
          \bigr],
            & T^*_{2} = \bigl[&  \textstyle{\diff{}{t} y_{4}} = \textstyle{\frac{1}{5}\cdot\left(\frac{3}{100000000} y_{2} y_{3} - \frac{4445391}{2000000} y_{4}\right)} \bigr],  \\
        && R_{2}^* = \bigl[&  \textstyle{\diff{}{t} y_{2}} = 0,
        \quad \textstyle{\diff{}{t} y_{3}} = 0,
        \quad \textstyle{\diff{}{t} y_{5}} = 0 \bigr].
\end{align*}
The total computation time was 0.578~s.

The multiple-time scale reduction of this avian influenza model emphasizes a
cascade of successive relaxations of different model variables. First, the
susceptible bird population relaxes rapidly. The reduced equation $T_1$ and
manifold $M_1$ suggest that the bird population dynamics is of the Allee type
and evolves toward the stable extinct state. It follows the relaxation of
infected human population that also evolves toward the extinct state, the end of
the epidemics.

\section{Some Remarks on Complexity}\label{se:complexity}
A detailed complexity analysis of our approach is beyond the scope of this
article. We collect some remarks on the asymptotic worst-case complexity of
various computational steps required by our algorithms:
\begin{enumerate}[(i)]
\item\label{item:tarski} the existential decision problem over real closed
  fields, for which we use SMT solving over \texttt{QF\_NRA} in 
  Algorithm~\ref{alg:ha};
\item\label{item:lp} linear programming, for which we used SMT solving over
  \texttt{QF\_LRA} in Algorithm~\ref{alg:tropicald};
\item\label{item:ilp} integer linear programming, for which we used SMT solving
  over \texttt{QF\_LIA} in Algorithm~\ref{alg:sdc}.
\end{enumerate}
Problem (\ref{item:tarski}) is in single exponential time \cite{Grigoriev:88a}.
Problem (\ref{item:lp}) is in polynomial \cite{Khachiyan:79a} time. Problem
(\ref{item:ilp}) is NP-complete; a proof can be found in \cite[Theorem
18.1]{Schrijver86}, where it is essentially attributed to \cite{cook71}. With
all these problems, the dominant complexity parameter is the number of
variables, which in our context corresponds to the number $n$ of differential
variables. Biological models impose reasonable upper bounds on $n$, which
corresponds to the number of species occurring there. When $n$ is bounded,
problem (\ref{item:tarski}) becomes polynomial \cite{Grigoriev:88a}. The same
holds for problem (\ref{item:ilp}); a proof based on Lenstra's algorithm
\cite{Lenstra:83} can be found in \cite[Corollary 18.7a]{Schrijver86}.

When characterizing the decision in l.\ref{td:sat} of
Algorithm~\ref{alg:tropicald} as linear programming in (\ref{item:lp}) above, we
tacitly assume that it is not applied to the tropical equilibration $\Pi$ as a
whole but to each contained polyhedron independently. The number of contained
polyhedra is in turn exponential in $n$, in the worst case. The motivation for
the computation of a disjunctive normal form in Algorithm~\ref{alg:tropicalize}
is to support the choice of a suitable point $(d_1, \dots, d_n)$ in
Algorithm~\ref{alg:tropicald} by providing information on the geometry of the
tropical equilibration as a union of polyhedra. It is possible to omit this and
accordingly drop the computation of the disjunctive normal form in
l.\ref{te:dnf} of Algorithm~\ref{alg:tropicalize} altogether. In that case,
solving in l.\ref{td:sat} of Algorithm~\ref{alg:tropicald} is applied to the
formula $\bigwedge \bigvee \bigwedge P$ instead of $\Pi$. This is a linear decision problem over ordered
fields, which is NP-complete \cite{GathenSieveking:76a}, and solutions can be
found in single exponential time \cite{Weispfenning:88}.

Efficient implementations do not necessarily use theoretically optimal
algorithms. At the time of writing, SMT solving and corresponding decision
procedures are a very active field of research and thus subject to constant
change. We deliberately refrain from going into details about the current state
of the art at this point.

Beyond the examples and computation times discussed here, we currently have no
systematic empirical data that would allow substantial statements on the
practical performance and limits of our implementations.

\section{Concluding Remarks}\label{se:conclusion}

We provided a symbolic method for automated model reduction of biological
networks described by ordinary differential equations with multiple time scales.
Our method is applicable to systems with two time scales or more, superseding
traditional slow-fast reduction methods that can cope with\linebreak only two
time scales. We also proposed, for the first time, the automatic verification of
hyperbolicity conditions required for the validity of the reduction. Our
theoretical framework is accompanied by rigorous algorithms and prototypical
implementations, which we successfully applied to real-world problems from the
BioModels database \cite{le2006biomodels}.

We would like to list some open points and possible extensions of our research
here. Our reduction algorithm is based on a fixed scaling
\eqref{fullscalesimplifiedint} leading to a fixed ordering of the time scales of
different variables. In our reduction scheme, different variables relax
hierarchically, first the fastest ones, then the second fastest, and finally the
slowest ones, which justifies our geometric picture of nested invariant
manifolds. However, there are situations, e.g.~in models of relaxation
oscillations, when the ordering of time scales changes with time: variables that
were fast can become slow at a later time, and vice versa. In order to cope with
such situations, one would like to use different scalings for different time
segments. One attempt to implement such a procedure has been provided in
\cite{sommer2016hybrid}.

Although our proposed method identifies the full hierarchy of time scales, the
subsequent reduction may stop early in this hierarchy when hyperbolic
attractivity is not satisfied at some stage. One possible reason is the presence
of conservation laws, also known as first integrals, at the given reduction
stage. Such conservation laws force an eigenvalue zero for the Jacobian. A
theorem by Schneider and Wilhelm \cite{SchneiderWilhelm:00a} can be employed to
reduce such a setting to the hyperbolically attractive case. As for the behavior
of first integrals when proceeding to the reduced system, see the discussion of
the non-standard case in \cite{lawa} for two time scales; an extension to
multiple time scales should be straightforward. Work in progress is concerned
with the introduction of new slow variables, one for each independent
conservation law of the fast subsystem. This is applied to networks with
multiple time scales and approximate linear and polynomial conservation laws.

More generally, it is of interest to consider cases when hyperbolic attractivity
fails but hyperbolicity still holds: In such cases, Cardin and Teixeira show
there still exist invariant manifolds \cite{cartex}. Testing for hyperbolicity
is more involved than testing for hyperbolic attractivity, but in theory it is
well understood, and there exists an algorithmic approach due to Routh
\cite{gantmacher}. In the case of hyperbolicity, but not attractivity, the
ensuing global dynamics may be quite interesting; for instance slow-fast cycles
may appear.

Concerning differentiability requirements, we check in Sect.~\ref{se:diff} for
smoothness of the full system. However, Fenichel's results, and in principle
also those by Cardin and Teixeira, require only sufficient finite
differentiability. Therefore, given a differential equation system and a
scaling, invariant manifolds and corresponding reduced systems exist for $C^p$
functions with fixed $p<\infty$. Going through the details will involve intricate
analysis that is left to future work.

In the introduction we sketched a Michaelis--Menten system abstracting from the
known numerical values for the reaction rate constants $k_1$, $k_{-1}$, $k_2$.
It would be indeed interesting to work on such parametric data. In the presence
of parameters, one would consider effective quantifier elimination over real
closed fields
\cite{CollinsHong:91,Weispfenning:97b,DolzmannSturm:98a,Kosta:16a,Sturm:17a,Sturm:ISSAC2018}
as a generalization of SMT solving. Robust implementations are freely available
\cite{Brown:03a,DolzmannSturm:97a} and well supported. They have been
successfully applied to problems in chemical reaction network theory during the
past decade
\cite{SturmWeber:08a,SturmWeber:09a,WeberSturm:11a,ErramiSeiler:11a,ErramiEiswirth:15a,BradfordDavenport:17c,EnglandErrami:17b,BradfordDavenport:19a,Grigoriev2020,Seiler2020}.
Such a generalization is not quite straightforward. With the tropical scaling in
Sect.~\ref{se:tropscale}, Algorithm~\ref{alg:tropicalc} would introduce
logarithms of polynomials in the parametric coefficients, which is not
compatible with the logic framework used here. Similar tropicalization methods,
which are unfortunately not compatible with our abstract view on scaling in
Sect.~\ref{se:scale}, require only logarithms of individual parametric
coefficients \cite{sgfwr}. Such a more special form would allow the use of
abstraction in the logic engine.

From a point of view of user-oriented software, it would be most desirable to
develop automatic strategies for determining good values for $\varepsilon_*$ and for
choices of $(d_1, \dots, d_n)$ from the tropical equilibration in
Algorithm~\ref{alg:tropicald}.

\section*{Acknowledgments}
We are grateful to the anonymous referees for their thorough reviews and
numerous helpful comments.

This work has been supported by the interdisciplinary bilateral project
ANR-17-CE40-0036 and DFG-391322026 SYMBIONT
\cite{BoulierFages:18b,BoulierFages:18a}. The SYMBIONT project owes much to the
interdisciplinary perspective, the scientific vision, and the enthusiasm of our
colleague Andreas Weber, who unexpectedly passed away in March 2020. Andreas was
also a major proponent of and contributor to tropical methods for the reduction
of reaction networks, which are continued and extended in the present work. We
miss him as an extraordinary scientist and person.

\providecommand{\noop}[1]{}

\appendix
\section{Further Computational Examples}\label{app:examples}

Recall that in Sect.~\ref{se:compex} we have discussed computations for several
systems from the BioModels database \cite{le2006biomodels}. While the focus
there was on biological results, we discuss here examples where reduction stops
at $\ell < m$ for various reasons. Again, we have conducted our computations on a
standard desktop computer with a 3.3~GHz 6-core Intel 5820K CPU and 16~GB of
main memory. Computation times listed are CPU times.

\subsection{Hypertoxiticy of a Painkiller}
\biomd{609} describes the metabolism and the related hepatotoxicity of
acetaminophen, a pain killer \cite{Reddyhoff:15}. The species concentrations
over time of \textit{Sulphate PAPS}, \textit{GSH}, \textit{NAPQI},
\textit{Paracetamol APAP}, and \textit{Protein adducts} are mapped to
differential variables $y_1$, \dots,~$y_5$, respectively. The input system is given
by
\begin{align*}
S = \bigl[ & \textstyle{\diff{}{t} y_{1}} = \textstyle{- 226000000000000 y_{1} y_{4} - 2 y_{1} + \frac{53}{2000000000000000}},  \\
&\textstyle{\diff{}{t} y_{2}} = \textstyle{- 1600000000000000000 y_{2} y_{3} - 2 y_{2} + \frac{687}{50000000000000000}},  \\
&\textstyle{\diff{}{t} y_{3}} = \textstyle{- 1600000000000000000 y_{2} y_{3} - \frac{220063}{2000} y_{3} + \frac{63}{200} y_{4}},  \\
&\textstyle{\diff{}{t} y_{4}} = \textstyle{- 226000000000000 y_{1} y_{4} + \frac{63}{2000} y_{3} - \frac{661}{200} y_{4}},  \\
&\textstyle{\diff{}{t} y_{5}} = \textstyle{110y_{3}} \bigr].
\end{align*}
Since there is only one monomial on the right hand side of the equation for
$\diff{}{t}y_5$, equilibration is impossible. This causes Algorithm~\ref{alg:te}
to return in l.\ref{te:dnf} a disjunctive normal form $\Pi$ equivalent to
``false,'' which describes the empty set. Hence Algorithm~\ref{alg:tropicald}
returns $\bot$, and Algorithm~\ref{alg:scale} returns the empty list. The total
computation time was 0.006~s.

\subsection{Transmission Dynamics of Rabies}
\biomd{726} examines the transmission dynamics of rabies between dogs and humans
\cite{Ruan:17}. Species concentrations over time are mapped to differential
variables $y_1$, \dots,~$y_{8}$ as described in Table \ref{tab:bm726}.
\begin{table}
  \centering
  \caption{Mapping of species concentrations to differential variables for
    \biomd{726}\label{tab:bm726}}
  
  \begin{tabular}{lll}
    \hline
    Species                & Species variable  & Differential variable \\
    \hline
    \textit{Susceptible dogs}   & \texttt{S_d} & $y_1$ \\
    \textit{Exposed dogs}       & \texttt{E_d} & $y_2$ \\
    \textit{Infectious dogs}    & \texttt{I_d} & $y_3$ \\
    \textit{Recovered dogs}     & \texttt{R_d} & $y_4$ \\
    \textit{Susceptible humans} & \texttt{S_h} & $y_5$ \\
    \textit{Exposed humans}     & \texttt{E_h} & $y_6$ \\
    \textit{Infectious humans}  & \texttt{I_h} & $y_7$ \\
    \textit{Recovered humans}   & \texttt{R_h} & $y_8$ \\
    \hline
  \end{tabular}
\end{table}
The input system is given by
\begin{align*}
S = \bigl[ & \textstyle{\diff{}{t} y_{1}} = \textstyle{- \frac{79}{500000000} y_{1} y_{3} - \frac{17}{100} y_{1} + \frac{18}{5} y_{2} + y_{4} + 3000000},  \\
&\textstyle{\diff{}{t} y_{2}} = \textstyle{\frac{79}{500000000} y_{1} y_{3} - \frac{617}{100} y_{2}},  \\
&\textstyle{\diff{}{t} y_{3}} = \textstyle{\frac{12}{5} y_{2} - \frac{27}{25} y_{3}},  \\
&\textstyle{\diff{}{t} y_{4}} = \textstyle{\frac{9}{100} y_{1} + \frac{9}{100} y_{2} - \frac{27}{25} y_{4}},  \\
&\textstyle{\diff{}{t} y_{5}} = \textstyle{- \frac{229}{100000000000000} y_{3} y_{5} - \frac{3}{1000} y_{5} + \frac{18}{5} y_{6} + y_{8} + 15400000},  \\
&\textstyle{\diff{}{t} y_{6}} = \textstyle{\frac{229}{100000000000000} y_{3} y_{5} - \frac{6543}{1000} y_{6}},  \\
&\textstyle{\diff{}{t} y_{7}} = \textstyle{\frac{12}{5} y_{6} - \frac{1343}{1000} y_{7}},  \\
&\textstyle{\diff{}{t} y_{8}} = \textstyle{\frac{27}{50} y_{6} - \frac{1003}{1000} y_{8}} \bigr].
\end{align*}
We choose $\varepsilon_* = \frac{1}{5}$, $p = 1$, and Algorithm~\ref{alg:tropicald}
selects $D = (-10, -10, -11, -9, -14, -7, -8, -7)$ from the tropical
equilibration. Algorithm~\ref{alg:scale} then yields the following scaled and
truncated system with three time scales:
\begin{align*}
  T_{1} = \bigl[ 
    \textstyle{\diff{}{\tau} x_{1}} &= \textstyle{1 \cdot \left(-\frac{395}{256} x_{1} x_{3} + \frac{18}{25} x_2\right)},  \\
    \textstyle{\diff{}{\tau} x_{2}} &= \textstyle{1 \cdot \left(\frac{395}{256} x_{1} x_{3} - \frac{617}{500} x_2\right)},  \\
    \textstyle{\diff{}{\tau} x_{6}} &= \textstyle{1 \cdot \left(\frac{28625}{16384} x_{3} x_{5} - \frac{6543}{5000} x_6\right)} \bigr],  \\
  T_{2} = \bigl[ 
    \textstyle{\diff{}{\tau} x_{3}} &= \textstyle{\delta \cdot \left(\frac{12}{25} x_{2} - \frac{27}{25} x_{3}\right)},  \\
    \textstyle{\diff{}{\tau} x_{4}} &= \textstyle{\delta \cdot \left(\frac{9}{20} x_{1} + \frac{9}{20} x_{2} - \frac{27}{25} x_{4}\right)},  \\
    \textstyle{\diff{}{\tau} x_{7}} &= \textstyle{\delta \cdot \left(\frac{12}{25} x_{6} - \frac{1343}{1000} x_{7}\right)},  \\
    \textstyle{\diff{}{\tau} x_{8}} &= \textstyle{\delta \cdot \left(\frac{27}{50} x_{6} - \frac{1003}{1000} x_{8}\right)} \bigr],  \\
  T_{3} = \bigl[ 
    \textstyle{\diff{}{\tau} x_{5}} &= \textstyle{\delta^{5} \cdot \left(\frac{4928}{3125} - \frac{15}{8} x_{5}\right)} \bigr].
\end{align*}
Equating the right hand sides $F_1$ of the differential equations in $T_1$ to
zero equivalently yields
\[
M_1 = [ - 9875 x_{1} x_{3} + 4608 x_{2} = 0, \quad
49375 x_{1} x_{3} - 39488 x_{2} = 0, \quad
17890625 x_{3} x_{5} - 13400064 x_{6} = 0 ].
\] 
In l.\ref{ha:neimodels} of Algorithm~\ref{alg:ha}, $U \circ M_1$ is tested for
satisfiability. This fails, which means that the corresponding manifold $\mM_1$
is empty over the positive first orthant. Consequently,
Algorithms~\ref{alg:cls}, \ref{alg:simpl}, and \ref{alg:tback} return empty
lists. The total computation time was 0.921~s.

\subsection{Negative Feedback Loop Between Tumor Suppressor and Oncogene}
\biomd{156} describes the dynamics of a negative feedback loop between the tumor
suppressor protein p53 and the oncogene protein Mdm2 in human cells
\cite{Geva:06}. The species concentrations over time for \textit{P53}
(\texttt{x}), \textit{Mdm2} (\texttt{y}), and \textit{Precursor Mdm2}
(\texttt{y0}) are mapped to differential variables $y_1$, \dots,~$y_3$,
respectively. The input system is given by
\begin{align*}
S = \bigl[ & \textstyle{\diff{}{t} y_{1}} = \textstyle{- \frac{37}{10} y_{1} y_{2} + 2 y_{1}}, \\
& \textstyle{\diff{}{t} y_{2}} = \textstyle{- \frac{9}{10} y_{2} + \frac{11}{10} y_{3}}, \\
& \textstyle{\diff{}{t} y_{3}} = \textstyle{\frac{3}{2} y_{1} - \frac{11}{10} y_{3}} \bigr].
\end{align*}
We choose $\varepsilon_* = \frac{1}{2}$, $p = 1$, and select $D = (2, 1, 1)$.
Algorithm~\ref{alg:scale} then yields the following scaled and truncated system
with two time scales:
\begin{align*}
  T_{1} = \bigl[ 
    \textstyle{\diff{}{\tau} x_{1}} &= \textstyle{1 \cdot \left(-\frac{37}{40} x_{1} x_{2} + x_1\right)} \bigr],  \\
  T_{2} = \bigl[ 
    \textstyle{\diff{}{\tau} x_{2}} &= \textstyle{\delta \cdot \left(-\frac{9}{10} x_{2} + \frac{11}{10} x_{3}\right)},  \\
    \textstyle{\diff{}{\tau} x_{3}} &= \textstyle{\delta \cdot \left(\frac{3}{4} x_{1} - \frac{11}{10} x_{3}\right)} \bigr].
\end{align*}
Analogously to the previous section we obtain
$M_1 = [ - 37 x_{1} x_{2} + 40 x_{1} = 0 ]$, for which we find $\mM_1$ to be
non-empty over the positive first orthant in l.\ref{ha:neimodels} of
Algorithm~\ref{alg:ha}. However in l.\ref{ha:finally}, the test for hyperbolic
attractivity fails with $M_1$ and the Hurwitz conditions
\[
  \textstyle \Gamma = \left\{ \frac{37}{40}x_2 - 1 > 0 \right\},
\]
so that ``false'' is returned. Therefore, Algorithm~\ref{alg:cls} breaks the
for-loop in l.\ref{cm:break} with $k=1$ and returns the empty list in
l.\ref{cm:smalll}. Obviously, the simplified and back-translated systems are
empty lists as well. The total computation time was 0.453~s.

\subsection{CD4 T-Cells in the Spread of HIV}
\biomd{663} describes how CD4 T-cells can influence the spread of the HIV
infection \cite{Wodarz:07}. Species concentrations over time for
\textit{Infected T-cells} (\texttt{x}), \textit{Uninfected T-cells}
(\texttt{y}), and \textit{Free viruses} (\texttt{v}) are mapped to variables
$y_1$, \dots,~$y_3$, respectively. The input system is given by
\begin{align*}
S = \bigl[ & \textstyle{\diff{}{t} y_{1}} = \textstyle{- \frac{1}{10} y_{1}^{2} y_{3} - \frac{1}{10} y_{1} y_{2} y_{3} + \frac{4}{5} y_{1} y_{3} - \frac{1}{10} y_{1}},  \\
&\textstyle{\diff{}{t} y_{2}} = \textstyle{- \frac{1}{10} y_{1} y_{2} y_{3} + \frac{1}{5} y_{1} y_{3} - \frac{1}{10} y_{2}^{2} y_{3} + y_{2} y_{3} - \frac{1}{5} y_{2}},  \\
&\textstyle{\diff{}{t} y_{3}} = \textstyle{y_{2} - \frac{1}{2} y_{3}} \bigr].
\end{align*}
We choose $\varepsilon_* = \frac{1}{2}$, $p = 5$, and $D = (1, 4, 3)$. The choice of
$p=5$ causes fractional powers of numbers in the scaled and truncated system
\begin{align*}
  \textstyle T_{1} &= \textstyle \left[ \diff{}{\tau} x_{3} = 1 \cdot \left(x_{2} - x_{3}\right) \right], \\
  \textstyle T_{2} &= \textstyle \left[ \diff{}{\tau} x_{2} = \delta^{7} \cdot \left(\frac{4}{5} \sqrt[5]{4}\, x_{1} x_{3} - \frac{4}{5} \sqrt[5]{4}\, x_{2}\right) \right], \\
  \textstyle T_{3} &= \textstyle \left[ \diff{}{\tau} x_{1} = \delta^{12} \cdot \left(\frac{4}{5} \sqrt[5]{4}\, x_{1} x_{3} - \frac{4}{5} \sqrt[5]{4}\, x_{1}\right) \right].
\end{align*}
However, such input is not accepted with the SMT logic \texttt{QF\_NRA} used in
Algorithm \ref{alg:ha}. As discussed in Sect.~\ref{se:ha}, we catch the
corresponding error from the SMT solver, convert to floats, and restart, which
solves the problem.

Similarly to the previous example, the Hurwitz test in l.{20} of Algorithm
\ref{alg:ha} succeeds for $k=1$ but fails for $k=2$ in Algorithm \ref{alg:cls}.
Since there are fewer than two reduced systems, we return the empty list.
Consequently, the list of simplified reduced systems and the corresponding list
of back-transformed systems are empty as well. The total computation time was
0.390~s.


\begin{thebibliography}{10}

\bibitem{cvc4}
Clark Barrett, Christopher~L. Conway, Morgan Deters, Liana Hadarean, Dejan
  Jovanovi{\'c}, Tim King, Andrew Reynolds, and Cesare Tinelli.
\newblock {CVC4}.
\newblock In G.~Gopalakrishnan and S.~Qadeer, editors, {\em Proc. {CAV} 2011},
  volume 6806 of {\em LNCS}, pages 171--177. Springer, 2011.
\newblock \href {https://doi.org/10.1007/978-3-642-22110-1_14}
  {\path{doi:10.1007/978-3-642-22110-1_14}}.

\bibitem{BarrettFontaine:17a}
Clark Barrett, Pascal Fontaine, and Cesare Tinelli.
\newblock The {SMT-LIB} standard: Version 2.6.
\newblock Technical report, Department of Computer Science, The University of
  Iowa, 2017.

\bibitem{BeckerWeispfenning:93a}
Thomas Becker, Volker Weispfenning, and Heinz Kredel.
\newblock {\em Gr{\"o}bner Bases, a Computational Approach to Commutative
  Algebra}, volume 141 of {\em Graduate Texts in Mathematics}.
\newblock Springer, 1993.
\newblock \href {https://doi.org/10.1007/978-1-4612-0913-3}
  {\path{doi:10.1007/978-1-4612-0913-3}}.

\bibitem{bogart2007computing}
Tristram Bogart, Anders~Nedergaard Jensen, David Speyer, Bernd Sturmfels, and
  Rekha~R. Thomas.
\newblock Computing tropical varieties.
\newblock {\em J. Symb. Comput.}, 42(1--2):54--73, January--February 2007.
\newblock \href {https://doi.org/10.1016/j.jsc.2006.02.004}
  {\path{doi:10.1016/j.jsc.2006.02.004}}.

\bibitem{BoulierFages:18b}
Fran{\c c}ois Boulier, Fran{\c c}ois Fages, Ovidiu Radulescu, Satya~S. Samal,
  Aandreas Schuppert, Werner~M. Seiler, Thomas Sturm, Sebastian Walcher, and
  Anderas Weber.
\newblock The {SYMBIONT} project: Symbolic methods for biological networks.
\newblock {\em F1000Research}, 7(1341), August 2018.
\newblock \href {https://doi.org/10.7490/f1000research.1115995.1}
  {\path{doi:10.7490/f1000research.1115995.1}}.

\bibitem{BoulierFages:18a}
Fran{\c c}ois Boulier, Fran{\c c}ois Fages, Ovidiu Radulescu, Satya~S. Samal,
  Andreas Schuppert, Werner~M. Seiler, Thomas Sturm, Sebastian Walcher, and
  Andreas Weber.
\newblock The {SYMBIONT} project: Symbolic methods for biological networks.
\newblock {\em {ACM} Communications in Computer Algebra}, 52(3):67--70,
  December 2018.
\newblock \href {https://doi.org/10.1145/3313880.3313885}
  {\path{doi:10.1145/3313880.3313885}}.

\bibitem{BradfordDavenport:17c}
Russell Bradford, James~H. Davenport, Matthew England, Hassan Errami, Vladimir
  Gerdt, Dima Grigoriev, Charles Hoyt, Marek Ko{\v s}ta, Ovidiu Radulescu,
  Thomas Sturm, and Andreas Weber.
\newblock A case study on the parametric occurrence of multiple steady states.
\newblock In M.~Burr, editor, {\em Proc. {ISSAC} 2017}, pages 45--52. ACM,
  2017.
\newblock \href {https://doi.org/10.1145/3087604.3087622}
  {\path{doi:10.1145/3087604.3087622}}.

\bibitem{BradfordDavenport:19a}
Russell Bradford, James~H. Davenport, Matthew England, Hassan Errami, Vladimir
  Gerdt, Dima Grigoriev, Charles Hoyt, Marek Ko{\v s}ta, Ovidiu Radulescu,
  Thomas Sturm, and Andreas Weber.
\newblock Identifying the parametric occurrence of multiple steady states for
  some biological networks.
\newblock {\em J. Symb. Comput.}, 98:84--119, May--June 2020.
\newblock \href {https://doi.org/10.1016/j.jsc.2019.07.008}
  {\path{doi:10.1016/j.jsc.2019.07.008}}.

\bibitem{Brown:03a}
Christopher~W. Brown.
\newblock {QEPCAD B}: A program for computing with semi-algebraic sets using
  {CAD}s.
\newblock {\em ACM SIGSAM Bulletin}, 37(4):97--108, December 2003.
\newblock \href {https://doi.org/10.1145/968708.968710}
  {\path{doi:10.1145/968708.968710}}.

\bibitem{Buchberger:65a}
Bruno Buchberger.
\newblock {\em {E}in {A}lgorithmus zum {A}u{f}{f}inden der {B}asiselemente des
  {R}est\-klassen\-ringes nach einem nulldimensionalen {P}olynomideal}.
\newblock Doctoral dissertation, Mathematical Institute, University of
  Innsbruck, Innsbruck, Austria, 1965.

\bibitem{cartex}
Pedro~T. Cardin and Marco~A. Teixeira.
\newblock {Fenichel} theory for multiple time scale singular perturbation
  problems.
\newblock {\em {SIAM} J. Appl. Dyn. Syst.}, 16(3):1425--1452, January 2017.
\newblock \href {https://doi.org/10.1137/16m1067202}
  {\path{doi:10.1137/16m1067202}}.

\bibitem{cartexcorr}
Pedro~T. Cardin and Marco~A. Teixeira.
\newblock Corrigendum: {Fenichel} theory for multiple time scale singular
  perturbation problems.
\newblock {\em {SIAM} J. Appl. Dyn. Syst.}, 18(2):1223, January 2019.
\newblock \href {https://doi.org/10.1137/19M1241660}
  {\path{doi:10.1137/19M1241660}}.

\bibitem{mathsat}
Alessandro Cimatti, Alberto Griggio, Bastiaan Schaafsma, and Roberto
  Sebastiani.
\newblock The {MathSAT5} {SMT} solver.
\newblock In N.~Piterman and S.~A. Smolka, editors, {\em Proc. {TACAS} 2013},
  volume 7795 of {\em LNCS}, pages 93--107. Springer, 2013.
\newblock \href {https://doi.org/10.1007/978-3-642-36742-7_7}
  {\path{doi:10.1007/978-3-642-36742-7_7}}.

\bibitem{Collins:75}
George~E. Collins.
\newblock Quantifier elimination for the elementary theory of real closed
  fields by cylindrical algebraic decomposition.
\newblock In H.~Brakhage, editor, {\em Automata Theory and Formal Languages.
  2nd GI Conference}, volume~33 of {\em LNCS}, pages 134--183. Springer, 1975.
\newblock \href {https://doi.org/10.1007/3-540-07407-4_17}
  {\path{doi:10.1007/3-540-07407-4_17}}.

\bibitem{CollinsHong:91}
George~E. Collins and Hoon Hong.
\newblock Partial cylindrical algebraic decomposition for quantifier
  elimination.
\newblock {\em J. Symb. Comput.}, 12(3):299--328, September 1991.
\newblock \href {https://doi.org/10.1016/S0747-7171(08)80152-6}
  {\path{doi:10.1016/S0747-7171(08)80152-6}}.

\bibitem{cook71}
Stephen~A. Cook.
\newblock The complexity of theorem-proving procedures.
\newblock In {\em Proc. {STOC} '71}, pages 151--158. ACM Press, New York, NY,
  1971.
\newblock \href {https://doi.org/10.1145/800157.805047}
  {\path{doi:10.1145/800157.805047}}.

\bibitem{smtrat}
Florian Corzilius, Gereon Kremer, Sebastian Junges, Stefan Schupp, and Erika
  {\'A}brah{\'a}m.
\newblock {SMT-RAT}: An open source {C++} toolbox for strategic and parallel
  {SMT} solving.
\newblock In M.~Heule and S.~Weaver, editors, {\em Proc.~SAT 2015}, volume 9340
  of {\em LNCS}, pages 360--368. Springer, 2015.
\newblock \href {https://doi.org/10.1007/978-3-319-24318-4_26}
  {\path{doi:10.1007/978-3-319-24318-4_26}}.

\bibitem{CurryFeys:58}
Haskell~B. Curry and Robert Feys.
\newblock {\em Combinatory Logic}, volume~I of {\em Studies in Logic and the
  Foundations of Mathematics}.
\newblock North Holland Publishing Company, Amsterdam, The Netherlands, 1958.

\bibitem{DolzmannSturm:97a}
Andreas Dolzmann and Thomas Sturm.
\newblock Redlog: Computer algebra meets computer logic.
\newblock {\em ACM SIGSAM Bulletin}, 31(2):2--9, June 1997.
\newblock \href {https://doi.org/10.1145/261320.261324}
  {\path{doi:10.1145/261320.261324}}.

\bibitem{DolzmannSturm:98a}
Andreas Dolzmann, Thomas Sturm, and Volker Weispfenning.
\newblock Real quantifier elimination in practice.
\newblock In B.~H. Matzat, G.-M. Greuel, and G.~Hiss, editors, {\em Algorithmic
  Algebra and Number Theory}, pages 221--247. Springer, 1998.
\newblock \href {https://doi.org/10.1007/978-3-642-59932-3_11}
  {\path{doi:10.1007/978-3-642-59932-3_11}}.

\bibitem{EnglandErrami:17b}
Matthew England, Hassan Errami, Dima Grigoriev, Ovidiu Radulescu, Thomas Sturm,
  and Andreas Weber.
\newblock Symbolic versus numerical computation and visualization of parameter
  regions for multistationarity of biological networks.
\newblock In {\em Proc. {CASC} 2017}, volume 10490 of {\em LNCS}, pages
  93--108. Springer, 2017.
\newblock \href {https://doi.org/10.1007/978-3-319-66320-3_8}
  {\path{doi:10.1007/978-3-319-66320-3_8}}.

\bibitem{ErramiEiswirth:15a}
Hassan Errami, Markus Eiswirth, Dima Grigoriev, Werner~M. Seiler, Thomas Sturm,
  and Andreas Weber.
\newblock Detection of {Hopf} bifurcations in chemical reaction networks using
  convex coordinates.
\newblock {\em J. Comput. Phys.}, 291:279--302, June 2015.
\newblock \href {https://doi.org/10.1016/j.jcp.2015.02.050}
  {\path{doi:10.1016/j.jcp.2015.02.050}}.

\bibitem{ErramiSeiler:11a}
Hassan Errami, Werner~M. Seiler, Thomas Sturm, and Andreas Weber.
\newblock On {M}uldowney's criteria for polynomial vector fields with
  constraintsja.
\newblock In V.~P. Gerdt, W.~Koepf, E.~W. Mayr, and E.~V. Vorozhtsov, editors,
  {\em Proc. CASC 2011}, volume 6885 of {\em LNCS}, pages 135--143. Springer,
  2011.
\newblock \href {https://doi.org/10.1007/978-3-642-23568-9_11}
  {\path{doi:10.1007/978-3-642-23568-9_11}}.

\bibitem{Feinberg:19a}
Martin Feinberg.
\newblock {\em Foundations of Chemical Reaction Network Theory}, volume 202 of
  {\em Applied Mathematical Sciences}.
\newblock Springer, 2019.
\newblock \href {https://doi.org/10.1007/978-3-030-03858-8}
  {\path{doi:10.1007/978-3-030-03858-8}}.

\bibitem{fenichel}
Neil Fenichel.
\newblock Geometric singular perturbation theory for ordinary differential
  equations.
\newblock {\em J. Differ. Equations}, 31(1):53--98, January 1979.
\newblock \href {https://doi.org/10.1016/0022-0396(79)90152-9}
  {\path{doi:10.1016/0022-0396(79)90152-9}}.

\bibitem{Forrest:17}
Stephen Forrest.
\newblock Integration of {SMT-LIB} support into {Maple}.
\newblock In M.~England and V.~Ganesch, editors, {\em Proc. Satisfiability
  Checking and Symbolic Computation 2017}, volume 1974 of {\em {CEUR} Workshop
  Proceedings}. CEUR-WS, Kaiserslautern, Germany, July 2017.
\newblock \href {https://doi.org/10.1145/3055282.3055285}
  {\path{doi:10.1145/3055282.3055285}}.

\bibitem{gantmacher}
Feliks~R. Gantmacher.
\newblock {\em The Theory of Matrices. Vol. 2. Transl. from the Russian by K.
  A. Hirsch.}
\newblock Providence, RI: AMS Chelsea Publishing, reprint of the 1959
  translation edition, 1998.

\bibitem{pysmt:15}
Marco Gario and Andrea Micheli.
\newblock {PySMT}: A solver-agnostic library for fast prototyping of
  {SMT}-based algorithms.
\newblock In {\em {SMT} Workshop 2015. 13th International Workshop on
  Satisfiability Modulo Theories, Affiliated With the 27th International
  Conference on Computer Aided Verification}, San Francisco, CA, July 2015.

\bibitem{GathenSieveking:76a}
Joachim {\noop{Gathen}{von zur Gathen}} and Malte Sieveking.
\newblock Weitere zum {E}rf{\"u}llungsproblem polynomial {\"a}quivalente
  kombinatorische {A}ufgaben.
\newblock In {\em {K}omplexit{\"a}t von {E}ntscheidungsproblemen}, volume~43 of
  {\em LNCS}, chapter~4, pages 49--71. Springer, 1976.
\newblock \href {https://doi.org/10.1007/3-540-07805-3_5}
  {\path{doi:10.1007/3-540-07805-3_5}}.

\bibitem{Geva:06}
Naama Geva-Zatorsky, Nitzan Rosenfeld, Shalev Itzkovitz, Ron Milo, Alex Sigal,
  Erez Dekel, Talia Yarnitzky, Yuvalal Liron, Paz Polak, Galit Lahav, and Uri
  Alon.
\newblock Oscillations and variability in the p53 system.
\newblock {\em Mol. Syst. Biol.}, 2(1):2006--0033, June 2006.
\newblock \href {https://doi.org/doi.org/10.1038/msb4100068}
  {\path{doi:doi.org/10.1038/msb4100068}}.

\bibitem{gwz}
Alexandra Goeke, Sebastian Walcher, and Eva Zerz.
\newblock Determining ``small parameters'' for quasi-steady state.
\newblock {\em J. Differ. Equations}, 259(3):1149--1180, August 2015.
\newblock \href {https://doi.org/10.1016/j.jde.2015.02.038}
  {\path{doi:10.1016/j.jde.2015.02.038}}.

\bibitem{Grigoriev:88a}
Dima Grigoriev.
\newblock Complexity of deciding {T}arski algebra.
\newblock {\em J. Symb. Comput.}, 5(1--2):65--108, February--April 1988.
\newblock \href {https://doi.org/10.1016/S0747-7171(88)80006-3}
  {\path{doi:10.1016/S0747-7171(88)80006-3}}.

\bibitem{Grigoriev2020}
Dima Grigoriev, Alexandru Iosif, Hamid Rahkooy, Thomas Sturm, and Andreas
  Weber.
\newblock Efficiently and effectively recognizing toricity of steady state
  varieties.
\newblock {\em Math. Comput. Sci.}, 2020.
\newblock In press.
\newblock \href {https://doi.org/10.1007/s11786-020-00479-9}
  {\path{doi:10.1007/s11786-020-00479-9}}.

\bibitem{hta}
Frederick~G. Heineken, Henry~M. Tsuchiya, and Rutherford Aris.
\newblock On the mathematical status of the pseudo-steady state hypothesis of
  biochemical kinetics.
\newblock {\em Math. Biosci.}, 1(1):95--113, March 1967.
\newblock \href {https://doi.org/10.1016/0025-5564(67)90029-6}
  {\path{doi:10.1016/0025-5564(67)90029-6}}.

\bibitem{hopp}
Frank Hoppensteadt.
\newblock On systems of ordinary differential equations with several parameters
  multiplying the derivatives.
\newblock {\em J. Differ. Equations}, 5(1):106--116, January 1969.
\newblock \href {https://doi.org/10.1016/0022-0396(69)90106-5}
  {\path{doi:10.1016/0022-0396(69)90106-5}}.

\bibitem{Hurwitz:95a}
Adolf Hurwitz.
\newblock Ueber die {B}edingungen, unter welchen eine {G}leichung nur {W}urzeln
  mit negativen reellen {T}heilen besitzt.
\newblock {\em Math. Ann.}, 46:273--284, June 1895.
\newblock \href {https://doi.org/10.1007/BF01446812}
  {\path{doi:10.1007/BF01446812}}.

\bibitem{IEEE754}
{IEEE Std. 754-2019}, July 2019.
\newblock {IEEE} Standard for Floating-Point Arithmetic.
\newblock \href {https://doi.org/10.1109/IEEESTD.2019.8766229}
  {\path{doi:10.1109/IEEESTD.2019.8766229}}.

\bibitem{Khachiyan:79a}
Leonid~G. Khachiyan.
\newblock A polynomial algorithm in linear programming.
\newblock {\em Soviet Math. Dokl.}, 20(1):191--194, 1979.

\bibitem{Kosta:16a}
Marek Ko{\v s}ta.
\newblock {\em New Concepts for Real Quantifier Elimination by Virtual
  Substitution}.
\newblock Doctoral dissertation, Saarland University, Germany, December 2016.
\newblock \href {https://doi.org/10.22028/D291-26679}
  {\path{doi:10.22028/D291-26679}}.

\bibitem{krwa}
Niclas Kruff and Sebastian Walcher.
\newblock Coordinate-independent singular perturbation reduction for systems
  with three time scales.
\newblock {\em Math. Biosci. Eng.}, 16(5):5062--5091, June 2019.
\newblock \href {https://doi.org/10.3934/mbe.2019255}
  {\path{doi:10.3934/mbe.2019255}}.

\bibitem{lawa}
Christian Lax and Sebastian Walcher.
\newblock Singular perturbations and scaling.
\newblock {\em Discrete Cont. Dyn.-B}, 25(1):1--29, January 2020.
\newblock \href {https://doi.org/10.3934/dcdsb.2019170}
  {\path{doi:10.3934/dcdsb.2019170}}.

\bibitem{LeeLao:18}
Hanl Lee and Angelyn Lao.
\newblock Transmission dynamics and control strategies assessment of avian
  influenza {A} ({H5N6}) in the {Philippines}.
\newblock {\em Infectious Disease Modelling}, 3:35--59, March 2018.
\newblock \href {https://doi.org/10.1016/j.idm.2018.03.004}
  {\path{doi:10.1016/j.idm.2018.03.004}}.

\bibitem{Legewie:06}
Stefan Legewie, Nils Bl{\"u}thgen, and Hanspeter Herzel.
\newblock Mathematical modeling identifies inhibitors of apoptosis as mediators
  of positive feedback and bistability.
\newblock {\em PLoS Comput. Biol.}, 2(9):e120, September 2006.
\newblock \href {https://doi.org/10.1371/journal.pcbi.0020120}
  {\path{doi:10.1371/journal.pcbi.0020120}}.

\bibitem{Lenstra:83}
Hendrik~W. Lenstra, Jr.
\newblock Integer programming with a fixed number of variables.
\newblock {\em Math. Oper. Res.}, 8(4):538--548, November 1983.
\newblock \href {https://doi.org/10.1287/moor.8.4.538}
  {\path{doi:10.1287/moor.8.4.538}}.

\bibitem{litvinov2007maslov}
Grigori{\u\i}~L. Litvinov.
\newblock {Maslov} dequantization, idempotent and tropical mathematics: A brief
  introduction.
\newblock {\em J. Math. Sci.}, 140(3):426--444, January 2007.
\newblock \href {https://doi.org/10.1007/s10958-007-0450-5}
  {\path{doi:10.1007/s10958-007-0450-5}}.

\bibitem{litvinov2009tropical}
Grigori{\u\i}~L. Litvinov and Sergej~N. Sergeev.
\newblock {\em Tropical and Idempotent Mathematics: International Workshop
  TROPICAL-07, Tropical and Idempotent Mathematics, August 25--30, 2007,
  Independent University of Moscow and Laboratory J.-V. Poncelet}, volume 495
  of {\em Contemporary Mathematics}.
\newblock {AMS}, 2009.

\bibitem{Liu:17}
Sanhong Liu, Shigui Ruan, and Xinan Zhang.
\newblock Nonlinear dynamics of avian influenza epidemic models.
\newblock {\em Math. Biosci.}, 283:118--135, January 2017.
\newblock \href {https://doi.org/10.1016/j.mbs.2016.11.014}
  {\path{doi:10.1016/j.mbs.2016.11.014}}.

\bibitem{Lueders:20a}
Christoph L{\"u}ders.
\newblock Computing tropical prevarieties with satisfiability modulo theories
  {(SMT)} solvers.
\newblock In P.~Fontaine, K.~Korovin, I.~S. Kotsireas, P.~R{\"{u}}mmer, and
  S.~Tourret, editors, {\em Proc. {SC-Square} 2020, Co-Located With {IJCAR}
  2020}, volume 2752 of {\em {CEUR} Workshop Proceedings}, pages 189--203.
  CEUR-WS, June--July 2020.
\newblock URL: \url{http://ceur-ws.org/Vol-2752/paper14.pdf}.

\bibitem{sympy:17}
Aaron Meurer, Christopher~P. Smith, Mateusz Paprocki, Ond{\v r}ej {\v
  C}ert{\'i}k, Sergey~B. Kirpichev, Matthew Rocklin, Amit Kumar, Sergiu Ivanov,
  Jason~K. Moore, Sartaj Singh, Thilina Rathnayake, Sean Vig, Brian~E. Granger,
  Richard~P. Muller, Francesco Bonazzi, Harsh Gupta, Shivam Vats, Fredrik
  Johansson, Fabian Pedregosa, Matthew~J. Curry, Andy~R. Terrel, {\v S}t{\v
  e}p{\'a}n Rou{\v c}ka, Ashutosh Saboo, Isuru Fernando, Sumith Kulal, Robert
  Cimrman, and Anthony Scopatz.
\newblock {SymPy}: Symbolic computing in {Python}.
\newblock {\em PeerJ Computer Science}, 3:e103, January 2017.
\newblock \href {https://doi.org/10.7717/peerj-cs.103}
  {\path{doi:10.7717/peerj-cs.103}}.

\bibitem{mikhalkin2005enumerative}
Grigory Mikhalkin.
\newblock Enumerative tropical algebraic geometry in $\mathbb{R}^2$.
\newblock {\em J. Amer. Math. Soc.}, 18(2):313--377, January 2005.
\newblock \href {https://doi.org/10.1090/S0894-0347-05-00477-7}
  {\path{doi:10.1090/S0894-0347-05-00477-7}}.

\bibitem{z3}
Leonardo {\noop{Moura}{de Moura}} and Nikolaj Bj{\o}rner.
\newblock {Z}3: An efficient {SMT} solver.
\newblock In C.~R. Ramakrishnan and J.~Rehof, editors, {\em Proc. TACAS 2008},
  volume 4963 of {\em LNCS}, pages 337--340. Springer, 2008.
\newblock \href {https://doi.org/10.1007/978-3-540-78800-3_24}
  {\path{doi:10.1007/978-3-540-78800-3_24}}.

\bibitem{NieuwenhuisOliveras:06b}
Robert Nieuwenhuis, Albert Oliveras, and Cesare Tinelli.
\newblock Solving {SAT} and {SAT} modulo theories: From an abstract
  {D}avis--{P}utnam--{L}ogemann--{L}oveland procedure to {DPLL(T)}.
\newblock {\em J. {ACM}}, 53(6):937--977, November 2006.
\newblock \href {https://doi.org/10.1145/1217856.1217859}
  {\path{doi:10.1145/1217856.1217859}}.

\bibitem{nipp1988algorithmic}
Kaspar Nipp.
\newblock An algorithmic approach for solving singularly perturbed initial
  value problems.
\newblock In U.~Kirchgraber and H.~Walther, editors, {\em Dynamics Reported},
  volume~1, pages 173--263. John Wiley \& Sons and B. G. Teubner, 1988.
\newblock \href {https://doi.org/10.1007/978-3-322-96656-8_4}
  {\path{doi:10.1007/978-3-322-96656-8_4}}.

\bibitem{noel2012tropical}
Vincent Noel, Dima Grigoriev, Sergei Vakulenko, and Ovidiu Radulescu.
\newblock Tropical geometries and dynamics of biochemical networks application
  to hybrid cell cycle models.
\newblock In J.~Feret and A.~Levchenko, editors, {\em Proc. {SASB} 2011},
  volume 284 of {\em {ENTCS}}, pages 75--91. Elsevier, 2012.
\newblock \href {https://doi.org/10.1016/j.entcs.2012.05.016}
  {\path{doi:10.1016/j.entcs.2012.05.016}}.

\bibitem{noelgvr}
Vincent Noel, Dima Grigoriev, Sergei Vakulenko, and Ovidiu Radulescu.
\newblock Tropicalization and tropical equilibration of chemical reactions.
\newblock In G.~L. Litvinov and S.~N. Sergeev, editors, {\em Tropical and
  Idempotent Mathematics and Applications}, volume 616 of {\em Contemporary
  Mathematics}, pages 261--277. {AMS}, 2014.
\newblock \href {https://doi.org/10.1090/conm/616/12316}
  {\path{doi:10.1090/conm/616/12316}}.

\bibitem{NoethenWalcher:09}
Lena Noethen and Sebastian Walcher.
\newblock Quasi-steady state and nearly invariant sets.
\newblock {\em SIAM J. Appl. Math.}, 70(4):1341--1363, January 2009.
\newblock \href {https://doi.org/10.1137/090758180}
  {\path{doi:10.1137/090758180}}.

\bibitem{le2006biomodels}
Nicolas {\noop{Nov{\`e}re}{Le Nov{\`e}re}}, Benjamin Bornstein, Alexander
  Broicher, M{\'e}lanie Courtot, Marco Donizelli, Harish Dharuri, Lu~Li,
  Herbert Sauro, Maria Schilstra, Bruce Shapiro, Jacky~L. Snoep, and Michael
  Hucka.
\newblock {BioModels} database: A free, centralized database of curated,
  published, quantitative kinetic models of biochemical and cellular systems.
\newblock {\em Nucleic acids res.}, 34(suppl\_1):D689--D691, January 2006.
\newblock \href {https://doi.org/10.1093/nar/gkj092}
  {\path{doi:10.1093/nar/gkj092}}.

\bibitem{rgzn}
Ovidiu Radulescu, Alexander~N. Gorban, Andrei Zinovyev, and Vincent Noel.
\newblock Reduction of dynamical biochemical reactions networks in
  computational biology.
\newblock {\em Front. Genet.}, 3:131, July 2012.
\newblock \href {https://doi.org/10.3389/fgene.2012.00131}
  {\path{doi:10.3389/fgene.2012.00131}}.

\bibitem{rvg}
Ovidiu Radulescu, Sergei Vakulenko, and Dima Grigoriev.
\newblock Model reduction of biochemical reactions networks by tropical
  analysis methods.
\newblock {\em Math. Model. Nat. Pheno.}, 10(3):124--138, June 2015.
\newblock \href {https://doi.org/10.1051/mmnp/201510310}
  {\path{doi:10.1051/mmnp/201510310}}.

\bibitem{Reddyhoff:15}
Dennis Reddyhoff, John Ward, Dominic Williams, Sophie Regan, and Steven Webb.
\newblock Timescale analysis of a mathematical model of acetaminophen
  metabolism and toxicity.
\newblock {\em J. Theor. Biol}, 386:132--146, December 2015.
\newblock \href {https://doi.org/10.1016/j.jtbi.2015.08.021}
  {\path{doi:10.1016/j.jtbi.2015.08.021}}.

\bibitem{Ruan:17}
Shigui Ruan.
\newblock Modeling the transmission dynamics and control of rabies in {China}.
\newblock {\em Math. Biosci.}, 286:65--93, April 2017.
\newblock \href {https://doi.org/10.1016/j.mbs.2017.02.005}
  {\path{doi:10.1016/j.mbs.2017.02.005}}.

\bibitem{sgfr}
Satya~S. Samal, Dima Grigoriev, Holger Fr{\"o}hlich, and Ovidiu Radulescu.
\newblock Analysis of reaction network systems using tropical geometry.
\newblock In V.~Gerdt, W.~Koepf, W.~Seiler, and E.~Vorozhtsov, editors, {\em
  Proc. {CASC} 2015}, volume 9301 of {\em LNCS}, pages 424--439. Springer,
  2015.
\newblock \href {https://doi.org/10.1007/978-3-319-24021-3_31}
  {\path{doi:10.1007/978-3-319-24021-3_31}}.

\bibitem{sgfwr}
Satya~S. Samal, Dima Grigoriev, Holger Fr{\"o}hlich, Andreas Weber, and Ovidiu
  Radulescu.
\newblock A geometric method for model reduction of biochemical networks with
  polynomial rate functions.
\newblock {\em B. Math. Biol.}, 77(12):2180--2211, October 2015.
\newblock \href {https://doi.org/10.1007/s11538-015-0118-0}
  {\path{doi:10.1007/s11538-015-0118-0}}.

\bibitem{SchneiderWilhelm:00a}
Klaus~R. Schneider and Thomas Wilhelm.
\newblock Model reduction by extended quasi-steady-state approximation.
\newblock {\em J. Math. Biol.}, 40(5):443--450, May 2000.
\newblock \href {https://doi.org/10.1007/s002850000026}
  {\path{doi:10.1007/s002850000026}}.

\bibitem{Schrijver86}
Alexander Schrijver.
\newblock {\em Theory of Linear and Integer Programming}.
\newblock John Wiley \& Sons, Chichester, 1998.
\newblock \href {https://doi.org/10.1002/oca.4660100108}
  {\path{doi:10.1002/oca.4660100108}}.

\bibitem{segel1989quasi}
Lee~A. Segel and Marshall Slemrod.
\newblock The quasi-steady-state assumption: A case study in perturbation.
\newblock {\em {SIAM} Rev.}, 31(3):446--477, September 1989.
\newblock \href {https://doi.org/10.1137/1031091} {\path{doi:10.1137/1031091}}.

\bibitem{Seiler2020}
Werner~M. Seiler, Matthias Sei{\ss}, and Thomas Sturm.
\newblock A logic based approach to finding real singularities of implicit
  ordinary differential equations.
\newblock {\em Math. Comput. Sci.}, 2020.
\newblock In press.
\newblock \href {https://doi.org/10.1007/s11786-020-00485-x}
  {\path{doi:10.1007/s11786-020-00485-x}}.

\bibitem{sommer2016hybrid}
Jasha Sommer-Simpson, John Reinitz, Leonid Fridlyand, Louis Philipson, and
  Ovidiu Radulescu.
\newblock Hybrid reductions of computational models of ion channels coupled to
  cellular biochemistry.
\newblock In E.~Bartocci, P.~Lio, and N.~Paoletti, editors, {\em Proc. {CMSB}
  2016}, volume 9859 of {\em LNCS}, page 273. Springer, 2016.
\newblock \href {https://doi.org/10.1007/978-3-319-45177-0_17}
  {\path{doi:10.1007/978-3-319-45177-0_17}}.

\bibitem{Sturm:17a}
Thomas Sturm.
\newblock A survey of some methods for real quantifier elimination, decision,
  and satisfiability and their applications.
\newblock {\em Math. Comput. Sci.}, 11(3--4):483--502, December 2017.
\newblock \href {https://doi.org/10.1007/s11786-017-0319-z}
  {\path{doi:10.1007/s11786-017-0319-z}}.

\bibitem{Sturm:ISSAC2018}
Thomas Sturm.
\newblock Thirty years of virtual substitution: Foundations, techniques,
  applications.
\newblock In C.~Arreche, editor, {\em Proc. ISSAC 2018}, pages 11--16. ACM,
  July 2018.
\newblock \href {https://doi.org/10.1145/3208976.3209030}
  {\path{doi:10.1145/3208976.3209030}}.

\bibitem{SturmWeber:08a}
Thomas Sturm and Andreas Weber.
\newblock Investigating generic methods to solve {H}opf bifurcation problems in
  algebraic biology.
\newblock In K.~Horimoto, editor, {\em Proc. Algebraic Biology 2008}, volume
  5147 of {\em LNCS}, pages 200--215. Springer, 2008.
\newblock \href {https://doi.org/10.1007/978-3-540-85101-1_15}
  {\path{doi:10.1007/978-3-540-85101-1_15}}.

\bibitem{SturmWeber:09a}
Thomas Sturm, Andreas Weber, Essam~O. Abdel-Rahman, and M'hammed {El Kahoui}.
\newblock Investigating algebraic and logical algorithms to solve {H}opf
  bifurcation problems in algebraic biology.
\newblock {\em Math. Comput. Sci.}, 2(3):493--515, March 2009.
\newblock \href {https://doi.org/10.1007/s11786-008-0067-1}
  {\path{doi:10.1007/s11786-008-0067-1}}.

\bibitem{Tarski:48a}
Alfred Tarski.
\newblock A decision method for elementary algebra and geometry. {P}repared for
  publication by {J}.~{C}.~{C}.~{M}c{K}insey.
\newblock RAND Report R109, August 1, 1948, Revised May 1951, Second Edition,
  RAND, Santa Monica, CA, 1957.

\bibitem{tikh}
Andrei~Nikolaevich Tikhonov.
\newblock Systems of differential equations containing small parameters in the
  derivatives.
\newblock {\em Mat. Sb. (N. S.)}, 73(3):575--586, 1952.

\bibitem{valorani2009g}
Mauro Valorani and Samuel Paolucci.
\newblock The {G}-scheme: A framework for multi-scale adaptive model reduction.
\newblock {\em J. Comput. Phys.}, 228(13):4665--4701, July 2009.
\newblock \href {https://doi.org/10.1016/j.jcp.2009.03.011}
  {\path{doi:10.1016/j.jcp.2009.03.011}}.

\bibitem{VilarJansen:06}
Jose M.~G. Vilar, Ronald Jansen, and Chris Sander.
\newblock Signal processing in the {TGF}-$\beta$ superfamily ligand-receptor
  network.
\newblock {\em PLoS Comput. Biol.}, 2(1):e3, January 2006.
\newblock \href {https://doi.org/10.1371/journal.pcbi.0020003}
  {\path{doi:10.1371/journal.pcbi.0020003}}.

\bibitem{viro2001dequantization}
Oleg Viro.
\newblock Dequantization of real algebraic geometry on logarithmic paper.
\newblock In C.~Casacuberta, R.~M. Mir{\'o}-Roig, J.~Verdera, and
  S.~Xamb{\'o}-Descamps, editors, {\em European Congress of Mathematics},
  volume 201 of {\em Progress in Mathematics}, pages 135--146. Springer, 2001.
\newblock \href {https://doi.org/10.1007/978-3-0348-8268-2_8}
  {\path{doi:10.1007/978-3-0348-8268-2_8}}.

\bibitem{WeberSturm:11a}
Andreas Weber, Thomas Sturm, and Essam~O. Abdel-Rahman.
\newblock Algorithmic global criteria for excluding oscillations.
\newblock {\em B. Math. Biol.}, 73(4):899--916, April 2011.
\newblock \href {https://doi.org/10.1007/s11538-010-9618-0}
  {\path{doi:10.1007/s11538-010-9618-0}}.

\bibitem{Weispfenning:88}
Volker Weispfenning.
\newblock The complexity of linear problems in fields.
\newblock {\em J. Symb. Comput.}, 5(1--2):3--27, February--April 1988.
\newblock \href {https://doi.org/10.1016/S0747-7171(88)80003-8}
  {\path{doi:10.1016/S0747-7171(88)80003-8}}.

\bibitem{Weispfenning:97b}
Volker Weispfenning.
\newblock Quantifier elimination for real algebra---the quadratic case and
  beyond.
\newblock {\em Appl. Algebr. Eng. Comm.}, 8(2):85--101, February 1997.
\newblock \href {https://doi.org/10.1007/s002000050055}
  {\path{doi:10.1007/s002000050055}}.

\bibitem{Wodarz:07}
Dominik Wodarz and Dean~H. Hamer.
\newblock Infection dynamics in {HIV}-specific {CD4} {T} cells: Does a {CD4}
  {T} cell boost benefit the host or the virus?
\newblock {\em Math. Biosci.}, 209(1):14--29, September 2007.
\newblock \href {https://doi.org/10.1016/j.mbs.2007.01.007}
  {\path{doi:10.1016/j.mbs.2007.01.007}}.
\end{thebibliography}
\end{document}